%% file: main_arXiv_2025.tex
\documentclass[ijds,nonblindrev]{informs4_arxiv}

\OneAndAHalfSpacedXI

\usepackage{amsmath,amssymb,amsfonts}
\usepackage{array}

\usepackage{natbib}
 \bibpunct[, ]{(}{)}{,}{a}{}{,}%
 \def\bibfont{\small}%

\usepackage{rotating}
\usepackage{fancyvrb}
\usepackage[capitalize]{cleveref}
\usepackage{subcaption}
\usepackage{xspace}
\usepackage{enumitem}

\input{local_def}

\TheoremsNumberedThrough     % Preferred (Theorem 1, Lemma 1, Theorem 2)
\ECRepeatTheorems

\EquationsNumberedThrough    % Default: (1), (2), ...
\usepackage[suppress]{color-edits}
\addauthor[Yichi]{yichi}{red}
\addauthor[Yichi]{yz}{magenta}
\addauthor[David]{dk}{blue}
\addauthor[Fang-Yi]{fang}{green}
\addauthor[Grant]{gs}{cyan}

\RUNTITLE{A System-level Analysis of Conference Peer Review}

\TITLE{A System-level Analysis of Conference Peer Review}

\ARTICLEAUTHORS{%
\AUTHOR{Yichi Zhang}
\AFF{Rutgers University, \EMAIL{yz1636@dimacs.rutgers.edu}, 
supported in part by NSF CCF \#2007256}

\AUTHOR{Fang-Yi Yu}
\AFF{George Mason University, \EMAIL{fangyiyu@gmu.edu}, 
supported in part by NSF IIS \#2007887}

\AUTHOR{Grant Schoenebeck}
\AFF{University of Michigan, \EMAIL{schoeneb@umich.edu}, 
supported in part by NSF CCF \#2007256}

\AUTHOR{David Kempe}
\AFF{University of Southern California, \EMAIL{david.m.kempe@gmail.com}, 
supported in part by ARO MURI W911NF1810208}
}
\RUNAUTHOR{Zhang, Yu, Schoenebeck, and Kempe}

\ABSTRACT{%
The conference peer review process involves three constituencies: authors desire (fast) acceptance in prestigious venues, conferences aim to accept more high-quality and fewer low-quality papers, and reviewers should be spared overwhelming workloads. These objectives often clash, in large part due to the noise inherent in peer review. In response, conferences have experimented with various policies, including adjusting acceptance thresholds and the number of reviews per submission. We investigate how well various policies work, and more importantly, why they do or do not work.
We model the conference-author interactions as a Stackelberg game: a prestigious conference commits to an acceptance policy, and authors respond by (re)submitting or not.
Using our model, we explain why many submissions are rejected every year (in practice, rejection rates reach approximately 75\%), despite many of them eventually being accepted at prestigious venues with minor improvements, a phenomenon we term the ``Resubmission Paradox''. \yzreplace{We show that reviewer inaccuracy and author patience encourage authors to (re)submit and thus increase the review workload for a fixed conference quality. }{We show that higher review quality benefits all three stakeholders: the conference can maintain the same quality with lower review burden and greater author welfare. Moreover, greater author patience leads to more persistent (re)submissions, thereby increasing the review workload for a fixed conference quality.}
For the robustness of these and other results, we consider variants of our model and conduct simulations based on plausible parameters estimated from real data.
}

\begin{document}
\maketitle

\section{Introduction}
\input{OR-resbumission/sections/Introduction}

\section{Model}\label{sec:model}
\input{OR-resbumission/sections/Model}

\section{Noiseless Authors: Thresholds and Resubmission Gaps} \label{sec:thresholds-gaps}
\input{OR-resbumission/sections/Best-response}

\section{Noiseless Authors: Tradeoffs and Acceptance Rate} \label{sec:continuous-tradeoffs}

\input{OR-resbumission/sections/Tradeoffs}

\section{Discussion and Conclusion}
\input{OR-resbumission/sections/Conclusion}

\bibliographystyle{abbrvnat}
\bibliography{names,conferences,reference}

\newpage

\ECSwitch

\section{Additional Proofs}
\label{app:proofs}
\input{OR-resbumission/sections/Proofs}

\section{When Are Threshold Acceptance Policies Optimal?}
\label{sec:testing_tradeoff}
\input{OR-resbumission/sections/When_optimal}

\section{Authors with Noisy Signals: ABM Experiments in Real-Data Estimated Models}
\label{sec:noisy_abm}
\input{OR-resbumission/sections/Noisy_ABM}
\dkcomment{Moving this up.}
\input{OR-resbumission/sections/Additional_appendix}

% \vspace{2mm}
% \section{Noiseless Authors and Real-World Parameters}\label{sec:noiseless-ICLR}
% \vspace{2mm}
% \input{OR-resbumission/sections/Categorical}

% \vspace{2mm}
% \section{Authors with Noisy Signals: ABM Experiments}\label{sec:noisy}
% \vspace{2mm}
% \input{OR-resbumission/sections/Noisy}

\section{Institutional Memory}
\label{sec:memory_policy}
\input{OR-resbumission/sections/Memory}

\end{document}

%% file: local_def.tex
\DeclareMathOperator{\E}{\mathbb{E}}

\newcommand{\QualSetSize}{L}

\newcommand{\AUTHUTIL}{U^{(a)}}

\usepackage{accents}
\newcommand{\ubar}[1]{\underaccent{\bar}{#1}}

\providecommand{\PROB}{\textrm {Prob}}
\providecommand{\Prob}[2][]{%
\ifthenelse{\equal{#1}{}}{\PROB[#2]}{\PROB_{#1}[#2]}}
\providecommand{\ProbC}[3][]{\Prob[#1]{#2\;|\;#3}}

\providecommand{\half}{\ensuremath{\frac{1}{2}}\xspace}

\providecommand{\SET}[1]{\ensuremath{\{ #1 \}}\xspace}
\providecommand{\Set}[2]{\ensuremath{\SET{#1 \mid #2}}\xspace}

\providecommand{\Expect}[2][]{\ensuremath{%
\ifthenelse{\equal{#1}{}}{\mathbb{E}}{\mathbb{E}_{#1}}%
\left[#2\right]}\xspace}
\providecommand{\ExpectC}[3][]{\Expect[#1]{#2\;|\;#3}}
\DeclareSymbolFont{AMSb}{U}{msb}{m}{n}
\DeclareMathSymbol{\N}{\mathord}{AMSb}{"4E}
\DeclareMathSymbol{\R}{\mathord}{AMSb}{"52}
\DeclareMathSymbol{\Z}{\mathord}{AMSb}{"5A}

\newcommand{\Qual}{\ensuremath{Q}\xspace}
\newcommand{\qual}{\ensuremath{q}\xspace}
\newcommand{\QualSet}[1][]{\ensuremath{%
\ifthenelse{\equal{#1}{}}{\mathcal{Q}}{\mathcal{Q}^{#1}}}\xspace}
\newcommand{\QualDist}{\ensuremath{\BFp}\xspace}
\newcommand{\QualDens}[1]{\ensuremath{p(#1)}\xspace}
\newcommand{\QualProb}[1]{\ensuremath{p_{#1}}\xspace}

\newcommand{\SigSet}{\ensuremath{\Sigma}\xspace}

\newcommand{\AuthSig}{S^{(a)}}
\newcommand{\REVSIG}{s}
\newcommand{\REVSIGP}{s'}
\newcommand{\RevSig}[2][]{\ensuremath{%
\ifthenelse{\equal{#1}{}}{s_{#2}}{s_{#1,#2}}}\xspace}
\newcommand{\RevSigP}[2][]{%
\ifthenelse{\equal{#1}{}}{s'_{#2}}{s'_{#1,#2}}}
\newcommand{\RevSigV}[1][]{\ensuremath{%
\ifthenelse{\equal{#1}{}}{\BFs}{\BFs_{#1}}}\xspace}
\newcommand{\RevSigVP}[1][]{\ensuremath{%
\ifthenelse{\equal{#1}{}}{\BFs'}{\BFs'_{#1}}}\xspace}
\newcommand{\RevSigVH}[1][]{\ensuremath{%
\ifthenelse{\equal{#1}{}}{\hat{\BFs}}{\hat{\BFs}_{#1}}}\xspace}
\newcommand{\AuthSigDist}[1][]{\BFalpha_{#1}}

\newcommand{\AUTHSIGPROB}{\ensuremath{\alpha}\xspace}

\newcommand{\RevSigDist}[1][]{\ensuremath{\BFbeta_{#1}}\xspace}
\newcommand{\RevSigDistP}[1][]{\ensuremath{\BFbeta'_{#1}}\xspace}

\newcommand{\REVSIGPROB}{\ensuremath{\beta}\xspace}
\newcommand{\RevSigProb}[2][]{\ensuremath{\ifthenelse{\equal{#1}{}}{\REVSIGPROB(#2)}{\REVSIGPROB_{#1}(#2)}}\xspace}
\newcommand{\REVSIGPROBP}{\beta'}
\newcommand{\RevSigProbP}[2][]{\ifthenelse{\equal{#1}{}}{\REVSIGPROBP(#2)}{\REVSIGPROBP_{#1}(#2)}}

\newcommand{\CONFUTIL}{U^{(c)}}
\newcommand{\CONFUTILH}{\hat{U}^{(c)}}

\newcommand{\QualDistTilde}{\tilde{\BFp}}
\newcommand{\RevSigDistTilde}{\tilde{\BFbeta}}

\newcommand{\REVNOISEDIST}[1][]{\ifthenelse{\equal{#1}{}}{F^{(r)}}{F^{(r)}_{#1}}}
\newcommand{\RevNoiseDist}[2][]{\REVNOISEDIST[#1]\left(#2\right)}

\newcommand{\NumReviews}[1][]{\ensuremath{%
\ifthenelse{\equal{#1}{}}{m}{m_{#1}}}\xspace}
\newcommand{\SigCount}[2][]{%
\ifthenelse{\equal{#1}{}}{c_{#2}}{c_{#1,#2}}}
\newcommand{\SigCountV}[1][]{%
\ifthenelse{\equal{#1}{}}{\BFc}{\BFc_{#1}}}
\newcommand{\SigCountSpace}[1][]{%
\ifthenelse{\equal{#1}{}}{\mathcal{C}}{\mathcal{C}_{#1}}}
\newcommand{\ACCMAP}[1][]{\ensuremath{%
\ifthenelse{\equal{#1}{}}{\phi}{\ifthenelse{\equal{#1}{'}}{\phi'}{\phi_{#1}}}}\xspace}
\newcommand{\AccMap}[2][]{\ensuremath{\ACCMAP[#1](#2)}\xspace}
\newcommand{\THRESH}[1][]{%
\ifthenelse{\equal{#1}{}}{\theta}{\theta_{#1}}}
\newcommand{\ACCP}{P_{\text{acc}}}
\newcommand{\AccP}[2]{\ensuremath{\ACCP(#1,#2)}\xspace}
\newcommand{\PaperReviews}{\ensuremath{R}\xspace}
\newcommand{\PaperReviewsH}{\ensuremath{\hat{R}}\xspace}
\newcommand{\NumSubmitDist}{\ensuremath{\BFb}\xspace}
\newcommand{\NumSubmitProb}[1]{\ensuremath{b_{#1}}\xspace}
\newcommand{\NumNewPapers}{\ensuremath{K}\xspace}

\newcommand{\TD}{\ensuremath{\eta}\xspace}
\newcommand{\ConfValue}{\ensuremath{V}\xspace}
\newcommand{\AccRate}{\ensuremath{\gamma}\xspace}

\newcommand{\Gaussian}[2]{\mathcal{N}(#1,#2)}

%% file: OR-resbumission/sections/Introduction.tex
Conferences play an important role in the publication and {scientific dissemination process in computer science}. They aim to provide attendees with access to high-quality and recent research results (in addition to networking opportunities), and authors view the ability to publish and share their recent results at high-quality venues as conferring scientific credibility and thus status. In order to do so, conferences rely on significant volunteer work from the community, most notably in reviewing large numbers of submitted papers to evaluate their scientific merit.
% \gscomment{Not sure this is the best place for it, but I feel like we should comment on the fact that all three of these "parties" are typically the same researches just in different roles.  Thus, what is good for the "conference" is also good for a researcher as it will raise the prestige of his field to have a high quality conference.}\yichicomment{Added a footnote}

While the conference publication process seems to have served the community fairly well overall --- in particular enabling a high speed of dissemination of scientific results --- different members of the community often see significant room for improvement: authors often feel that their submitted papers are not evaluated sufficiently competently, while conference attendees sometimes find the program diluted with less interesting work.
% \footnote{\yichiedit{or there are too many papers to find the really interesting ones}}\dkcomment{Not sure about this footnote. That's the meaning of diluted, that the stuff you're interested in gets drowned out.}
Attempts have been made by conferences to mitigate such concerns, one of which is to increase the number of reviews assigned to each submission.
For example, nowadays, it is no longer rare for top computer science conferences to assign five or more reviewers to one paper.
However, this approach increases the review burden; indeed, many members of the scientific community now feel overloaded with conference reviewing requests. The approach has the potential to lead to a vicious cycle: the substantial increase in peer review workload, which is further exacerbated by the growth of the research community, may lead to lower-quality reviews, in turn leading to more resubmissions of the same papers and thus a higher review burden.

Our high-level goal is to understand how the parameters of the system, together with the conference's review policy decisions, affect the tradeoff between the conference quality, the review burden on the community, and the welfare of the authors.
In particular, the tradeoffs between conference quality and each of the other two dimensions are particularly interesting.%
\footnote{We note that the three parties in the peer review system are usually the same researchers serving in different roles. Consequently, the tradeoffs among their utilities can be viewed as tradeoffs among different dimensions of the research field.}
Given a fixed number and quality distribution of submissions, the conference's quality is increased by accepting more good and fewer bad papers.
Distinguishing between good and bad papers more accurately requires more reviews, along with enticing self-selection by authors.

We model the conference review process as a Stackelberg game between a top conference and the (homogeneous and strategic) authors (described in detail in \cref{sec:model}).
The authors' papers have qualities (positive or negative) drawn from a commonly known distribution over a (finite or infinite) set; however, both the conference's reviewers and (in some of our analysis) the authors themselves only obtain imperfect signals about the quality of a paper.
The conference's quality is the sum of qualities of the accepted papers,  normalized by the total number of papers.
The conference commits to a review and acceptance policy, which prescribes how many independent reviewers are assigned to each paper, and what the criteria for acceptance are.
In \cref{sec:memory_policy}, we also study policies with ``institutional memory,'' including policies limiting the number of times a paper can be resubmitted, and policies requiring past reviews to be included with resubmissions.

In response to the conference's policies, the authors face a binary decision in each round: submit to the prestigious conference or opt for a safer outside option. Submission to the top conference offers the potential for high utility upon acceptance (proportional to the average quality of accepted papers), but yields zero utility if the submission is rejected.
In contrast, the outside option --- such as a second-tier conference or arXiv --- provides guaranteed but smaller positive utility. 
The utility of acceptance to either venue is exponentially time-discounted in the number of resubmission rounds, modeling that authors prefer timely acceptance of their result. 
Authors best-respond, i.e., choose a utility-maximizing action, based on their private (potentially noisy) signal about their paper's quality as well as the historical reviews collected when the paper was rejected in previous rounds. 
We primarily focus on the following high-level questions: 

\begin{itemize}
    \item The fact that rejected papers can be resubmitted means that weaker papers may be eventually accepted if authors are patient enough. What is the impact of resubmissions on the gap between a conference's intended acceptance threshold and the de facto distribution of paper qualities at the conference?
    \item What is the range of Pareto optimal review and acceptance policies with regard to the tradeoff between conference quality and review burden?  How is it impacted by authors' patience, the number of reviews per paper, and/or review quality?
    \item How does the conference's acceptance policy relate to its acceptance rate? This is non-obvious, as a strict policy may lead to a lot of self-selection on the part of the authors, and thus to a \emph{higher} acceptance rate.
    \item How does the number of reviews per round affect the overall review burden? What is the tradeoff between review quality and quantity, i.e., how many noisy reviews approximate one high-quality review?
    \item How helpful is institutional memory? Can the conference significantly lower the review burden or increase quality by limiting the number of resubmissions, treating resubmitted papers differently, or requiring past reviews to be included?
\end{itemize}

To ensure robustness of our results, we investigate these questions under different models regarding the quality distributions of papers and distributions of review noise, and the information authors have about the quality of their own work.
In the main body of the paper, we focus on settings with noiseless authors who possess perfect information about their own paper quality. 
We consider both a continuous model where the paper quality follows a continuous distribution and reviews are continuous with additive noise, and a categorical model with several discrete paper qualities and a discrete scale of review scores.
Authors do not learn new information from reviews, so each author's decision is either to submit her paper until it is accepted, or to immediately take the outside option. The simplicity of this model, due to the authors' binary best response, allows us to prove theoretical results and draw intuitive conclusions.

In the appendix, we consider agent-based modeling under a categorical model, where authors only observe noisy signals about their papers' qualities. This model allows for more realistic modeling of real-world conferences, and we estimate the parameters from publicly available review data for the ICLR 2020 conference \citep{iclr2020review}. Our agent-based simulations, built on both this detailed categorical model and a simplified binary variant, serve two purposes. First, they provide robustness checks for the theoretical insights derived from our continuous model. Second, they create a flexible platform for investigating sophisticated acceptance policies that incorporate institutional memory and other complex features difficult to analyze theoretically.

While the models are necessarily a simplification, and one should therefore be cautious of directly basing concrete decisions on the results, we believe that the fundamental insights derived from our analysis (summarized below) remain robust across a reasonable range of variants of our model. We thus envision that our theoretical results (and simulations) can steer the discussion, uncover parameters to focus on, and inform decision makers in practice.

\subsection{Summary of Results}

% \gsreplace{Due to space limitation, we leave a substantial fraction of our results in the appendix. In the main body, we thus focus on the setting of noiseless authors where our theoretical results can provide useful intuitions on how to interpret our contributions in more general settings.}
% \gscomment{do we really focus on the continuous model?  I would propose removing this.}
In the main body, we focus on the setting with noiseless authors, where we can obtain strong theoretical results.  The appendix contains proofs, some further details, and additional results, largely using simulations on models estimated by real data, that both illustrate the robustness of our theoretical results and explore more complex settings.

\vspace{2mm}
\emph{Noiseless Authors: Theoretical Results.}
We first assume that authors know the quality of their manuscripts, but reviewers only obtain noisy signals of the quality. 

We first explain a counter-intuitive phenomenon which we call the \emph{resubmission paradox}: real-world conferences typically have an acceptance rate of around 25\%; yet, many rejected papers are accepted later at other (or even the same) prestigious venues \citep{cortes2021inconsistency}.
Therefore, it is sometimes argued that conferences should accept every paper that is ``above the bar'', by lowering the acceptance threshold. Our model can help explain this counterintuitive phenomenon. Because every submitted paper will be resubmitted until accepted (since the authors learn no new information from the reviews), the conference's \emph{acceptance threshold} induces a \emph{de facto threshold}: a manuscript quality above which every paper will be eventually accepted and below which no manuscript will be submitted. The acceptance threshold and de facto threshold are typically different, sometimes significantly so (in particular when the authors are very patient, the noise of reviews is large, or the conference is very selective, resulting in a large utility upon acceptance) --- we call the difference the \emph{resubmission gap}. 
In most situations, if a conference were to naively set an acceptance threshold identical to their ideal de facto threshold, it would ultimately accept substandard papers, in rounds in which enough reviewers had noisy reviews that ended up too high. We call this phenomenon the \emph{resubmission paradox.} Instead, the conference should select an acceptance threshold higher than the desired threshold of conference acceptances.  However, the resulting ``optimal'' is then counter-intuitive: every paper that is submitted is eventually accepted, yet each round, many of the papers are rejected.
% We exactly characterize the resubmission gap in several settings.

Second, we consider the relationship between the conference quality, the review burden (the total number of reviews per paper throughout its resubmission process), and the authors' welfare. 
Both review burden and author welfare favor inclusive, lenient policies with high acceptance rates, since resubmissions increase review load and reduce authors' utility through discounting. In contrast, maximizing conference quality requires setting an appropriate threshold to ensure that only positive-quality papers are ultimately accepted. This tension gives rise to a \emph{quality-burden tradeoff (QB-tradeoff)}. We show that whether the Pareto frontier of this tradeoff is achieved through lenient or stringent acceptance thresholds depends on the distribution of paper quality and the level of review noise. Stringent policies tend to dominate when low-quality papers are prevalent and reviews are less noisy. % \gscomment{Perhaps we could rewrite this in light of the new organization.  AS is, I think this is one of our least crisp points and in this summary we should shorten it as much as possible.  Currently, the reader is going to get lost and not see the other amazing stuff we ahve done.}

% At a high level, the review burden closely tracks the number of papers just above the de facto threshold, because such papers typically require numerous resubmissions before acceptance.   
% In particular, a conference may be able to significantly reduce the long-term review burden by choosing the threshold such that there are fewer borderline papers near the resulting de facto threshold.

% \dkreplace{Third, we show that if the prestige of the conference increases, the patience of authors increases, or the noise of the reviews increases (in a certain technical sense), then the tradeoff between the  conference's quality and the review burden becomes strictly worse. Thus, one cause of more reviews may be the high prestige placed on certain venues.  This also warns that policies which effect these items may have unintended consequences.}
Third, we explore how the system parameters affect the utilities of the three stakeholders in the review system. We first show that a better review quality --- in the sense of Blackwell-dominance --- can benefit all parties:  the conference can attain higher quality with a lower review burden and increased author welfare. This result holds broadly when the system with better review quality is allowed to adopt any acceptance policy, including non-monotone ones that do not necessarily favor papers with more positive reviews. If we restrict attention to threshold acceptance policies --- a more natural and interpretable class --- the result continues to hold under the additional assumption that review signals satisfy a monotone likelihood ratio. %\gscomment{the previous sentences seem too in the weeds.  I think we shoud rephrase more succinctly.}
We also analyze the impact of author patience. We show that, conditioned on the same conference quality, systems with more patient authors require a higher review burden but do not necessarily yield higher author welfare. This highlights a potential risk: policies that reduce the cost of resubmissions may unintentionally reduce the utility of at least one stakeholder, including the authors themselves.

% show that if the prestige of the conference increases or the patience of authors increases, then the tradeoff between the  conference's quality and the review burden becomes strictly worse. Thus, one cause of more reviews may be the high prestige placed on certain venues.  This also warns that policies which affect these items may have unintended consequences.
% We obtain similar results in terms of the review quality. If the signals of reviewers  in one setting Blackwell-dominate those in another setting, the conference can obtain a weakly better tradeoff between quality and review burden. 
% However, doing so may necessitate using non-monotone acceptance policies, which do not treat papers with strictly more positive reviews better than those with strictly less positive reviews.  We show that a better tradeoff for better signals can be achieved under a restriction to threshold policies as well, but only with the additional assumption that the review signal is \emph{informative} in the sense of monotone likelihood ratios.
 
Fourth, we examine factors influencing the conference acceptance rate. When the review noise is small, the relationship between the acceptance threshold and acceptance rate depends on two key quantities: the acceptance probability of borderline papers and a measure resembling the \emph{hazard rate} of the prior distribution over paper qualities. We show that raising the acceptance threshold lowers the acceptance probability of borderline papers, and thereby reduces the acceptance rate. 
\dkcomment{In the following three sentences, we write ``hazard rate''. Do we mean ``hazard rate'' or our quantity?}
If the quality prior has an increasing hazard rate --- typical for thin-tailed distributions such as the Gaussian --- then a higher threshold increases the hazard rate, which further reduces the acceptance rate. However, for priors with non-monotone hazard rate, the relationship may be non-monotonic, and a stricter threshold can sometimes \emph{increase} the acceptance rate, such as when many borderline papers (which might otherwise be rejected multiple times) cease to be submitted in the first place. Intuitively, for quality distributions with non-monotone hazard rates, raising the acceptance threshold may sharply reduce the proportion of borderline papers, which are major contributors to the review burden through many rounds of resubmissions. We confirm the robustness of these patterns beyond the small-noise setting through simulations.

\vspace{2mm}
\emph{Real-Data-Estimated Categorical Model with Noisy Authors}  
We further examine the robustness of our earlier findings using a model estimated from real data, in a setting where authors are uncertain about their paper's quality and can learn about it through reviewer feedback (\cref{sec:noisy_abm}).
Our results broadly align with the theoretical insights but also reveal new dynamics introduced by the noisy author assumption. In particular, because authors’ self-selection is no longer perfect, the maximum achievable conference quality varies across different settings. We find that increasing the number of reviews per paper can improve the maximum conference quality but at the cost of a higher review burden. In contrast, enhancing review quality raises the maximum achievable quality while requiring a smaller review burden. This result further emphasizes the importance of review quality over quantity.
% In this context, the Pareto frontier for conference quality vs.~review burden is significantly worse than with perfectly appraised authors\textemdash the reason is that the latter can be compelled to self-select with carefully chosen acceptance thresholds.  However, it is not clear if such gains are realizable in practice. 
% Perhaps, new discipline norms of learning the quality of one's paper prior to submission (for example, by sharing early manuscripts with colleagues for feedback) could provide authors with accurate quality signals. 
% Or perhaps, in practice, overcoming authors' unfounded admiration of their own work is not possible.

\vspace{2mm}
\emph{Memory in the System.}
Another popular type of proposal is to give the system more memory, either by limiting the number of resubmissions of the same manuscript or by reusing reviews. 
% \yichicomment{Removed a footnote here.}
% \dkcomment{This seems like a pretty important footnote. I would move it to the main text, in parentheses. Otherwise, reviewers may call this out\textemdash as for instance did the AE at MS.}
(Another important reason for requiring the inclusion of prior reviews --- not modeled here --- is that it lets the conference ascertain that specific concerns from earlier versions have been addressed.)
Our first result shows that limiting the number of resubmissions reduces the overall review burden but also lowers the maximum achievable conference quality. Empirically, we further find that retaining historical reviews across resubmissions allows the conference to maintain (or slightly improve) quality while reducing the review burden. However, this effect can be marginal. Thus, the primary benefit of reusing reviews may be to track how papers are revised --- an objective outside the scope of this study. It should be noted, though, that our analysis here is preliminary and only carried out for the binary model due to computational constraints. 
% \dkcomment{We say the same thing twice in a row.}
% Our results here (in \cref{sec:memory_policy}) show that the main effect of having memory within the review process is that the conference can reduce the review burden while preserving the same (or slightly better) conference quality. However, such an effect can be marginal in some cases; thus, it is not clear whether it is worthwhile to implement these policies broadly. It should be noted, though, that our analysis here is rather preliminary, and only carried out for the binary model.

Our models necessarily abstract away several other aspects of the conference submission ecosystem, which would also be worth investigating.  
% \dkcomment{This is not really true any more, right?} 
These limitations are discussed in more depth in \cref{sec:limitations}.

\subsection{Related Work}
\label{sec:related-work}

\gscomment{Here is another paper that Misha found, we should probably acknowledge him actually, it also has the insite the adding friction often helps things.  https://journals.plos.org/plosone/article?id=10.1371/journal.pone.0246675}

Not surprisingly, given the importance of peer review in science, several attempts have been made by different research fields to simulate, understand, and improve the process.
When considering the systems level, agent-based models (ABMs) have been one of the techniques of choice.
The review article by \citet{feliciani2019scoping} gives a fairly comprehensive summary of this line of work. It suggests several general themes to focus on in models: editorial strategies, matching submissions with reviewers, decision making, biases and calibration, and comparisons of alternative peer review systems. 

Among the more prominent works using ABMs are those by \citet{kovanis2016complex,kovanis2017evaluating}. 
\citet{kovanis2016complex} propose a model for a holistic study of the scientific publication ecosystem\textemdash this model includes the acquisition of resources (such as status) by authors, which can be leveraged into future papers.  Their subsequent work \citep{kovanis2017evaluating} builds on these models and implementations to evaluate several alternative systems for peer review.
These models and results differ from ours in several key dimensions: the authors are not strategic, they do not focus on fine-grained policies by journals (or in our case, conferences), and due to the holistic nature and complexity of the model, the model is only amenable to simulation, but not analytically tractable.

Two papers by \citet{bianchi2018peer} and \citet{squazzoni2012saint} also use agent-based modeling approaches. They particularly focus on the fact that researchers must decide how to divide their time between writing and reviewing papers, and investigate (experimentally) the impact of various policies on the efficiency of peer review.
Similarly, \citet{thurner2011peer} and \citet{d2017can} use agent-based models to investigate a specific aspect of peer review, namely, selfish behavior on the part of referees, who may not have incentives to see other strong work published.

\citet{allesina2012modeling} also uses agent-based modeling, in this case, to understand the impact of different high-level approaches (editorial desk rejects, bidding on papers, etc.) on the overall reviewing load.
\citet{Roebber2011PeerReview} use agent-based modeling to evaluate strategies for program officers of funding agencies. One of their findings is similar to ours: that requiring unanimous support for accepting a proposal (i.e., setting a high threshold) can discourage authors from submitting many proposals, thus lowering the review burden.

A more analytical approach is taken by \citet{smith2021accept}. Here, the authors are also interested in the impact of self selection on the acceptance rate of a journal (or university). They study a system with multiple journals or universities announcing different thresholds in the presence of noisy reviews. Due to their motivation, their model does not appear to account for resubmissions, thus differing from our work in a key aspect.

A model for the journal peer review process was proposed and analyzed by \citet{azar2015model}. 
In this model, the authors also know their paper’s quality, referees observe noisy signals, and editors set an acceptance threshold in equilibrium. Similar to one of our insights in \cref{sec:dominating-value-discount}, a key takeaway from this work is that higher submission costs, while discouraging authors, can improve journal quality by deterring low-quality submissions. A key distinction in our model is that authors may resubmit the same paper multiple times to the same venue, reflecting the difference between the typical scenarios for conferences and journals. This feature gives rise to a broader range of design dimensions, such as the QB-tradeoff.

Several other works do more basic theoretical analysis of the impact of conference policies. In particular, they focus on the false positives (accepting bad papers) and false negatives (rejecting good papers) arising as a function of the number of reviewers and their individual qualities.
Based on such calculations, \citet{herron2012expert} suggests that obtaining a large number of low-quality reviews may be better than a small number of expert reviews.  
\citet{neff2006peer} focus in particular on the role of desk rejects by an expert editor.

In response to concerns by authors about the evaluation of their work when submitted to conferences or journals, the scientific community in general, and CS community in particular, has engaged in significant self-evaluation efforts. Many of these have focused on the quality and consistency of reviews provided to conferences \citep{shah2018design,cicchetti1991reliability,ebel1951estimation}.
\citet{cole1981chance} and \citet{tran2020open} study the reproducibility and randomness in review scores and acceptance decisions.
Furthermore, the community has experimented with (or at least suggested) different formats, including increasing the number of reviews per paper, multi-level or multi-stage evaluation processes, having reviews from past submissions follow a paper upon resubmission, and many others \citep{jecmen2020mitigating,noothigattu2018choosing,shah2018design,rogers2020can}.
Many of these approaches appear primarily driven by concern for authors and their desire for accurate evaluation of their submission, though some of them are also part of our evaluation.

Other attempts that try to mitigate the overwhelming demand for reviewers consider solutions based on mechanism design. 
\citet{srinivasan2021auctions} combine a bidding system and peer prediction to simultaneously incentivize high-quality reviews and high-quality submissions. \citet{su2021best} designs a mechanism that elicits ranking information truthfully from the authors, which is proven to empirically benefit the conference's quality.

In our idealized model, we assume that all reviews are i.i.d., that is, reviewers are subject to the same noise distribution. Naturally, this is a simplification. In reality, one of the difficulties faced by conferences and journals is how to aggregate the scores from reviewers with possibly very different scales or expectations.
Indeed, such aggregation is a well-known fundamental problem in statistics \citep{ammar2012efficient,cook2007creating,wang2018your}.
%\dkcomment{Should we say more about these papers?}

Peer review can also be viewed through the lens of a principal-agent problem: the principal decides on the review process and the acceptance rule, and the agents respond.
Besides peer review, related applications include admitting college students \citep{kannan2021best} and recruiting faculty \citep{zhang2021classification}, or endorsing a product \citep{gill2012optimal,lerner2006model}.
In particular, \citet{gill2012optimal,lerner2006model} study models in which the agent (such as an author) can choose among venues one that maximizes the expected utility, as determined by the chance of success and the prestige.

%% file: OR-resbumission/sections/Model.tex
We consider a process of a (large) group of \emph{authors} each submitting a paper to a prestigious \emph{conference}.
\dkdeletecomment{too much detail here? Also mentioned in discussion later, I think.}{This single conference can refer to multiple ``equivalent'' conferences, e.g., ICML/NeurIPS or AAAI/IJCAI or STOC/FOCS.}
We model the submission-reviewing process as a multi-round game \dkedit{between two types of players: the conference and the authors. 
First, the conference commits to an acceptance policy, which maps vectors of reviews to a decision to accept or reject a paper.
Subsequently, }
in each round, each author decides to submit the paper either to the \dkdelete{prestigious} conference or take the outside option. 
Upon submission, the \dkdelete{prestigious} conference \dkreplace{will send the paper out for review and, based on the reviews,}{obtains reviews and applies its acceptance policy to} decide to accept or reject.
If the paper is rejected from the prestigious conference, the author sees the reviews, and faces the same decision problem in the next round. 
\dkedit{More precisely, \fangedit{given parameters $(\ACCMAP, \NumReviews, \RevSigDist, \QualDist, \AuthSigDist, \TD)$, }\fangcomment{do we introduce $\NumNewPapers$?}the game proceeds as follows:
\dkcomment{I am not sure I like Fang-Yi's edits, adding the notation. I think that the roles of the parameters are pretty clear, and the overview is supposed to contextualize the later detailed subsections. I don't think we need to be fully formal here and psecify the distributions etc.}
\begin{enumerate}
    \item The conference commits to a number $\NumReviews$ of reviews that will be solicited, and an \emph{acceptance policy} $\ACCMAP$, which maps vectors of $\NumReviews$ reviews to an outcome from $\SET{\text{accept}, \text{reject}}$.
    Acceptance policies are discussed in detail in \cref{sec:model-acc-policy}.
    \item The author's paper quality $\Qual=\qual$ is realized from a commonly known distribution $\QualDist$; assumptions on the distribution are discussed in \cref{sec:paper-qualities-signals}.
    \item The authors observes a (possibly noisy) signal $\AuthSig$ about $\qual$ \fangedit{from $\AuthSigDist$}, as detailed in \cref{sec:paper-qualities-signals}.
    \item Based on the signal, the conference's policy, and past reviews (if any), the author decides whether to submit the paper to the conference or take the outside option. The decision process is detailed in \cref{sec:author-decision}.
    \label{enum:author-decision}
    \item If the paper was submitted, the conference obtains $\NumReviews$ reviews. These reviews are signals $\RevSig{i}$ drawn i.i.d.~from a commonly known family of distributions $\RevSigDist[q]$ conditioned on the true but unknown quality $q$ of the paper.
    \item The conference applies the acceptance policy $\ACCMAP$ to the vector of reviews to determine an outcome from $\SET{\text{accept}, \text{reject}}$.
    \item If the paper was submitted and rejected, the game returns to step (\ref{enum:author-decision}). Otherwise, the game ends.
 \end{enumerate}

 If the paper is accepted to the conference after $t$ prior rounds of rejection, the author's utility is $\TD^t \cdot \ConfValue$, where $\TD$ is the commonly known author time discount factor, and $\ConfValue$ is the (endogenous) conference prestige, determined by the qualities of accepted papers. If the author takes the outside option after $t$ prior rounds of rejection, her utility is $\TD^t$. The author's utility is discussed in more detail in \cref{sec:author-decision}.

 The conference's utility/quality is 1 plus the average quality of papers, averaged over all rounds, in the limit as the number of rounds grows large. The conference utility is discussed in detail in \cref{sec:model-acc-policy}.

 In addition to the conference and authors, we are also interested in the utility/cost of the reviewers, which is simply the total number of reviews per conference/paper. The concept of review burden, as well as the tradeoffs we study between the different parties' utilities, is discussed in detail in \cref{subsec:number-reviews}. We emphasize that while we are interested in the utility of the reviewers, they have no agency in our game-theoretic formulation: only the conference and authors make decisions.
}

\dkdeletecomment{now subsumed in the overview}{We will discuss the utilities of three parties: the conference committee, reviewers, and authors. However, in the game, only two of the main agents actually make decisions: the author of each paper and the prestigious conference;  the reviewers simply provide reviews.}

\subsection{\dkedit{Paper Qualities and Signals}}
\label{sec:paper-qualities-signals}

Each paper has a quality $\Qual$ drawn i.i.d.~from a commonly known prior paper quality distribution $\QualDist$ over the set of possible qualities. The support of $\QualDist$, i.e., the set of all possible paper qualities, is denoted by $\QualSet \subseteq \mathbb{R}$, and larger qualities correspond to better papers.
% We note that the assumption that $\QualDist$ is commonly known is not essential\textemdash it only matters that the conference (e.g., PC chair) knows the distribution.
Without loss of generality, there exist both negative and positive values in $\QualSet$; otherwise, the conference would simply accept/reject all papers without review. 
When the set of qualities is discrete, we write $\QualProb{q} = \Prob{\Qual = q}$; when it is continuous, we use $\QualDens{q}$ to denote the density of $\QualDist$ at $q$.
Because all papers' qualities are drawn from the same distribution, we will not need to reference a specific paper in our notation.

The conference cannot observe the true quality of the paper, but will solicit some number, $\NumReviews$, of reviews for each submission.  
Each review $S_j$, for $j = 1, \ldots, \NumReviews$, is a random variable drawn i.i.d.~from a distribution $\RevSigDist[q]$ where $q$ is the paper's quality. 
Thus, the reviews are independent conditioned on the paper's quality.
The outcome of the $j$th review is $\RevSig{j} \in \SigSet \subseteq \R$ where $\SigSet$ is the set of possible review scores, and a higher score denotes a more positive review. We assume that $\RevSigDist[q]$ has full support on $\SigSet$ for every $q$.
We write $\RevSigV$ for the vector of the $\NumReviews$ reviews. We let $U(\RevSigV) = \ExpectC[\Qual]{\Qual}{\RevSigV, \QualDist, \RevSigDist}$ denote the expected quality of a paper conditioned on the reviews $\RevSigV$.\footnote{Note that this notation of expected quality only depends on the quality prior and the review signal, but not the author's best response and resubmission. }
% We assume that the reviews are \emph{informative} about the paper's quality, in the sense of monotone likelihood ratio: 
Some of our results rely on the following assumption on the review signal.

\begin{definition}[\citet{karlin1956theory}] \label{def:informative}
The reviewer's family of signal distributions $\RevSigDist$ satisfies the monotone likelihood ratio (MLR) property if for any $q'>q$, whenever $\REVSIGP > \REVSIG$, the following holds: $\frac{\RevSigProb[q']{\REVSIGP}}{\RevSigProb[q]{\REVSIGP}} > \frac{\RevSigProb[q']{\REVSIG}}{\RevSigProb[q]{\REVSIG}}$.
\end{definition}
% \yichicomment{Need better name than "informative". 
% "Why not impose Definition 1 as a running assumption
% instead, or simply say “assume that reviewers’ signals satisfy MLRP.” (Imposing
% MLRP is sufficiently standard that the authors should not be concerned that
% this sounds overly technical or narrow.)"}

It is well known that a wide class of commonly used distributions --- including the Gaussian, Beta, and Exponential distributions --- have monotone likelihood ratios \citep{karlin:rubin,karlin1956theory}.

For some of our results, we assume that the authors perfectly observe $\Qual$, the paper's quality. This model has the advantage of being analytically tractable, because the authors do not learn new information from the reviews. We call such authors \emph{noiseless}.
For other results, we consider \emph{noisy} authors, who themselves only assess their papers' qualities approximately.
In this case, we assume that authors have noisy signals $\AuthSig$, which, similar to the conference's signals, are drawn according to some distribution $\AuthSigDist[q]$ for a paper of quality $q$Since our theoretical results are primarily derived in the noiseless setting, we do not impose any specific assumptions on the distribution of author noise. Instead, we investigate its effect through agent-based model simulations, using noise rates estimated from real data.
The author's signal is independent of the conference's signals conditioned on $\Qual$.
Noisy authors will update their beliefs about their papers' qualities based on the reviews in a Bayesian way. 

\dkedit{
\begin{remark}
We remark that our results for noiseless authors can be generalized to other settings in which the authors obtain no new information from observing reviews when their submission is rejected.
Perhaps the most natural way to model such a setting is to assume that each review signal $S_j$ is an i.i.d.~\emph{garbling} of the author's signal $\AuthSig$; in that case, authors learn no new information from reviews.
Because neither the author nor the conference can obtain any more information about a paper's quality beyond its conditional expected quality $\ExpectC{Q}{\AuthSig}$, we can simply treat this expectation itself as the paper's quality, and redefine the prior over paper qualities as the corresponding distribution over expected qualities based on signals. That is, if $\ExpectC{Q}{\sigma_q} = q$, then we define $p'_q = \sum_{q'} \QualProb{q'} \cdot \RevSigProb[q']{\sigma_q}$. Note that under the MLR property, the quality conditioned on a signal is strictly increasing (see \cref{lem:monotone_expected_quality}), so no two signals can have the same expected quality. In particular, $\sigma_q$ is in fact unique (if a $\sigma$ exists at all for a particular $q$).
With this reduction, our results carry over directly from the setting of noiseless authors. We stick to the model of noiseless authors here to avoid unnecessarily complex notation.
\end{remark}
}

We study categorical and continuous models in this paper, as defined next. When a result does not specify one of these models explicitly, it holds for both models.

\begin{description}
\item{\textbf{Categorical Model:}}
In a categorical model, both $\QualSet \subset \R$ and $\SigSet \subset \R$ are finite (and ordered by their natural order on $\R$).
% \dkcomment{I suggest deleting the discussion of Dawid-Skene here. Seems to waste 3 lines for a model which is really not that similar.}
% Such a model bears some similarity to the Dawid-Skene model \citep{dawid1979maximum}, although the latter typically does not assume an ordering on labels and thus also no \yzdelete{informativeness} requirement akin to monotone likelihood ratio.

\item{\textbf{Continuous Model:}} 
In the continuous model, we assume that both the quality distribution $\QualDist$ and the noise have atomless full support on all of $\R$, and are suitably continuous.
Thus, $\QualSet = \SigSet = \R$, and we make the following specific assumptions:
\begin{itemize}
    \item The cumulative distribution function (cdf) of $\QualDist$ is continuous and strictly increasing. In other words, the probability density function (pdf) is positive and finite: $\QualDens{q} \in (0,\infty)$ for all $q\in \QualSet$.
    \item The reviewers' signals are obtained by adding to the true quality $\Qual$ a noise term drawn from a distribution $\REVNOISEDIST$ which is \emph{independent} of $\Qual$, has a strictly positive and \emph{continuous} pdf, and has zero mean.
% \fangreplace{
%     Let $\REVNOISEDIST$ denote the cdf of this noise distribution, and $f^{(r)}$ the pdf.
%     Then, we assume that $f^{(r)}(x) \in (0, \infty)$ for all $x$, that $f^{(r)}(x)$ is a continuous function of $x$, and that $\int_x x \cdot f^{(r)}(x) \, d x = 0$.
%     The cdf of the distribution of reviewer signals conditioned on a quality $\Qual=q$ is $\RevNoiseDist{x-q}$.}
Specifically, let $\REVNOISEDIST$ and  $f^{(r)}$ denote the cdf and pdf of this noise distribution, respectively.  The cdf of the distribution of reviewer signals conditioned on a quality $\Qual=q$ is $\RevNoiseDist{x-q}$.  We assume that not only $F^{(r)}$, but also $f^{(r)}$ is continuous, with $f^{(r)}(x) \in (0, \infty)$ for all $x$, and $\int_x x \cdot f^{(r)}(x) \, d x = 0$.
\end{itemize}

Note that the zero-mean assumption is without loss of generality so long as the noise distribution is independent of the quality $\Qual$, as any (known) bias could be subtracted from the reviews by the conference.

\dkcomment{Rewrote this again to add continuity assumptions. I think that those might be needed. Maybe there is some very clever way to get around it and still have continuous expected quality in signals, but I don't see it right now. If someone sees a way to avoid it, please feel free to implement. Also, please definitely check my edits here!!!}
\yzcomment{Is it possible that some of these overlap with the MLR assumption? I briefly check and it doesn't seem to be the case, but one can take another look.}
\dkcomment{I was wondering, too, but did not immediately see a way to make it happen.
Though it may be true. I think that by setting $s'-q' = s-q$, one can prove that $(f(x))^2 > f(x-\delta) f(x+\delta)$ for all $x, \delta > 0$, which looks like a weaker form of log-concavity. Not sure if it implies log-concavity by repeated application, or maybe by using different pairs. Maybe MLR implies log-concavity more generally. If so, there may be known results that log-concavity implies continuity.}
\end{description}

\subsection{Conference Acceptance Policy and Quality}
\label{sec:model-acc-policy}

The conference's main lever of control is its acceptance policy.
We primarily focus on \emph{memoryless} acceptance policies.
Under a memoryless acceptance policy, (1) the author can submit a paper an unlimited number of times; 
(2) the same number of reviews $\NumReviews$ and decision policy are used in every round; and (3) in each round $t$ in which the author (re-)submits the paper, the conference's decision depends only on the reviews obtained in round $t$. In other words, each submitted paper is treated as a fresh paper. For that reason, we typically omit the round $t$ from the notation when discussing memoryless policies. 
We discuss alternatives to memoryless policies in \cref{sec:memory_policy}, in particular, limiting the number of times a paper can be submitted, and having old reviews follow a resubmission. 
Except for these sections, unless stated otherwise, all acceptance policies are memoryless.  
We call an acceptance policy \emph{non-trivial} if it neither accepts nor rejects all submissions with quality in $\QualSet$.

A memoryless acceptance policy is characterized by a function $\ACCMAP: \SigSet^m \to [0,1]$ which determines the probability with which each combination of review signals leads to a paper's acceptance.\fangcomment{Future note: we consider non-adaptive acceptance policy.  Alternatively, if one can change the acceptance policy based on the number of submission, this adaptive policy can have more power.}

We primarily focus on \emph{monotone} acceptance policies. 
A policy is monotone if for any two vectors $\RevSigV, \RevSigVP$ of reviews, a higher expected quality of a paper
$U(\RevSigV) \geq U(\RevSigVP)$ implies that a higher acceptance probability $\AccMap{\RevSigV} \geq \AccMap{\RevSigVP}$.
A particularly natural class of monotone acceptance policies prescribes a conditional expected quality threshold.  
We largely restrict our attention to threshold policies because 1) they comprise a natural set of policies, and other policies seem unlikely to arise in practice; 2) it is both conceptually and computationally easy to search over such policies; and 3) we will see that such policies are sufficiently rich to induce any author strategy that can be induced using any monotone acceptance policy.

\begin{definition} \label{def:threshold-policy}
A \emph{threshold acceptance policy} $\ACCMAP[\tau, r]$ is characterized by a threshold $\tau \in \mathbb{R} \cup \{-\infty, +\infty\}$ and a probability $r\in [0,1]$. It accepts a paper with reviews $\RevSigV$ when 
$U(\RevSigV) > \tau$, rejects the paper when 
$U(\RevSigV) < \tau$, and accepts the paper with probability $r$ when $U(\RevSigV) = \tau$.
\end{definition}

When the distribution $\Prob[\RevSigV]{U(\RevSigV)  = \tau}$ has no point mass for any $\tau$, the third (knife-edge) case is an event of probability 0; we then omit $r$ from the notation, and simply use $\ACCMAP[\tau]$ to denote the threshold acceptance policy with threshold $\tau$; in particular, this is true in the continuous model and some generalizations.

Once an acceptance policy $\ACCMAP$ is fixed, the probability of a submitted paper being accepted in a particular round is only a function of its underlying quality $q$. We denote this probability by $\AccP{\ACCMAP}{q}$.
The following proposition shows that when $\ACCMAP$ is monotone, $\AccP{\ACCMAP}{q}$ is monotone in $q$. We defer the proof to \cref{app:proof-monotone-prob}.

\begin{proposition} \label{prop:monotone-prob}
Assume that the reviewers' signals satisfy the  MLR property.
If $\ACCMAP$ is a monotone acceptance policy, then $\AccP{\ACCMAP}{q}$ is non-decreasing in $q$. Moreover, if $\ACCMAP[\tau,r]$ is a non-trivial threshold policy, then $\AccP{\ACCMAP[\tau,r]}{q}$ is strictly increasing in $q$, for any prior $\QualDist$.
\end{proposition}  
 
In round $t$ of (re-)submission, we suppose that the conference faces a (the same) number $\NumNewPapers$ of new papers as well as the previously rejected papers from the past $t-1$ rounds. 
We define the \emph{quality} $\CONFUTIL$ of the conference as the long-run expected total quality of accepted paper.  Formally, $\CONFUTIL$ is the expectation of the total value of accepted papers' quality normalized by $\NumNewPapers$ and the number of rounds $t$ with $t\to \infty$.\fangcomment{$t$ or $T$?} \dkcomment{Not sure. I thought we use $T$ for the conference's lever, but do we do so consistently?} As is typical in Stackelberg games, we assume that the authors best-respond; if there are multiple best responses, for convenience of analysis, we assume that we can prescribe a particular tie breaking.
Note that when $t \to \infty$, the expected number of submissions in round $t$ that have previously been submitted $\ell$ times converges for any $\ell \in \mathbb{N}$.

\subsection{Author Utility and Decisions}
\label{sec:author-decision}
\dkdeletecomment{Now subsumed earlier?}{
In terms of timing, first, the conference announces its review and acceptance policy; subsequently, in each round $t$, the author decides whether to submit the paper to the conference or take the side option.
% The game ends when the paper is accepted at the conference or at the ``sure bet'' option. 
The game ends when the paper is accepted at the conference or the author takes the side option.}

Two key factors affect the author's utility: their time discount factor $\TD$, capturing how patient they are, and the prestige $\ConfValue$ they ascribe to the conference.
When the paper is accepted at the conference in round $t$, the author's utility is therefore $\yichireplace{\AUTHUTIL}{u^{(a)}} = \TD^{t-1} \cdot \ConfValue$; when the author takes the side option in round $t$, her utility is $\TD^{t-1}$, i.e., the utility of the outside option is normalized to 1.
The $\TD^{t-1}$ term encodes exponential time discounting and models that authors would like their work to be published in a timely manner.
Besides the utility loss due to time discounting, rejection does not cause additional cost for the author.

We assume that the prestige $\ConfValue$ is 1 plus the average quality of the papers accepted by the conference, which is the reward to authors whose papers are accepted by the conference.\fangcomment{Future note: I feel this definition loses the control of conference's exogenous value.  Our result should hold under $V = 1+\mu+\lambda\E[q|accepted]$.} \dkcomment{Indeed, it seems like Fang-Yi's generalization might check out, and subsume both results.}
This assumption aligns the preference for papers between authors and the conference.
The conference aims to accept only positive-quality papers, while authors are attracted to the conference over the outside option only when the average quality of accepted papers is larger than 1.\footnote{While not included here, we also have shown that the results hold when $V > 1$ is a fixed constant, e.g., the prestige of the conference is based on its past reputation, and not directly affected by current decisions or accepted papers.}
We assume that $\Set{q\in \QualSet}{q>0}\neq \emptyset$, \dkreplace{meaning}{i.e., that some papers are of positive quality. This assumption implies that $\ConfValue$ can be made larger than the outside option, when only positive value paper are submitted and accepted; otherwise, no author would ever submit to the conference.}
\yzcomment{removed an assumption.}

The author's decisions depend on all available information, i.e., her own (possibly perfect) private signal $\AuthSig$ as well as all the reviews she has received for previous submissions.
We assume that the author is rational and Bayesian, so her decisions are based on posterior quality distributions taking into account all available information. 
She will submit to the conference in round $t$ if and only if her expected utility from doing so (over all future time steps $t' \geq t$) exceeds her expected utility from the side option (which is exactly $\TD^{t-1}$ at the point she is making the decision).
Notice that unless the author obtains a perfect signal, the reviews she obtains (in addition to her own signal $\AuthSig$) change her posterior conditional probability over the paper's quality, which in turn changes her belief of the probability distribution of future reviews. 

\begin{definition}
  The model in which authors have perfect information about their papers' qualities, and papers may be resubmitted an unlimited number of times, is called the model of \emph{noiseless authors with unlimited resubmissions}.
\end{definition}

Under the model of noiseless authors with unlimited resubmissions, a theoretical analysis becomes more tractable. This is because authors' beliefs of their papers' qualities will not be updated based on the reviews. As a result, the papers that are submitted in the first round will be repeatedly resubmitted until acceptance. In turn, this implies that the quality of the conference depends entirely on the authors' self-selection. 

\subsection{Review Burden and Tradeoffs}  
\label{subsec:number-reviews}
% % As we discussed in the introduction, we are primarily interested in the tradeoff between the conference's quality and the review burden, which is captured by the number of requested reviews imposed on the community. \yichicomment{Removed a footnote here.}
% \footnote{This downplays the utility of the authors, a fact which is discussed in \cref{sec:limitations}.} 
% We call this the QB-tradeoff. 
% We denote the expected number of reviews of a paper by $\PaperReviews$, which will be called the \emph{review burden}.
As we discussed in the introduction, we are primarily interested in the utility of the conference committee, the reviewers, and the authors.
The ``utility'' of the conference committee is essentially captured by the conference quality $\CONFUTIL$ as defined in \cref{sec:model-acc-policy}.
The utility of the authors, captured by the author welfare, $\AUTHUTIL$, is the sum of all authors' utilities conditioned on acceptance to the prestigious conference, and is denoted by $\AUTHUTIL$.\fangcomment{Future note.  Why don't we define the author welfare as the sum of all author's utility, but only those getting accepted to the prestigious one?}\yzcomment{I don't feel strongly about the difference. Maybe a note for future.}
We use the \emph{review burden} $\PaperReviews$ as a measure of the reviewers' utility, which is the expected number of reviews assigned to a random paper.
% \gscomment{I don't understand why author's welfare is conditioned on submitting.  Do we mean that their paper is yet to be accepted?  That might be a clearer way to state it.}\yzcomment{changed to conditioned on acceptance.}
To be precise, the conference's policy, along with an author's best response, determines a probability distribution $\NumSubmitDist$,
where $\NumSubmitProb{t} = \Prob{\text{The paper is submitted at least $t$ times}}$.
Then, $\PaperReviews = \NumReviews \cdot \sum_{t = 1}^{\infty}  \NumSubmitProb{t}$ is the expected number of reviews incurred by a submission.
In the limit as the number of rounds
% \gscomment{please double check this edit.}\yzcomment{I don't think this makes sense. What we require is a large number of rounds of resubmissions so that the accepted papers in this round contain resubmitted papers from many rounds ago, and the quality of accepted papers converges}\fangcomment{Are we using the number of review incurred by a submission or the number of review required in each round?  The former does not require $t\to \infty$.}\yzcomment{I think we want to claim these are the same if $t\to \infty$}\fangcomment{We define"$\PaperReviews$ as a measure of the reviewers' utility, which is the expected number of reviews assigned to a random paper."  Why do we the later?}\gscomment{looks okay now with Fang-Yi's edit.} 
becomes large, the reviewing load is spread out evenly over rounds. Therefore, with $\NumNewPapers$ new submissions in each round, the review burden on the community is $\NumNewPapers \PaperReviews$.

We focus on the tradeoffs between pairs of stakeholders, and in particular, we find the tradeoffs between conference quality and review burden, as well as between conference quality and author welfare, to be the most interesting. Intuitively, the tradeoff between review burden and author welfare is trivially dominated by the simple policy of accepting every paper without review.
Starting with the tradeoff between conference quality and review burden (QB-tradeoff), we say that a policy $\ACCMAP$ with $\CONFUTIL$ and $\PaperReviews$ \emph{weakly QB dominates} 
% \gscomment{seems like $\CONFUTIL \PaperReviews$ or $CR$ dominates would be a better term.  
% Also, this seems like a bit of a bandaid because we are too lazy to factor in author utility, which nearly always aligns with the above two.} 
another policy $\hat{\ACCMAP}$ with $\CONFUTILH$ and $\PaperReviewsH$ if (1) it has a higher (or equal) expected conference utility, $\CONFUTIL \geq \CONFUTILH$, and (2) the number of reviews is smaller (or equal), $\PaperReviews \leq \PaperReviewsH$.  
We say that a policy $\ACCMAP$ \emph{QB dominates} another policy $\hat{\ACCMAP}$ if it weakly QB dominates, and at least one of the inequalities is strict.
Given a set of policies, a policy is \emph{QB Pareto optimal} if it is not strictly QB dominated by any other policy in the set.
If policy $\ACCMAP$ and policy $\hat{\ACCMAP}$ do not induce unique conference qualities and review burdens (because they allow different best responses by the author), we say that $\ACCMAP$ (weakly) QB dominates $\hat{\ACCMAP}$ if every pair $(\CONFUTILH,\PaperReviewsH)$ induced by $\hat{\ACCMAP}$ is (weakly) dominated by a pair $(\CONFUTIL,\PaperReviews)$ induced by $\ACCMAP$.

We say that the QB-tradeoff in one setting \emph{weakly dominates} another if for any point (corresponding to a policy $\hat{\ACCMAP}$) of the second QB-tradeoff, there exists a point (corresponding to a policy $\ACCMAP$) in the first setting that weakly dominates it.

Typically, instead of optimizing over all memoryless acceptance policies, we will restrict our attention to memoryless threshold acceptance policies.  In such a case, we can look at the \emph{QB-tradeoff curve}, which maps out each point of the QB-tradeoff as the acceptance policy threshold $\tau$ increases from $-\infty$ to $+\infty$, and for $\tau$ where $\Prob[\RevSigV]{U(\RevSigV)  = \tau} > 0$, $r$ increases from 0 to 1.  

We say that one QB-tradeoff curve $\mathcal{C}$ \emph{(weakly) dominates} another QB-tradeoff curve $\hat{\mathcal{C}}$ if for any point on $\hat{\mathcal{C}}$ that does not correspond to accepting all papers or rejecting all papers, there exists a point on $\mathcal{C}$ which (weakly) dominates it.
Note that this implies that no point on the Pareto frontier of the QB-tradeoff curve $\mathcal{C}$ is dominated by any of the points on the QB-tradeoff curve $\hat{\mathcal{C}}$.  

The definition of the tradeoff between conference quality and author welfare (QA-tradeoff) is analogous --- replacing review burden with author welfare in all relevant concepts --- while noting that higher author welfare is preferred, just as lower review burden is preferred. We then use \emph{QA dominance} to refer to the relationship between two policies when considering conference quality and author welfare. When the context is clear, we omit the QA or QB qualifiers.

%% file: OR-resbumission/sections/Best-response.tex
% \gscomment{Reorg.  Define threhsold acc policy and threshold response.  Give intuitive overview of current 3.1  (say that threshold policies are general, and authors typically have de facto threshold.  Then give powerful story of resubmission gap.  Then 3.1 becomes "threshold policies and the resubmission paradox".  3.2 "an illustratrative example with Gaussians"  do an example of Gaussian (what is currently 3.3).  Then 3.3 "Authors best response policies".  3.4 "Generality of threshold strategies" becomes the current technical work from 3.1. "Generality of threshold strategies" (last two might be combined?)   }

We first focus on the case of noiseless authors. Recall that in this setting, because the author will not update her belief about the paper's quality, she will either submit to the conference until the paper is accepted (or, in the case of limited submissions, until the paper may not be resubmitted anymore), or immediately take the side option. 
In turn, this allows us to analytically characterize the relationship between the author's submission equilibrium and the conference's acceptance policy.

% \dkdeletecomment{I really don't like ``Table of Contents'' paragraphs, and even less for just a single section. I don't think a reader gets anything out of this paragraph.}{
% \yzreplace{
% In this section, we focus on the \emph{resubmission gap}: the difference between the threshold the conference sets for acceptance and the actual threshold of accepted papers, taking into account the resubmission of previously rejected papers. An analysis of the resubmission gap is of interest in its own right (since it may crucially inform conference acceptance policies), and also serves as the foundation of our further investigation of tradeoffs in \cref{sec:continuous-tradeoffs}.
% }{
% This section is organized as follows. In \cref{subsec:paradox}, we \gsreplace{present}{state} the high-level technical results and highlight their key practical implications\textemdash an explanation of the resubmission paradox. In \cref{sec:additive-noise}, we then introduce a simple \gsreplace{yet}{and} intuitive additive noise model to illustrate how model parameters shape the author’s equilibrium behavior. Finally, we provide general theoretical justifications for our claims across a range of models and acceptance policies in \cref{subsec:threshold_broader}.}}

\subsection{Threshold Policies and The Resubmission Paradox}
\label{subsec:paradox}

% \begin{quote}
% \emph{At many prestigious conferences, a nontrivial fraction of papers rejected in one year are resubmitted to the same or closely related venues in subsequent years and are eventually accepted --- sometimes with only minor changes.}
% \end{quote}

% \dkcomment{Are we actually quoting someone here? If not, quote environment seems wrong. Instead, I would suggest something like the following:}
Our analysis begins from the following often puzzling real-world phenomenon: 
At many prestigious conferences, a nontrivial fraction of rejected papers are subsequently resubmitted (often multiple times) to the same or closely related venues, only to be eventually accepted --- sometimes with only minor changes.

We refer to this counterintuitive phenomenon as the \emph{resubmission paradox}.
Its prevalence has led to suggestions of raising acceptance rates so that more of these papers are accepted in their initial round of submission, thereby reducing repeated reviews and lowering the review burden on the community.
\yzcomment{It'd be nice to have a citation.}

In this subsection, we explain the resubmission paradox by uncovering the relationship between a conference’s acceptance policy and authors’ equilibrium behavior. The paradox turns out not to be an accident, but a byproduct of two forces: the conference’s desire to maintain its prestige and authors’ self-selection in deciding where to submit. Simply lowering the acceptance threshold sounds like it should help, but it risks drawing in more lower-quality submissions and thus fails to reduce --- and may even increase --- the overall review burden.

\paragraph{Threshold equilibrium.}
In order to understand the effects of a conference's acceptance policy, it is necessary to characterize the authors' equilibrium under the policy. Fortunately, under mild conditions on the acceptance policy, authors follow a threshold strategy at equilibrium, submitting papers if and only if the quality exceeds some threshold. This observation motivates a central concept in our analysis, which offers a simple and intuitive explanation for the resubmission paradox.

\begin{definition}[De Facto Threshold]
Consider a conference with a memoryless acceptance policy~$\ACCMAP$ and noiseless authors.  
We say that a value~$\theta$ is the \emph{de facto threshold under $\ACCMAP$} if, in any non-atomic symmetric equilibrium\footnote{``Non-atomic'' here means that any unilateral deviation by a single author does not affect the conference value $\ConfValue$, and ``symmetric'' means that all authors follow the same strategy.}, authors submit --- and continue resubmitting --- any paper with quality~$\Qual = q$ whenever $q > \theta$, and take the outside option if $q < \theta$.
%\dkedit{$\theta$ is a \emph{de facto threshold} if it is a de facto threshold under some policy $\ACCMAP$.}
\end{definition}

We call $\theta$ the de facto threshold under $\ACCMAP$ because every paper whose quality is above $\theta$ will be submitted until accepted, so de facto, the conference will accept all papers above $\theta$. 
Based on the definition, the de facto threshold exists for an acceptance policy \ACCMAP if and only if authors submit papers according to the above threshold strategy in every non-atomic symmetric equilibrium. In particular, if in a strategy profile, there are qualities $q \neq q'$ such that an author is indifferent between submitting or not submitting papers of qualities $q$ and $q'$, then a de facto threshold does not exist under \ACCMAP.

% \gscomment{I think that the order of the bullet points below should be reversed.  Then it goes from most easy to understand to harder to understand.  I think that will help the reader.}
% \yzcomment{Sounds good. Do we also want to change the proposition?}\gscomment{I think so, unless it makes the proof more difficult, in which case just leave it.}
We defer a detailed discussion of the existence of the de facto threshold to \cref{sec:submission-threshold}, and provide only a high-level summary (of \cref{prop:de_facto}) here:  
\begin{itemize}[leftmargin=2em]
    \item In the continuous model, under a threshold acceptance policy~$\ACCMAP[\tau]$, every equilibrium corresponds to a unique de facto threshold~$\theta$, and $\theta$ and $\tau$ are in one-to-one correspondence.
    \item In the categorical model, under a threshold acceptance policy~$\ACCMAP[\tau,r]$, every equilibrium corresponds to some de facto threshold~$\theta$, but the mapping from $\tau$ to $\theta$ may not be unique.
    \item In both the categorical and continuous models, a de facto threshold~$\theta$ exists under any monotone \ACCMAP. However, other non-threshold equilibria may also exist.
\end{itemize}

Threshold acceptance policies comprise a very natural class of policies for a conference to apply. We further justify this by showing the following in \cref{prop:threshold-policy}: If $\theta$ is such that including all papers of quality strictly exceeding $\theta$, and possibly some papers of quality exactly $\theta$, would result in a conference of value $\ConfValue > 1$, then there exists a threshold acceptance policy inducing $\theta$ as its de facto threshold.

The concept of de facto thresholds, along with the characterization from the aforementioned \cref{prop:threshold-policy}, allows for a clean characterization of the policy maximizing conference quality:

% \dkcomment{Moved up the proposition, proof, and discussion from 3.3, per Grant's suggestion.}
\begin{proposition}\label{prop:max-general}
The conference quality-maximizing policy induces a de facto threshold of $\theta = 0$, and there is a threshold acceptance policy inducing this de facto threshold.
 \end{proposition}

\proof{Proof of \cref{prop:max-general}}
Given a de facto threshold $\theta$, the conference's quality $\CONFUTIL$ is $\int_{\theta}^\infty \qual \,d\QualDens{\qual}$. (For discrete quality distributions, the corresponding conference quality is obtained by replacing the integral with the sum and the density function with the probability function.)

This integral/sum is maximized at $\theta=0$.
Furthermore, because we assumed that papers of positive quality occur with positive probability, the integral is strictly positive for $\theta = 0$, so that $\ConfValue > 1$.
By the discussion preceding this proposition (and proved in \cref{prop:threshold-policy}), $\theta=0$ is induced as a de facto threshold by some threshold acceptance policy.
% \gscomment{it is not clear to me that there is a general way to find a threshold policy that maps to some defacto threshold.  Could we modify prop 5 so that this is more clear?}
\Halmos
\endproof

\paragraph{Explaining the resubmission paradox.}
The acceptance threshold and the de facto threshold tend to be different.
The threshold $\tau$ can be interpreted as the ``declared'' quality goal of the conference, attained in isolation. $\theta$ determines the actual quality of the conference, taking into consideration resubmissions and reviewing noise. The difference between $\tau$ and $\theta$ is a key concept we study, which we term the ``resubmission gap''.

\begin{definition}[Resubmission Gap]
\label{def:resubmission-gap}
Consider a memoryless conference with a non-trivial threshold acceptance policy $\ACCMAP[\tau,r]$, and a noiseless author.
Let $\theta$ be the smallest de facto threshold under $\ACCMAP[\tau,r]$. 
The difference between the acceptance threshold $\tau$ and $\theta$ is called the \emph{resubmission gap}.
\end{definition}
In general, the de facto threshold is not unique in the categorical model; we define the resubmission gap based on the \emph{smallest} de facto threshold. This issue disappears in the continuous model, where the resubmission gap is unique for every non-trivial threshold acceptance policy. 

The name ``resubmission gap'' reflects the discrepancy between the ``stated'' and ``de facto'' quality of accepted papers, caused by the fact that authors are free to resubmit papers and that the reviews are noisy.
The fact that the actual threshold of papers at the conference differs from the declared acceptance threshold for each year in isolation (namely, by the resubmission gap) may appear somewhat unexpected at first, and is frequently missing from conversations about policy changes. 

% \yzreplace{
% The resubmission gap leads to the following oddity, which we call the \emph{resubmission paradox}: all submitted papers will eventually be accepted, but in any one round, many (and often most) are rejected. This creates a substantial and seemingly unnecessary burden on the review process, as papers are reviewed multiple times before final acceptance. While accepting all papers initially seems like a solution, in our model, it is not.  The only way to initially accept all papers is to lower the acceptance threshold, which in turn lowers the de facto threshold and incentivizes additional submissions.  Moreover, these additional papers will often be rejected multiple times, but will be resubmitted until they are accepted.  Thus, this apparent solution actually reproduces the problem, just with a lower de facto threshold.
% }
% {
The resubmission gap is precisely what gives rise to the resubmission paradox highlighted at the start of this subsection. As mentioned there, a common proposal is to lower the acceptance threshold in order to reduce the review burden by cutting down on resubmissions. However, this approach is unlikely to help. While a smaller threshold does allow more papers to be accepted earlier, it simultaneously lowers the de facto threshold and encourages additional submissions. Moreover, many of these new submissions are sub-standard and will still be rejected multiple times until eventual acceptance. Thus, this apparent solution actually reproduces the problem, just with a lower de facto threshold.

\paragraph{The importance of the resubmission gap.}
We highlight the central role of the resubmission gap in many of our key results.

Fixing a de facto threshold fixes the conference quality $\CONFUTIL$ and the conference value $\ConfValue$. (In the categorical model, we also need to fix the submission probability in the edge case $\Qual = \theta$.) 
As we have seen, a particular de facto threshold $\theta = 0$ maximizes the conference quality (\cref{prop:max-general}).
Now, for a fixed de facto threshold, a larger resubmission gap implies a higher acceptance threshold. A higher threshold, in turn, means that each paper is expected to be submitted more times before eventual acceptance. Consequently, a larger resubmission gap naturally increases the review burden --- since each paper generates more reviews\textemdash and decreases author utility, as authors prefer earlier acceptance given a fixed discount factor.

We will investigate the factors that affect the resubmission gap in \cref{sec:additive-noise}, and examine how the resubmission gap relates to another key concept in our paper\textemdash the QB- or QA-tradeoff\textemdash in \cref{sec:continuous-tradeoffs}.
% }

\subsection{An Illustrative Example in the Continuous Model With a Single Review}
\label{sec:additive-noise}
In order to illustrate the somewhat abstract definitions more concretely, we now characterize the relationship between the acceptance threshold $\tau$ and the de facto threshold $\theta$ specifically for the continuous model with $\NumReviews = 1$ review per paper. 
Importantly, recall that the review noise is additive and the noise distribution has zero mean and is independent of the true underlying quality.
That is, for any quality $\Qual=q$, the reviewer observes a signal of $q+X$ where the distribution of the review noise, $\RevNoiseDist{X}$, is independent of $q$. 

We first note that the conference value depends on the author's submission strategy. As we primarily consider symmetric threshold equilibria, we write the conference value $\ConfValue(\theta)$ as a function of a threshold quality $\theta$ such that authors with papers of quality $\Qual \ge \theta$ submit and resubmit until acceptance, and authors with $\Qual < \theta$ opt for the outside option. It follows that the conference value is strictly increasing in~$q$:
\[\ConfValue(\theta) = 1 + \frac{\int_\theta^\infty q \QualDens{q} dq}{\int_\theta^\infty \QualDens{q} dq}.\]

Under the single-review model, it is convenient to define a \emph{signal-based threshold acceptance policy}, which accepts a paper if and only if its review signal exceeds a threshold~$\tau_s$. Since review noise is additive to the true paper quality, and the noise distribution has an invertible cdf, there is a monotone one-to-one mapping between a signal threshold~$\tau_s$ and a threshold~$\tau$ applied to the paper’s expected quality (given a fixed prior on quality). 
The following proposition formalizes this correspondence.
% Under this model, the following proposition shows that
% % the resubmission gap of a threshold acceptance policy does not depend on its threshold $\tau$.
% there is a monotonic one-to-one correspondence between the acceptance threshold $\tau$ and the de facto threshold $\theta$.

\begin{proposition} \label{prop:gap-invariant}
  Given $\theta$ such that $\ConfValue(\theta)>1$, consider the continuous model with a single review and additive noise drawn from $\REVNOISEDIST$. 
  The acceptance threshold~$\tau_s$ of a signal-based threshold policy that induces $\theta$ as the de facto threshold is
  \[\tau_s = 
  % \theta + \left(\REVNOISEDIST\right)^{-1}{\left(1-\frac{1}{\rho(\theta)}\right)} = 
  \theta + \left(\REVNOISEDIST\right)^{-1}{\left(\frac{\ConfValue(\theta) - 1}{\ConfValue(\theta) - \TD}\right)},\]
  which is increasing in $\theta$.
  % In particular, the resubmission gap $\tau - \theta$ is independent of $\tau$ and the prior distribution $\QualDist$ over paper qualities.
\end{proposition}
% \fangcomment{this part is fishy.  The proposition provides the *threshold on the signal*, but $\tau$ should be the threshold on the conditional expectation $U(s)$.  I am worry, because the threshold acceptance policy $\tau$ should also depend on the prior distribution $\QualDist$. If there is no simple fix, we may need to further restrict to the two step Gaussian.}
% \fangcomment{This seems like a corollary of proposition 2.}
\proof{Proof of \cref{prop:gap-invariant}}
  % Consider an acceptance threshold $\tau$, and corresponding policy $\ACCMAP[\tau]$. Recall that under a continuous model, the probability that the conditional expected quality of a paper is exactly $\tau$ is 0, so $\tau$ uniquely defines the threshold acceptance policy.
  % By Proposition~\ref{prop:de_facto}, the de facto threshold $\theta$ satisfies $\AccP{\ACCMAP[\tau]}{\theta} = 1/\rho$.

We first show that the threshold acceptance policy can induce $\theta$ as a de facto threshold as long as $\ConfValue(\theta)>1$.
An author submits the paper if the expected utility of submitting is greater than 1 (the utility of the outside option), does not submit the paper if the expected utility is less than 1, and is indifferent between submitting and not submitting if the expected utility is equal to 1. 
Because a noiseless author does not learn any new information from rejection in previous rounds, she will make the same decision in future rounds, implying that she will submit until acceptance.
Let $\AccP{\ACCMAP[\tau_s]}{q}$ be the acceptance probability of a paper with quality $Q=q$ under a signal-based threshold acceptance policy. By Proposition~\ref{prop:monotone-prob}, this probability is strictly increasing in $q$.
The expected utility can be obtained as the time-discounted sum of the utility from acceptance:
% \dkcomment{Should this be $u^{(a)}$ or $U^{(a)}$? I am not quite clear on which we use for what.}\yzcomment{Should be small u. Small u is individual utility, U is welfare.}
\begin{align}
u^{(a)}(\ACCMAP[\tau_s],q)
& = 
\sum_{t\ge 1} \ConfValue \cdot \TD^{t-1} \AccP{\ACCMAP[\tau_s]}{q} \cdot (1-\AccP{\ACCMAP[\tau_s]}{q})^{t-1} 
\; = \; \frac{\ConfValue \cdot \AccP{\ACCMAP[\tau_s]}{q}}{1-\TD \cdot (1-\AccP{\ACCMAP[\tau_s]}{q})}.
\end{align}
Solving the inequalities $u^{(a)}(\ACCMAP[\tau_s],q) \ge 1$ and $u^{(a)}(\ACCMAP[\tau_s],q) < 1$, and $u^{(a)}(\ACCMAP[\tau_s],q) = 1$ for $q$, the author submits the paper if $\AccP{\ACCMAP[\tau_s]}{q} > (1-\TD)/(\ConfValue(\theta)-\TD)$, does not submit if $\AccP{\ACCMAP[\tau_s]}{q} < (1-\TD)/(\ConfValue(\theta)-\TD)$, and is indifferent between submitting and not submitting if $\AccP{\ACCMAP[\tau_s]}{q} = (1-\TD)/(\ConfValue(\theta)-\TD)$, respectively.

  % In the continuous model, the de facto threshold \dkreplace{together with the}{and} quality distribution \dkedit{together} determine the conference value, i.e., $\ConfValue(\theta) = 1+ \frac{\int_\theta^\infty q\QualDens{q} dq}{\int_\theta^\infty \QualDens{q} dq}$.
  % By definition of the de facto threshold and \cref{lem:author_response}, $\AccP{\ACCMAP[\tau]}{\theta} = 1/\rho(\theta)$, where $\rho(\theta) = \frac{\ConfValue(\theta)-\TD}{1-\TD}$.

% \yzreplace{
% In the continuous model with a single review, the conference will accept a paper with quality $q$ if and only if the review $s$ satisfies $s=q+x>\tau$. This happens with probability $\AccP{\ACCMAP[\tau]}{q} = 1-\RevNoiseDist{\tau-q}$.
% Thus, $\tau$ solves $1/\rho(\theta) = 1-\RevNoiseDist{\tau-\theta}$. }

Under the signal-based threshold acceptance policy, a paper with quality $q$ is accepted if and only if the review $s$ satisfies $s=q+x>\tau_s$.
This happens with probability $\AccP{\ACCMAP[\tau_s]}{q} = 1-\RevNoiseDist{\tau_s-q}$.
Thus, $\tau_s$ solves $(1-\TD)/(\ConfValue(\theta)-\TD) = 1-\RevNoiseDist{\tau_s-\theta}$.

A solution for $\tau$ exists because the left-hand side of the above equation lies between 0 and 1 (because $\TD \in (0, 1)$ and $\ConfValue(\theta) > 1$) and $\RevNoiseDist{\cdot}$ has a continuous cdf. 
The uniqueness of $\theta$ follows because we assumed $\RevNoiseDist{\cdot}$ to be \emph{strictly} increasing.
Rearranging the equation leads to the relationship as shown in the statement of the proposition.
\Halmos
\endproof

% \yzdelete{
% Substituting $\theta = 0$ from Proposition~\ref{prop:max-general} into Proposition~\ref{prop:gap-invariant}, 
% %we immediately obtain the following corollary.
% \yichiedit{we immediately observe that the conference quality-maximizing acceptance threshold is $\left(\REVNOISEDIST\right)^{-1}{\left(\frac{\ConfValue(0) - 1}{\ConfValue(0) - \TD}\right)}$.
% Under this acceptance policy, the author submits and resubmits the paper if and only if the quality is non-negative. The resulting maximum quality for a conference is $\CONFUTIL = \int_{0}^\infty \qual \,\QualDens{\qual}\,dq$.}
% }

Because there is a monotone one-to-one mapping between $\tau_s$ and $\tau$, the resubmission gap $\tau-\theta$ is monotone increasing in the term $\left(\REVNOISEDIST\right)^{-1}{\left(\frac{\ConfValue(\theta) - 1}{\ConfValue(\theta) - \TD}\right)}$.
\cref{prop:gap-invariant} thus implies that there are four factors that affect the resubmission gap: the author's discount factor $\TD$, the paper quality distribution $\QualDist$, the review noise distribution $\REVNOISEDIST$, and the de facto threshold $\theta$ itself. A large discount factor $\TD$ (very patient authors), a high-quality-skewed quality distribution $\QualDist$ (which leads to a larger $\ConfValue$ given the same $\theta$), a large review noise, and a higher de facto threshold (which leads to a larger $\ConfValue$) all contribute to a large resubmission gap.
This means that the conference has to set a significantly higher bar to sufficiently discourage such resubmissions, and will reject many good papers repeatedly before they are finally accepted. 
% \yzdelete{In contrast, the combination of a lower-tier conference with a low de facto threshold (and thus a low conference value) and impatient authors may result in a negative resubmission gap\fangcomment{Why can the resubmission gap be negative?  Isn't $V\ge 1$?}\textemdash meaning that the conference needs to significantly lower its acceptance threshold to provide strong enough assurance to good papers that they will be immediately accepted. Note that this can only occur for $\ConfValue(\theta) < 2$; otherwise, the resubmission gap is non-negative for authors with any level of patience.}
% \yzcomment{Removed the discussion on negative gap, which was wrong as we did not consider $\tau$ on the expected quality space. Whether a resubmission gap can be negative depends on the mapping from $\tau$ to $\tau_s$, which depends on the prior. This may need some numerical tests. No time for this revision.}

The above intuition carries over to the more general setting where a larger resubmission gap has an important practical implication: a setting with a smaller resubmission gap can achieve the same conference quality at a lower review burden.  In \cref{subsec:gap-tradeoffs}, we will see this formally and then study what factors reliably increase this gap.

\subsection{Justifying Threshold Equilibria and Policies in Broader Settings}
\label{subsec:threshold_broader}

We now establish that de facto thresholds exist broadly across the models considered in this paper, and that any candidate threshold can be induced by a threshold acceptance policy as a de facto threshold. These results demonstrate the robustness of our central concepts --- such as the resubmission gap --- and thereby strengthen the practical insights derived from our model. 

\subsubsection{Characterizing De Facto Thresholds}
\label{sec:submission-threshold}

When the conference's acceptance policy $\ACCMAP$ is fixed, from the author's perspective, it induces an acceptance probability $\AccP{\ACCMAP}{q}$ for each paper quality $q$.
By Proposition~\ref{prop:monotone-prob}, this probability is weakly increasing in $q$ if the acceptance policy is monotone, and strictly increasing under a non-trivial threshold policy. 
We will show that as a result, if the acceptance policy is monotone, the author's equilibrium strategy can be characterized by a submission threshold $\theta$.
Note, however, that there may be other best responses that do not fit this pattern.
% the threshold best response clearly illustrates the impact of the conference's acceptance policy, aligning well with our goal.
% \dkcomment{I agree with the first half of the preceding sentence, but am unsure about what the second one is supposed to say. Are we justifying a focus only on author threshold policies? We didn't actually say that we focus on such policies.} \yichicomment{Commented out the second half.}

To state this result, we first define suitable nomenclature and notation.

\begin{definition}
In a \emph{$\theta$-threshold strategy} for the author, an author submits (and resubmits until they are accepted)  all papers of quality $\Qual = q > \theta$, and no paper of quality $q < \theta$. The author may handle papers of quality $q = \theta$ arbitrarily.  
An author strategy is called a \emph{threshold strategy} if it is a $\theta$-threshold strategy for some $\theta$.
\end{definition}

When $q = \theta$, the author is typically indifferent between submitting and not submitting. In this case, we assume that we can prescribe the author’s tie-breaking behavior for ease of analysis. 

% \yzdelete{
% We now formally define a central concept of our paper:\gscomment{this was already defined in 3.1.  I think we should remove it here. and the commentary below.}

% \begin{definition}[De Facto Threshold]
% Consider a conference with memoryless acceptance policy $\ACCMAP$ and noiseless authors.
% A value $\theta$ such that under every best response the author submits a paper of quality $\Qual = q$ if $q > \theta$ and does not submit any paper with $q < \theta$ is called a \emph{de facto threshold}.
% \end{definition}

% We call $\theta$ the de facto threshold because every paper whose quality is above $\theta$ will be submitted until accepted, so de facto, the conference will accept all papers above $\theta$.
% Based on the definition, the de facto threshold exists for an acceptance policy if and only if every best response of the author is a threshold strategy. In particular, if there are qualities $q \neq q'$ such that the author is indifferent between submitting or not submitting papers of qualities $q$ and $q'$, then the de facto threshold does not exist.
% }

% We begin our analysis by characterizing the best responses of an author to monotone acceptance policies.
% The following proposition characterizes the best responses of an author to monotone and threshold acceptance policies.
We capture the appeal of submitting to a conference by an \emph{attractiveness factor} $\rho$ (where larger values correspond to conferences more attractive to submit to), based on its value $\ConfValue$ and the discount factor $\TD$:
\begin{align}
\rho & := \frac{\ConfValue - \TD}{1-\TD}.
\label{eqn:rho-definition}
\end{align}

\begin{lemma}\label{lem:author_response}
    Consider a memoryless conference with a monotone acceptance policy $\ACCMAP$. 
    Suppose that the conference value $\ConfValue$ is fixed, and accordingly $\rho$ is fixed.
    Then, the author's best response is to submit the paper (in each round) if $\AccP{\ACCMAP}{\Qual} > 1/\rho$, take the side option when $\AccP{\ACCMAP}{\Qual} < 1/\rho$, and the author is indifferent between submitting or not submitting if $\AccP{\ACCMAP}{\Qual} = 1/\rho$.
\end{lemma}
We defer the proof to \cref{app:proof-author_resp}; it closely parallels the proof of \cref{prop:gap-invariant}. Note that we use the term ``best response'' rather than ``symmetric equilibrium'' in the statement of this lemma, since when $\ConfValue$ is fixed, each author’s utility is independent of the strategies of others. 
% \yzdelete{Intuitively, the acceptance probability of a paper, and thus the expected utility of the author, is increasing in its quality under monotone acceptance policies. This results in a threshold best response for the authors. Then, \cref{lem:author_response} follows from reasoning about the critical quality that makes the author indifferent between submitting and taking the side option in different settings.}

% \yzreplace{
% We further show that under a monotone acceptance policy, threshold strategies are always one of the author's best responses, and are the unique best response when the acceptance policy is a threshold one.
% To appreciate the result, first note that both the conference value $\ConfValue$ and the attractiveness factor $\rho$ depend on the author's response.}
We now examine how the dependence of the conference value~$\ConfValue$ and the attractiveness factor~$\rho$ on the author’s strategy influences the equilibrium.
Let $\ConfValue(\theta, r)$ denote the conference value when all authors with paper quality $Q > \theta$ choose to submit and resubmit until acceptance, those with $Q < \theta$ opt for the outside option, and authors with $Q = \theta$ submit with probability $r$. 
We then have the following formula for the conference value in the categorical model: 
% \gscomment{we just said above that we are ignoring r, but then it still apears below.}
% \dkreplace{\begin{equation*}
%     \ConfValue(q, r) = 1 + \frac{r\cdot q \cdot \QualProb{q} + \sum_{q' \in \QualSet, q' > q} q' \cdot \QualProb{q'}}{r\cdot \QualProb{q} + \sum_{q' \in \QualSet, q' > q} \QualProb{q'}}.
% \end{equation*}}{
\begin{equation*}
    \ConfValue(\theta, r) = 1 + \frac{r\cdot \theta \cdot \QualProb{\theta} + \sum_{q \in \QualSet, q > \theta} q \cdot \QualProb{q}}{r\cdot \QualProb{\theta} + \sum_{q \in \QualSet, q > q} \QualProb{q}}.
\end{equation*}
% \dkcomment{replaced all $q$ with $\theta$ and $q'$ with $q$.}
Because $\QualDist$ has full support on $\QualSet$, $\ConfValue(\theta,r)$ is strictly increasing in $\theta$ and strictly decreasing in $r$. This implies that the corresponding attractiveness factor $\rho(\theta,r)=\frac{\ConfValue(\theta,r) - \TD}{1-\TD}\in (0,1)$ is strictly increasing in $\theta$ and strictly decreasing in $r$.  

% \yzdelete{
% In the continuous setting, since papers with quality exactly $q$ have measure zero, the parameter $r$ can be omitted, and we write the conference value simply as $\ConfValue(q)$.
% The formula for the conference value is analogous to the categorical model, while we replace the summations with integrals.
% }

% Note that conference value has to be at least 1, implying that not all de facto thresholds are feasible. A de facto threshold $\theta\in \QualSet$ is feasible if there exists an $r\in [0,1]$ such that $\ConfValue(\theta, r)>1$, or if $\theta\notin \QualSet$, it is feasible if $\ConfValue(\bar{\theta}, 1)>1$ where $\bar{\theta} = \min\{q\in\QualSet\mid q>\theta\}$. In the continuous model, $\theta$ is feasible if $\ConfValue(\theta)>1$. 

For the existence of a meaningful conference value, the acceptance policy should at least encourage \emph{some} papers to submit.
The following definition captures this desideratum.
\begin{definition}
    We say that an acceptance policy $\ACCMAP$ is \emph{responsive} if it is non-trivial and there exists a quality $\bar{q}\in \QualSet$ such that $\AccP{\ACCMAP}{\bar{q}} > 1/\rho(\bar{q},1)$.
\end{definition}
\yichicomment{better name than responsive?}
\dkcomment{inviting? non-discouraging?}
\yzcomment{I don't feel strongly that these are better than the current one.}

\yzcomment{There is a suggestion of reversing the order of the following 3 points. I feel the wording and the proof are already well-adjusted for the ordering, and it doesn't seem like an essential change for this submission. I'll mark the idea and leave it for now.}
\begin{proposition} \label{prop:de_facto}
Assume that authors are noiseless, and the conference uses a monotone acceptance policy $\ACCMAP$ that is memoryless and responsive.
\begin{enumerate}
   \item There exists a $\theta$ such that it is a non-atomic symmetric equilibrium for every author to play a $\theta$-threshold strategy.%
   \footnote{There may also exist non-threshold equilibria, and thus $\theta$ may not be a de facto threshold.}
  \item If the conference policy $\ACCMAP$ is additionally a threshold one, then there exists a de facto threshold, implying that there exists only one non-atomic symmetric equilibrium; in this equilibrium, every author plays a $\theta$-threshold strategy.
  \item If the conference applies a threshold acceptance policy with threshold $\tau$ and the model is continuous, then the de facto threshold $\theta$ is unique: no best response is a $\theta'$-threshold strategy for $\theta' \neq \theta$.  
  Moreover, $\AccP{\ACCMAP[\tau]}{\theta} = 1/\rho(\theta)$. 
\end{enumerate}
\end{proposition}

% \begin{proposition} \label{prop:de_facto}
% Assume that authors are noiseless, and the conference uses a monotone acceptance policy $\ACCMAP$ that is memoryless and responsive.
% \begin{enumerate}
%   \item If the conference applies a threshold acceptance policy with threshold $\tau$ and the model is continuous, then there exists a unique de facto threshold $\theta$ such that no best response is a $\theta'$-threshold strategy for $\theta' \neq \theta$.  
%   Moreover, $\AccP{\ACCMAP[\tau]}{\theta} = 1/\rho(\theta)$. 
%   \item  If the conference applies a threshold acceptance policy with threshold $\tau$ and the model is categorical, then there exists a de facto threshold only one non-atomic symmetric equilibrium; in this equilibrium, every author plays a $\theta$-threshold strategy.
%    \item There exists a $\theta$ such that it is a non-atomic symmetric equilibrium for every author to play a $\theta$-threshold strategy.%
%    \footnote{There may also exist non-threshold equilibria, and thus $\theta$ may not be a de facto threshold}.
% \end{enumerate}
% \end{proposition}
% Note that our results (here and later) depend on $\TD$ and $\ConfValue$ only through $\rho$. Therefore, in a sense, they are ``interchangeable,'' albeit not linearly. That is, an increase in author patience ($\TD$) is tantamount to a (different) increase in conference prestige, as far as author behavior is concerned.

\Cref{prop:de_facto} follows from \cref{lem:author_response}, together with the monotonicity of $\AccP{\ACCMAP}{q}$ and $\rho(q)$.
This result implies that any monotone acceptance policy induces a form of self-selection among authors, whereby clearly substandard papers are not submitted in equilibrium.
In the following sections, we explore what factors affect the de facto threshold and how the conference can leverage this self-selection behavior to balance the utilities of different stakeholders in the system.

% \yichicomment{To do: what's below is unlikely to hold anymore.}
% \cref{prop:de_facto} assumes that the agents are allowed to resubmit their papers arbitrarily many times. This assumption is not essential: an essentially identical calculation (in \cref{sec:time_limited_policy}) shows that the de facto threshold does not change when the conference restricts the number of times a paper can be resubmitted.

\subsubsection{Sufficiency of Threshold Policies}
\label{sec:threshold-resubmission}  

% \dkreplace{
% \yzreplace{
% Threshold acceptance policies (see \cref{def:threshold-policy}) comprise a very natural class of policies for a conference to apply. 
% Recall that they accept all papers whose posterior (based on the reviews) expected quality strictly exceeds some threshold, and reject all papers whose posterior expected quality falls short of the threshold. }
% {Now, we justify the focus on threshold acceptance policies by showing} that every feasible de facto threshold can be induced by a threshold acceptance policy.}
Having shown that threshold best responses are ``typically'' best for \emph{authors}, we next show that threshold acceptance policies are sufficient for the \emph{conference}, in that they let it induce every candidate threshold as a de facto threshold.

\begin{definition}\label{def:candidate_threshold}
    We say that a value $\theta \in \R$ is a \emph{candidate threshold} if the resulting conference value is larger than the utility of the outside option. Formally, $\theta$ is \a candidate threshold if:
% \gsreplace{\[
% \begin{cases}
% \ConfValue(\dkreplace{\bar{\theta}}{\min\{q \in \QualSet \mid q > \theta\}}, 1) > 1 & \text{if } \theta \notin \QualSet%, \text{ where } \bar{\theta} = \min\{q \in \QualSet \mid q > \theta\}
% , \\
% \dkreplace{\exists\,}{\text{there exists an }} r \in [0,1] \text{ such that } \ConfValue(\theta, r) > 1 & \text{if } \theta \in \QualSet.
% \end{cases}
% \]}
\[
\begin{cases}
\ConfValue(\theta, r) > 1 \text{ for some } r \in [0,1]   & \text{if } \theta \in \QualSet
, \\
\ConfValue(\inf \Set{q \in \QualSet}{q > \theta}, 1) > 1 & \text{if } \theta \notin \QualSet%, \text{ where } \bar{\theta} = \min\{q \in \QualSet \mid q > \theta\}
.
\end{cases}
\]
Here, we use the standard convention that $\inf \emptyset = \infty$, so $\theta$ can only be a candidate threshold if there exists at least one $q \in \QualSet$ with $q \geq \theta$.
% \dkcomment{Replaced min by inf, as for continuous domains, the minimum of that set doesn't exist, and for discrete, there's no difference. Explicitly addressed the empty set.}

Let $\mathcal{C}$ denote the set of all candidate thresholds.
\end{definition}
% Intuitively, a feasible de facto threshold cannot be too low, as this would make the highest achievable conference value worse than the outside option. In such cases, even an acceptance probability of 1 would not make authors with borderline papers ($Q = \bar{\theta}$) indifferent between submitting and opting out.
Because $\ConfValue(q,r)$ is increasing in $q$, if $\theta$ is a candidate threshold and $\theta' > \theta$, then $\theta'$ is also a candidate threshold, so long as at least one $q \in \QualSet$ weakly exceeds $\theta'$. 

% \dkreplace{Therefore, there exists an infimum feasible de facto threshold $\theta_{\inf}$ such that any $\theta > \theta_{\inf}$ is feasible and any $\theta < \theta_{\inf}$ is infeasible. 
% Whether $\theta = \theta_{\inf}$ is feasible depends on whether the model is continuous or categorical.
% In the continuous model, there exists a $\theta_{\inf}$ such that $\ConfValue(\theta_{\inf}) = 1$, implying that $\theta = \theta_{\inf}$ is not feasible. In the categorical model, there exists $(\theta_{\inf}, r_{\inf})$ such that $\ConfValue(\theta_{\inf}, r_{\inf}) = 1$ for some probability $r_{\inf}\in (0,1]$. Therefore, $\ConfValue(\theta, r) > 1$ for any $\theta > \theta_{\inf}$ and $r\in (0,1]$ while $\ConfValue(\theta, r) \le 1$ for any $\theta < \theta_{\inf}$ and $r\in (0,1]$, meaning that $\theta = \theta_{\inf}$ is feasible but any $\theta < \theta_{\inf}$ is infeasible.}

Let $\theta_{\inf} = \inf \mathcal{C}$ be the infimum of all candidate thresholds.
By the preceding paragraph, we obtain that either $\mathcal{C} = [\theta_{\inf}, \max \QualSet]$ or $\mathcal{C} = (\theta_{\inf}, \max \QualSet]$.
Both cases can occur, depending on whether the model is continuous or categorical.
In the continuous model, $\ConfValue(\theta_{\inf}) = 1$, implying that $\theta_{\inf}$ is not a candidate threshold. 
In the categorical model, there exists an $r_{\inf} \in (0,1]$ such that $\ConfValue(\theta_{\inf}, r_{\inf}) = 1$. 
Therefore, $\ConfValue(\theta_{\inf}, r_{\inf}/2) > 1$, and $\theta_{\inf}$ is indeed a candidate threshold.

% \yichicomment{ the current statement sounds circular. Consider replacing “Let $\theta\in R$ be any de facto threshold”
% with “For every $\theta\in R$” (since we don’t know whether $\theta$ is a de facto threshold
% until a conference acceptance rule has been specified).}
% \dkcomment{I think I resolved this by calling it a candidate threshold, and removing the ``de facto threshold'' part of the definition, which I think was not needed --- this was only a property of conference values conditioned on included papers. I hope that I did not terribly mess this up.}

\begin{proposition} 
\label{prop:threshold-policy}
Let $\theta \in \mathcal{C}$ be a candidate threshold. 
Then, there exists a threshold acceptance policy with threshold $\hat{\tau}$ and probability $\hat{r}$ under which the following is a symmetric equilibrium for noiseless authors: submit to the prestigious conference which uses $\ACCMAP[\hat{\tau},\hat{r}]$ if $\Qual > \theta$ and take the side option if $\Qual < \theta$ (and the author is indifferent between submitting or not submitting if $\Qual=\theta\in \QualSet$). \dkcomment{addressed?}
Moreover, if the model is continuous and the authors are neither submitting all papers nor submitting no papers, the threshold $\hat{\tau}$ is unique.
\end{proposition}

We defer the proof to \cref{app:proof-threshold-policy}, while providing a high-level sketch here. 
% We say an acceptance policy is stricter than another if it accepts every paper with a (weakly) smaller probability. While restricting to threshold acceptance policies, we show that a threshold policy is stricter than another if and only if its threshold is higher. 
% First, note that a paper of quality $\theta$ has an acceptance probability of $1/\rho$ (\cref{prop:de_facto}). Then, we show that for threshold acceptance policies, the acceptance probability of any paper is monotone decreasing in the acceptance threshold. Thus, the proposition follows because the acceptance probability of a threshold policy with $\tau = -\infty$ is $1$, and with $\tau = \infty$, it is $0$. Therefore, because $1/\rho \in (0,1)$, there must exist (at least) one acceptance threshold such that the acceptance probability of a paper of quality $\theta$ equals $1/\rho$.
For any $\theta \in \mathcal{C}$, there exists a borderline quality $\hat{q} = \inf\Set{q \in \QualSet}{q \ge \theta}$ (which may differ from $\theta$ in the categorical model), such that if $\AccP{\ACCMAP}{\hat{q}} = 1/\rho(\hat{q}, r)$ for some $r$, then $\theta$ is a de facto threshold.
The key insight is that, under threshold acceptance policies, the acceptance probability decreases monotonically with the threshold $\tau$, ranging from 1 as $\tau = -\infty$ to 0 as $\tau = \infty$.
This implies that for any candidate threshold $\theta$ where the corresponding $\hat{q}$ satisfies that $1/\rho(\hat{q}, r)\in (0,1)$, there must exist an acceptance threshold such that the acceptance probability of a paper of quality $\hat{q}$ equals $1/\rho(\hat{q}, r)$. 

Proposition~\ref{prop:threshold-policy} in part justifies our focus on threshold policies.  For any monotone policy, Proposition~\ref{prop:de_facto} implies the existence of a symmetric equilibrium where authors play a threshold strategy, and thus, Proposition~\ref{prop:threshold-policy} implies the existence of a threshold policy for which the author best-responds in the same way.
\fangcomment{We need to say the threshold in \cref{prop:de_facto} are in $\mathcal{C}$ \cref{def:candidate_threshold}.} \yzcomment{Do we have a proof relying on this?}\fangcomment{This paragraph suggests we may use threshold acceptance policy to replace any monotone policy.  However, under \cref{prop:de_facto} there may be some monotone policy whose threshold equilibrium is not in $\mathcal{C}$.}
Thus, assuming that authors break ties in favor of using threshold strategies, any conference quality that can be achieved with a monotone policy can be achieved with a threshold acceptance policy.  
Interestingly, it does not follow that the QB-tradeoff for threshold policies weakly dominates that of all monotone acceptance policies, as we will show in \cref{sec:testing_tradeoff}.

While Proposition~\ref{prop:threshold-policy} implies the existence of an acceptance threshold $\tau$ inducing the desired submission threshold $\theta$, these two thresholds will typically be different. 
{Indeed, the difference is exactly the resubmission gap we defined in \cref{def:resubmission-gap}, and which is one of our central quantities of interest.}

% \dkcomment{There was a lot of duplication of actual text, and definitely content, defining resubmission gap and resubmission paradox again. I assume it was just a remnant of editing, not intentional.}

% \gscomment{I feel like the below could actually be moved into section 3.1.  3.3 seems to be for more technical analysis justifying particular assumptions and how general we can make things.  This is notes and bolts kinda stuff that we will use repeatedly, and it only require de facto threshold.}
% \dkcomment{Implemented this suggestion.}

%% file: OR-resbumission/sections/Tradeoffs.tex
% \gscomment{ Define QB trade off and QA tradeoff.  some pictures. 4.0.1 summary of results (little chart about noise and discount). 4.2 Dominanting Threshold.  4.3 Point-wise worse resubmission gap implies weakly dominant QB-tradeoff curve.  4.4 reviewer noise (general)  4.5 Author discount (continuous model only for QA).
% 4.6 Acceptance Rate.

% Look at two cases:  

% Current: 4.2 what thresholds work best (high versus low) 4.3 QB (insert how this is related to resubmission gap)  4.4 QA. 4.4 Review noise (more noise worse for  QB and QA), 4.5 discount good for QB mixed for QA.  Acceptance Rate (add comment about how acceptance rate depends on the distribution as well as the thresold change mu and see what happens.).     }

In this section, we build on the fundamental concepts of threshold policies, de facto thresholds, and resubmission gap, to undertake a more in-depth investigation of the tradeoffs a conference may face. 
% In particular, we consider the tradeoff between conference quality and review burden on the community, and between acceptance thresholds and acceptance rate. 
In particular, we focus on the utility of the three stakeholders in our model, as reflected by conference quality, review burden, and authors' welfare.
We also examine another central feature of a conference --- its acceptance rate --- and analyze how it depends on the acceptance threshold.

\subsection{Tradeoffs in the Continuous Model}
\label{subsec:QB-tradeoff}
We begin by characterizing the utilities of the three stakeholders in the continuous model, where there is a clean one-to-one mapping between the acceptance threshold and the de facto threshold, and visualize the resulting tradeoff curves using numerical examples. As noted in \cref{sec:model}, we are primarily interested in two tradeoffs: conference quality versus review burden (the QB-tradeoff) and conference quality versus author welfare (the QA-tradeoff). 

We have discussed the conference quality in the previous section, and observed that the conference will eventually accept every paper whose quality is above the de facto threshold $\theta$. Thus, the conference quality is $\CONFUTIL = \int_\theta^\infty q \QualDens{q}dq$. This expression shows that, once the quality prior is fixed, the de facto threshold~$\theta$ alone determines the conference quality.

As for the review burden, we consider the average number of reviews per \emph{paper} (including papers that were never submitted, and thus incurred $0$ reviews) as the relevant measure; note that because the total number of papers is constant, this is equivalent to considering the total number of reviews provided by the community.
This accounts for all reviews a paper receives over multiple rounds of submission --- first submission, second submission, and so on --- until it is eventually accepted.
Let $\theta$ be a candidate threshold, and let the corresponding acceptance policy be $\ACCMAP$. The total review burden under the continuous model is given by:
\begin{align}\label{eq:review_burden}
    R(\theta) &= m \int_{\theta}^\infty \QualDens{q} \cdot \left(\sum_{t=0}^{\infty} (1-\AccP{\ACCMAP}{q})^t \right) \; dq =  m\int_{\theta}^\infty \QualDens{q}/\AccP{\ACCMAP}{q} \; dq.
\end{align}

We next consider author welfare. In the continuous model, a de facto threshold $\theta$ determines a conference value $\ConfValue$ and induces an acceptance policy $\ACCMAP$, as shown in \cref{prop:de_facto}.
The expected utility of an author with a paper of quality $q\ge \theta$ who decides to submit and keep resubmitting until acceptance is
\begin{equation*}
    u^{(a)}(q, \ACCMAP, \ConfValue) 
    = \sum_{t=0}^{\infty} \ConfValue \cdot \AccP{\ACCMAP}{q} \cdot (\TD \cdot (1-\AccP{\ACCMAP}{q}))^t
    = \frac{\AccP{\ACCMAP}{q} \cdot V}{1-\eta \cdot (1-\AccP{\ACCMAP}{q})}.
\end{equation*}
The author welfare of all submitted papers is then $U^{(a)}(\theta) = \int_{\theta}^\infty \QualDens{q} \cdot u^{(a)}(q, \ACCMAP, \ConfValue) \, dq$.

Given a quality prior, varying the acceptance threshold (and thus the de facto threshold) varies the utilities of the three stakeholders in the above discussed ways. This allows us to visualize the QB tradeoffs and the QA tradeoffs.
In this section (and the following sections that contain discussions on the continuous model), we frequently use the following special cases of our general continuous model as examples for our analysis and plots.

The \emph{$(\sigma, \QualDist, \NumReviews, \TD)$-Gaussian model} is a continuous model with noiseless authors; the noise for each review is drawn from a Gaussian distribution $\REVNOISEDIST = \Gaussian{0}{\sigma}$. The parameters $\QualDist, \NumReviews$, and $\TD$ are, as before, the prior, the number of solicited reviews, and the discount factor, respectively. 
The \emph{$(\sigma, \mu_{\QualDist}, \sigma_{\QualDist}, \NumReviews, \TD)$-Double Gaussian model} is a Gaussian model with the prior $\QualDist = \Gaussian{\mu_{\QualDist}}{\sigma_{\QualDist}}$.

\begin{figure}[htb]
     \FIGURE
     {\begin{subfigure}[b]{0.48\textwidth}
         \centering
         \includegraphics[width=\textwidth]{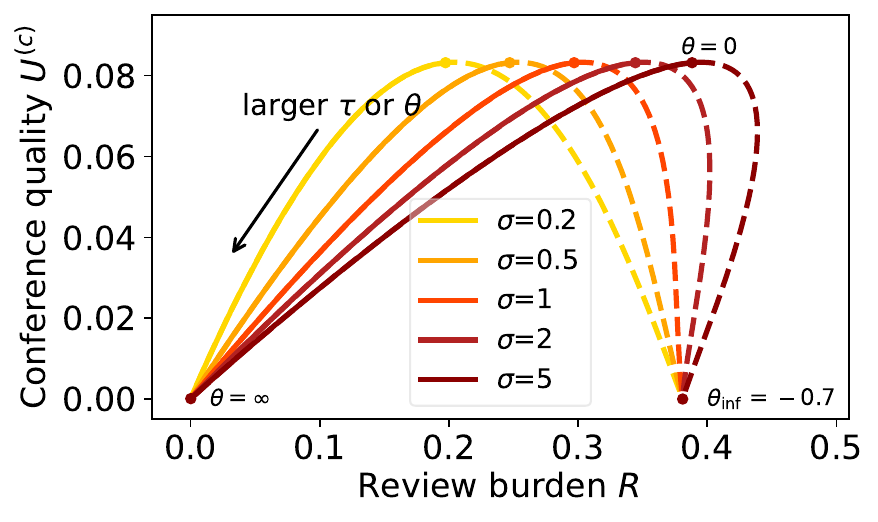}
         \captionsetup{size=}
         \caption{}
     \end{subfigure}
     \hfill
     \begin{subfigure}[b]{0.47\textwidth}
         \centering
         \includegraphics[width=\textwidth]{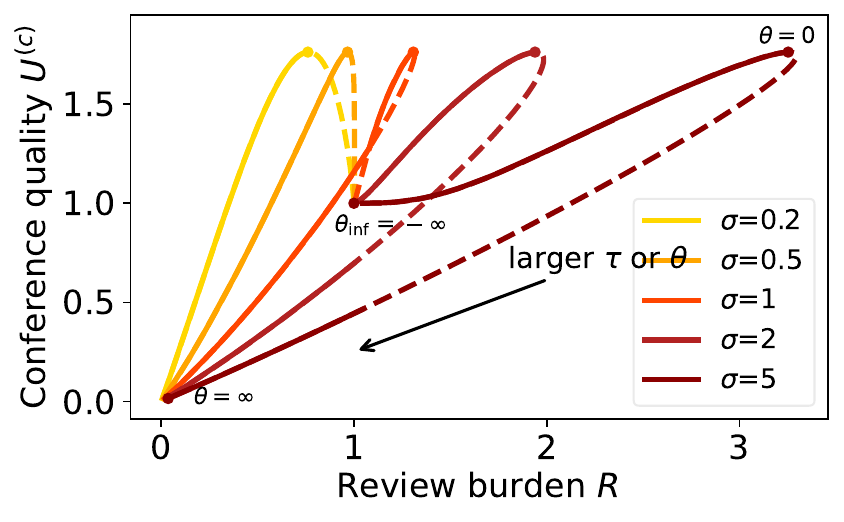}
         \captionsetup{size=}
         \caption{}
     \end{subfigure}
     \hfill
     }
     {QB-tradeoff Curves in the Continuous Model. \label{fig:Pareto_frontier_continuous}}
     {(a) shows QB-tradeoff curves of a \emph{$(\sigma, \mu_{\QualDist} = -1, \sigma_{\QualDist} = 1, \NumReviews = 1, \TD = .7)$-Double Gaussian model}.  (b) shows QB-tradeoff curves of a \emph{$(\sigma, \mu_{\QualDist} = .5, \sigma_{\QualDist} = 2, \NumReviews = 1, \TD = .7)$-Double Gaussian model}.  In each case, the review quality $\sigma$ is varied over five discrete options.  For each $\sigma$, the curve shows the possible QB-tradeoffs as the acceptance threshold is varied continuously. The Pareto frontier is shown with solid lines, while dominated points are shown with dashed lines. 
     % \yzedit{Note that the \fangedit{conference }quality-maximizing point on each curve corresponds to the case $\theta = 0$\fangedit{; these points are indicated by dots}.}\fangcomment{I guess the text: larger $\tau/\theta$ is not ratio.  may be we use $\tau$ or $\theta$.  Is it possible to present some $\tau$ or $\theta$ values in the figures?}
     }
     \end{figure}

\Cref{fig:Pareto_frontier_continuous} maps the QB-tradeoff for various settings. 
In each setting, there is a point at (0, 0) that corresponds to rejecting all submissions. As the threshold is decreased, high-quality papers start being submitted, increasing both the conference quality and review burden.  When the de facto threshold is 0, conference quality is maximized. Subsequently, a further decrease in the threshold leads to more low-quality papers being accepted, lowering the conference quality.  We see here that the effect on the review burden is mixed: sometimes it increases (e.g., the curves in \cref{fig:Pareto_frontier_continuous} (b) with $\sigma < 1$) while other times it decreases (e.g., the curves in \cref{fig:Pareto_frontier_continuous} (b) with $\sigma \ge 1$).  
As the threshold decreases further, 
% \dkreplace{we see a point where all submitted papers are accepted with one round of review.
% In this case, either all papers are submitted and accepted immediately, resulting in a review burden of 1 (e.g., panel (b)), or papers with quality above $\theta_{\inf}$ are submitted, resulting in a conference value $\ConfValue = 1$ and a review burden smaller than 1 (e.g., panel (a)).}{
two types of behaviors may emerge. If the expected paper quality under \QualDist is non-negative, it is strictly positive conditioned on being above any de facto threshold $\theta > -\infty$, so as the de facto threshold goes to $-\infty$, all papers are submitted and accepted in a single round of review. This behavior can be observed in \cref{fig:Pareto_frontier_continuous} (b), where the curves converge to a point with $R=1$. On the other hand, if the expected paper quality is negative, a sufficiently low de facto threshold $\theta$ leads to a conference value $\ConfValue < 1$, so no authors submit any more, and the QB-tradeoff becomes $(0,0)$. Thus, the lowest meaningful de facto threshold is $\theta = \theta_{\inf}$; this can be observed in \cref{fig:Pareto_frontier_continuous} (a).

% \dkcomment{Rewrote the preceding on 08/17.}

\begin{figure}[htb]
     \FIGURE
     {\includegraphics[width=0.6\textwidth]{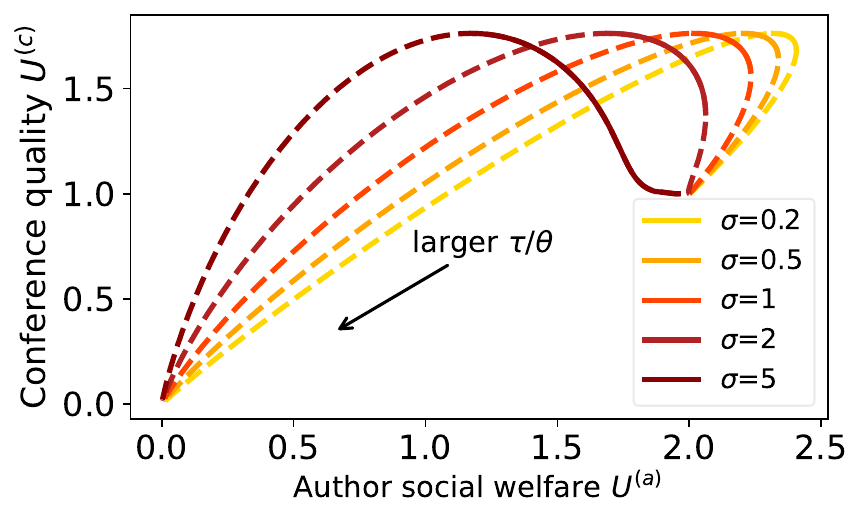}}
     {QA-tradeoff Curves in the Continuous Model. \label{fig:QA_Pareto_frontier_continuous}}
     {The figure shows QA-tradeoff curves of a \emph{$(\sigma, \mu_{\QualDist} = 1, \sigma_{\QualDist} = 3, \NumReviews = 1, \TD = .7)$-Double Gaussian model}. Similar to \cref{fig:Pareto_frontier_continuous}, the Pareto frontier is shown with solid lines, while dominated points are shown with dashed lines.}
\end{figure}

\Cref{fig:QA_Pareto_frontier_continuous} presents an example of QA-tradeoff curves under various review noise levels, where the paper quality prior has a mean of $1$ and a standard deviation of $3$.
The curve starts at the point where all papers are accepted after a single review round: since $\mu_{\QualDist} = 1$, the conference quality is 1, and author welfare equals the conference value, which is 2.
In this example, as the acceptance threshold increases, conference quality improves because authors with negative-quality papers opt not to submit. However, the effect on author welfare depends on the review noise. When the reviews are noisy (high $\sigma$), increasing the threshold primarily reduces the acceptance probability, leading to a monotonic decline in author welfare. When the reviews are accurate (low $\sigma$), the increased acceptance threshold may significantly increase the conference value, because it is much harder for low-quality papers to get accepted. The increased utility upon acceptance can outweigh the drop in acceptance rate and raise author welfare.
Yet, as the acceptance threshold continues to increase (in the left-hand side of \cref{fig:QA_Pareto_frontier_continuous}) so that the corresponding de facto threshold $\theta > 0$, both conference quality and author welfare eventually decline to 0.

\subsubsection{Dominating Acceptance Thresholds}

Using the preceding numerical examples, we examine which acceptance thresholds or de facto thresholds are QB- or QA-dominating, and how the model parameters influence the dominating thresholds.

% \gscomment{we might warn folks earlier that our discussion will center around whether, if you are willing to give up on perfect conference quality, for either review burden or author utility, will you will be better off letting in more borderline papers, or rejecting more borderline papers.  Both will reduce conference quality, but how will review burden (or author utility) react? }

\paragraph{QB-tradeoff curves.}
Notice that deviations from the quality-maximizing de facto threshold ($\theta = 0$) in \emph{either} direction could be QB-dominating. First, the conference can decrease the threshold to accept some negative-quality papers, in order to accept the positive-quality papers in fewer rounds; alternatively, the conference can increase the threshold to give up on some borderline papers with positive quality which might otherwise take a large number of rounds until acceptance.

Clearly, which intervals of strategies are Pareto optimal depends on the distribution of paper quality and the review noise. For example, if there is a substantially larger number of borderline papers with negative quality than positive quality, marginally lowering the threshold will both degrade conference quality and increase review burden, but marginally increasing the threshold will decrease the review burden, though still degrade conference quality.  The latter will be Pareto optimal while the former will not. This can be seen from panel (a) of \cref{fig:Pareto_frontier_continuous}: increasing $\tau$ or $\theta$ from the quality-maximizing point ($\theta = 0$) yields acceptance policies that are QB-dominating.

Furthermore, as observed from panel (b) of \cref{fig:Pareto_frontier_continuous}, Pareto-optimal tradeoffs are typically achieved by decreasing $\tau$ when $\theta = 0$ under high review noise, and by increasing $\tau$ when review noise is low. This stems from the fact that with high noise, increasing the de facto threshold $\theta$ by a fixed amount (or equivalently, changing the conference quality by a fixed amount), requires a larger increase in the acceptance threshold $\tau$ (\cref{prop:gap-invariant}). This leads to a significant drop in the acceptance probability for all papers, resulting in more rounds of resubmission and thus a larger review burden. In contrast, when review noise is low, a small change in $\tau$ suffices to induce the same increase in $\theta$ (and conference quality). This means that the marginal benefit of increasing $\tau$ outweighs the marginal cost: the gain from reducing review burden by forgoing borderline papers is greater than the loss from slightly lowering the acceptance probability for all papers.
% the effect of giving up borderline papers to save review burden more impactful than the effect of decreasing the acceptance probability of every paper. 

The preceding paragraph's insight can be expressed intuitively as follows: if the reviewing quality is low, then striving for high quality will not dissuade authors of low quality from submitting sufficiently, and the primary effect of imposing ``high standards'' will be to impose multiple rounds of review on most papers (good or bad), before ultimately they are all accepted anyway. Thus, the conference might as well admit that it cannot distinguish paper qualities well, and be lenient in accepting. If the review quality is high, however, things change: bad papers are sufficiently likely to be rejected that the deterrence effect may save the reviewers significant work.
% \dkcomment{Added this on 08/17. Is it sufficiently accurate and ``useful''? Does it belong here, or in the conclusions/discussion? Or should it just go?}
\gscomment{I like this observation.  Not sure this is possible, but could we move it to the front and then explain why.   Currently we are explaining why and then summarizing what is happening.}

\paragraph{QA-tradeoff curves.}
Different from the QB-tradeoff curves, where a sufficiently high acceptance threshold is never dominated, in QA-tradeoff curves, 
% the dominating regime \dkcomment{I don't understand the role of ``the dominating regime'' in this sentence.}, 
the dominating tradeoffs are usually achieved when $\theta\in [\ubar{\theta},0]$ for some negative value $\ubar{\theta}$ (see \cref{fig:QA_Pareto_frontier_continuous}). This means that at the point $\theta = 0$, lowering the acceptance threshold tends to be QA-dominating, primarily because authors prefer lenient acceptance policies.

Recall the suggestion discussed in \cref{subsec:paradox}: the conference can lower the acceptance threshold to accept more papers quickly. 
We have seen, both in the previous section and in \cref{fig:Pareto_frontier_continuous}, that this strategy does not necessarily reduce the review burden due to the resubmission gap. Here, we can see that sometimes it can also harm author welfare.
For example, in \cref{fig:QA_Pareto_frontier_continuous}, we can observe that the policies near the point (2,1) are typically dominated for curves with small $\sigma$. This is because the conference value is endogenous: the policy that trivially accepts all papers tends to have a low conference value, which harms the authors' utilities.

\subsection{Resubmission Gap and Dominating QB-tradeoffs}
\label{subsec:gap-tradeoffs}

Here, we present an intuitive connection between the key concept in \cref{sec:thresholds-gaps}, the resubmission gap, and the key concept in this section, the QB-tradeoff, using the continuous model.
Recall that fixing the number of reviews per paper, the resubmission gap depends on the quality prior, the review noise, the author's discount factor, and the de facto threshold. 
Let $\tau(\theta\mid \QualDist, \RevSigDist, \eta, m)$ be the acceptance threshold that induces a candidate threshold $\theta$ as a de facto threshold, which is a function of $\theta$ conditioned on other model parameters. We present the following result.

\begin{proposition}\label{prop:gap-QB-dominance}
    Consider the continuous model with a fixed number of reviews per paper $m$ and a fixed prior distribution of paper quality $\QualDist$.
    We compare two settings: one with review signal distribution and author-discount factor $(\RevSigDist, \eta)$ and another with $(\RevSigDist', \eta')$. If for every candidate threshold $\theta$, the corresponding resubmission gap is larger in the first setting than in the second, i.e., $\tau(\theta\mid \QualDist, \RevSigDist, \eta, m) >  \tau'(\theta\mid \QualDist, \RevSigDist', \eta', m)$, then the QB-tradeoff of the second setting weakly dominates that of the first.
\end{proposition}

We defer the proof to \cref{app:proof-gap-QB-dominance}.
Intuitively, a setting with a larger resubmission gap requires a larger acceptance threshold to induce the same de facto threshold, resulting in a lower acceptance probability for every submitted paper.
This thus leads to more rounds of resubmissions and increases the review burden.

In the next two subsections, we examine how a larger review noise and a larger author discount factor each enlarge the resubmission gap, thereby leading to a dominated QB-tradeoff. We also analyze how these parameters affect the QA-tradeoff. 
% \dkcomment{Should it be singular or plural?}

\subsection{Dominating Tradeoffs: Review Noise}
\label{sec:QB-trade-noise}

In comparing the different QB-tradeoff curves of \cref{fig:Pareto_frontier_continuous} and the QA-tradeoff curves of \cref{fig:QA_Pareto_frontier_continuous}, we observe that any curves corresponding to higher-quality (i.e., lower-variance) reviews dominate similar curves corresponding to lower-quality reviews.
We show that this is not a coincidence and holds not just for Gaussian noise in the reviews, but for Blackwell dominating review quality (defined in \cref{def:blackwell}).
Note that the former is a special case of the latter.
In other words, we show that a better review quality can simultaneously benefit all three stakeholders of the review system.
In this subsection, unless otherwise specified, tradeoffs without qualification refer to both QA- and QB-tradeoffs.
% Later in \cref{subsec:qa_tradeoff}, we further show that a better review quality also leads to dominating QA-tradeoffs\textemdash review policies with Blackwell dominating review quality can achieve the same conference quality at a higher author welfare.

The full story is a bit more subtle. Whether Blackwell-dominating reviews imply better tradeoffs depends on what space of acceptance policies the conference can optimize over. We show that if the conference has all memoryless acceptance policies available, then better reviews can always be used to simulate worse reviews, and the conference can thus obtain at least the same tradeoff. Therefore, better reviews weakly dominate worse reviews even when the signals do not satisfy the monotone likelihood ratio (MLR) property. 
However, if the conference is restricted to threshold policies and the reviews do not necessarily have MLR (\cref{def:informative}), carefully chosen ``worse'' reviews may actually permit the use of a better threshold policy, achieving a better tradeoff. 
However, such behavior is indeed the result of signals violating the MLR property: if the review signals have MLR, Blackwell dominance again implies a weakly better QB-tradeoff under threshold policies.

\begin{definition}[Blackwell Dominance \citep{bohnenblust1949reconnaissance,blackwell1953equivalent}]
\label{def:blackwell}
Let $\RevSigDist: \QualSet \times \SigSet \to [0,1]$ and $\RevSigDistP: \QualSet \times \SigSet' \to [0,1]$ be two review signal distributions.
$\RevSigDist$ \emph{Blackwell dominates} $\RevSigDistP$ if there exists a garbling $\gamma: \SigSet \times \SigSet' \to [0,1]$ from $\SigSet$ to $\SigSet'$, where for all $\REVSIG \in \SigSet$,
% \dkcomment{Replacing signal \emph{vector} with signals everywhere here. Please keep consistent with future edits.} 
$(\gamma(\REVSIG, \REVSIGP))_{\REVSIGP \in \SigSet'}$ is a distribution on $\SigSet'$, such that for all $\REVSIGP \in \SigSet'$ and all $q \in \QualSet$:
\begin{align*}
\RevSigProbP[q]{\REVSIGP} 
& = \sum_{\REVSIG \in \SigSet} \RevSigProb[q]{\REVSIG} \cdot \gamma(\REVSIG, \REVSIGP).
\end{align*}
\end{definition}

\subsubsection{General memoryless acceptance policies.}

We state the following proposition in the categorical model with $\NumReviews = 1$ review. We discuss the (straightforward) extension to the continuous model and multiple reviews below.

\begin{proposition}
\label{prop:blackwell}
Consider two settings with $\NumReviews = 1$ in the categorical model that are identical except for the review signal distributions (which need not have MLR): 
the distribution $\RevSigDist$ of the first setting Blackwell-dominates the distribution $\RevSigDistP$ of the second setting.
Then, over memoryless acceptance policies, the tradeoff in the first setting weakly dominates the tradeoff in the second. 
\end{proposition}

The proof is given in \cref{app:proof-blackwell}. At a high level, the proposition holds because if $\RevSigDist$ Blackwell-dominates $\RevSigDistP$, then in the setting with $\RevSigDist$, a policy $\ACCMAP$ can perform the garbling $\gamma$ from \cref{def:blackwell} itself and then apply $\ACCMAP[']$; if $\ACCMAP[']$ is monotone, one can show that so is the resulting $\ACCMAP$.
Thus, the tradeoff in the first setting must weakly dominate the second setting.

% \dkdeletecomment{Since we are not giving the proof here, this is meaningless?}{For the continuous setting, the proof can be modified by replacing summation with integration.}
When the two settings have $\NumReviews > 1$ reviews drawn independently from $\RevSigDist$ and $\RevSigDistP$, respectively, where $\RevSigDist$ Blackwell-dominates $\RevSigDistP$, we can use the fact that applying the same garbling independently in each dimension gives a garbling on the $\NumReviews$-dimensional signal vectors. Therefore, viewing the entire vector as just one signal, the distribution in the first setting Blackwell-dominates that in the second setting, and Proposition~\ref{prop:blackwell} applies directly. 
Performing this reduction relies on the fact that Proposition~\ref{prop:blackwell} did not require signals to have MLR. After all, the MLR property is defined only for scalar-valued signals.
Similarly to the case of \emph{better} reviews, when \emph{more} reviews are obtained in the first setting, and the reviews in both settings are drawn from the same distribution, the signal of the combined reviews in the first setting Blackwell-dominates the signal in the second setting: this is because discarding the additional signals is easily seen to be a garbling.

\subsubsection{Threshold acceptance policies.}
\label{subsec:threshold-BW-domi}

Proposition~\ref{prop:blackwell} shows that a better review quality (in the Blackwell sense) implies a better tradeoff if we allow the conference to apply any memoryless acceptance policy. This result even holds for review signals that do not necessarily have MLR. In the following proposition, we further show that even if the conference is restricted to applying threshold policies, the same result holds if the review signals  \emph{do} satisfy the monotone likelihood ratio property.

\begin{proposition} \label{prop:blackwell-threshold}
  Consider two settings with $\NumReviews = 1$ \fangcomment{continuous?}\yzcomment{this should hold for both model, even though the proof is written for the categorical model.} review that both satisfy the MLR property and are identical except for the review signal distributions: 
  the distribution $\RevSigDist$ of the first setting Blackwell-dominates the distribution $\RevSigDistP$ of the second setting.
  Then, over threshold acceptance policies, the QA- and QB-tradeoff curves in the first setting weakly dominate those in the second.
\end{proposition}

The following lemma is central to the proof of \cref{prop:blackwell-threshold}.

\begin{lemma}\label{claim:blackwell-RB-better}
    Consider two threshold acceptance policies $\ACCMAP$ and $\ACCMAP[']$ which accept papers of quality $\bar{q}$ with equal probability in the first and the second setting in \cref{prop:blackwell-threshold}, respectively.  
    Then, 
    % {Given two settings in \cref{prop:blackwell-threshold} with two threshold acceptance policies $\ACCMAP$ and $\ACCMAP[']$ respectively, if there exists $\bar{q}$ so that the probability of accepting papers of quality $\bar{q}$ are identical under both settings,} 
    % under the author's $\bar{q}$-threshold strategy, the review burden of $\ACCMAP$ in the first setting is no larger than the review burden of $\ACCMAP'$ in the second setting.  
    the acceptance probability of a paper of quality $q$ in the first setting is no less than that in the second setting, for any $q>\bar{q}$.
\end{lemma}

The intuition behind \cref{claim:blackwell-RB-better} is that a larger review noise enlarges the resubmission gap, requiring a larger acceptance threshold to induce the same de facto threshold. Consequently, in the setting with better review quality, every submitted paper has a higher acceptance probability.

% This lemma says that if the policies $\ACCMAP$ and $\ACCMAP[']$ induce indifference at the same quality threshold in the authors, then the first setting has a lower review burden than the second.
% This lemma says that if the policies $\ACCMAP$ and $\ACCMAP[']$ induce the same acceptance probability for a paper with quality $\bar{q}$ under two settings with different review quality, then every paper with a higher quality than $\bar{q}$ will be accepted with a larger probability under the setting with better review quality.
The proof of \cref{prop:blackwell-threshold}, which is deferred to \cref{app:proof-blackwell-threshold}, then follows by creating an appropriate $\ACCMAP$ from $\ACCMAP[']$ such that there exists a paper quality with the same acceptance probability in both settings and \cref{claim:blackwell-RB-better} can be applied.
The lemma then suggests that the first setting can accept every submitted paper with fewer rounds of resubmissions than the second setting, leading to a smaller review burden and a larger author welfare.
%to compare. 
% \yzdelete{Given a $\ACCMAP[']$ for the second setting, if there is a paper quality in the second setting that agents submit with a probability $r\in (0,1)$, this is pretty straightforward.  If there is no such paper quality, we make $\ACCMAP[']$ slightly less strict so that such a quality exists. 
% In this case, the modified $\ACCMAP[']$ has a smaller review burden\yzedit{ and a larger author welfare} than the old $\ACCMAP[']$ while maintaining the same conference quality. 
% Then, by \cref{claim:blackwell-RB-better}, we can show that there is a $\ACCMAP$ inducing the same threshold in the first setting which has a smaller review burden\yzedit{ and a larger author welfare} than the modified $\ACCMAP[']$.}
% \yzcomment{Removed the detailed explanation for the proof. I feel it's not clear at all, and the current intuition is good.}

While the result again generalizes from the categorical model to continuous signals, it does not generalize to $\NumReviews > 1$ signals. The reason is that it relies on signals having MLR, a property that is not preserved when combining signals (shown in \cref{ex:threshold-counterexample} in \cref{sec:testing_tradeoff}). 
In particular, we have numerically found counterexamples which suggest that even though the review signal distribution of one setting Blackwell-dominates the distribution of another setting (and both review signals have MLR), it is possible that after combining two independent signals in each setting, the tradeoff achieved by threshold policies in the first setting does not (weakly) dominate the tradeoff in the second setting. Our counterexamples use three types of paper qualities, and the signal set contains three signals. Unfortunately, these counterexamples are both complicated and unintuitive, so we omit them from this paper.
% The counter-examples were found by exhaustive computational search, and the code for this search is available at \url{https://github.com/yichiz97/Conference-Peer-Review}.
 
These examples also suggest that if the reviews do not have MLR, Blackwell dominance does not imply better tradeoffs under threshold acceptance policies.

\subsection{Dominating Tradeoffs: Discount Factor}
\label{sec:dominating-value-discount}

We now explore the effect of the author discount factor on QA- and QB-tradeoffs.
We show that more patient authors unequivocally lead to a worse QB-tradeoff. 
However, the effect on QA-tradeoff is mixed --- whether author welfare increases with $\eta$ depends on the quality prior, the review noise distribution, and the de facto threshold.
% Next, we elaborate on an observation made in \cref{sec:additive-noise}, namely, that more patient authors lead to a worse QB-tradeoff, because authors will be more persistent in resubmitting borderline papers.

\paragraph{QB-tradeoffs.} As we have seen in \cref{sec:additive-noise}, a larger $\eta$ \dkreplace{enlarges}{increases} the resubmission gap, \dkreplace{meaning}{implying} that authors are more persistent in resubmitting their rejected papers under the same condition. Therefore, the conference has to raise its acceptance threshold so as to maintain its quality. This then results in more resubmissions and thus leads to a dominated QB-tradeoff. We formalize this insight in the following proposition.
\begin{proposition} \label{lemma:query burden value}
% Consider two settings that are identical except that they have attractiveness factors $\rho$ and $\rho'<\rho$, respectively.
% Consider their corresponding QB-tradeoff curves as the acceptance threshold is varied from $-\infty$ to $\infty$.
% The QB-tradeoff curve of the setting with attractiveness factor $\rho'$ dominates the QB-tradeoff curve of the setting with $\rho$.
Consider two settings that are identical except that they have different author discount factors, $\eta > \eta'$. Then, the QB-tradeoff curve in the setting with discount factor $\eta'$ dominates the QB-tradeoff curve in the setting with $\eta$.
\end{proposition}

% \gscomment{Seems like this should mention that increasing the discount factor increases the resubmission gap, and that is what is driving this result.}

We defer the proof to \cref{app:proof-QB-tradeoff}. 
\cref{lemma:query burden value} suggests that having more patient authors will harm the QB-tradeoff, in the sense that it will be dominated by the original setting. 
% This result holds across all the models that we consider; it is also confirmed by our analysis based on real data in \cref{sec:noiseless-ICLR}.
It also cautions against certain peer review experiments that may increase the author discount factor --- for example, shortening the review cycle or eliminating the rebuttal phase.

\paragraph{QA-tradeoffs.}
% \gscomment{I am not sure if it is more complicated than in my mind, but can't we just say that there are two forces at work when the discount factor decreases: 1) the the resubmission gap increases; 2) that the authors utility for events in the future decrease.  The first of these means that the authors utility will increase because for the same defacto threshold, the acceptance threshold will be lower and thus every paper will be accepted sooner.  However, the second point is bad for every author that does not get his paper accepted in the first round.}
We now investigate how author welfare changes with the discount factor $\eta$. 
Intuitively, while increasing $\eta$ raises utility upon acceptance for a fixed number of resubmissions, it also widens the resubmission gap, forcing the conference to raise its threshold. This, in turn, increases the expected number of resubmissions, potentially reducing authors’ utility.

We formalize the above intuition in the continuous model.
Suppose that the de facto threshold is fixed at $\theta$, which implies fixed conference value $\ConfValue$ and conference quality $\CONFUTIL$.
The utility of an author with a paper of quality $q$ is then a function of the acceptance threshold $\tau$ (which depends on $\eta$), i.e., $u^{(a)}(q, \tau(\eta)) = \frac{\AccP{\tau(\eta)}{q} \cdot V}{1-\eta \cdot (1-\AccP{\tau(\eta)}{q})}$.
The acceptance threshold should make the authors with borderline papers indifferent between submitting and taking the outside option, i.e., $\tau(\eta)$ is the solution in $\tau$ to $\AccP{\tau}{\theta} = \frac{1-\eta}{V-\eta}$ with $V$ being a constant.
Writing the author's utility as a function of $q$ and $\eta$, we are primarily interested in the conditions under which $u^{(a)}(q, \eta) = \frac{\AccP{\tau(\eta)}{q} \cdot V}{1-\eta \cdot (1-\AccP{\tau(\eta)}{q})}$ is increasing in $\eta$.
% \fangcomment{instead of the following local characterization, can we have a figure showing $u^{(a)}$ under a fixed de facto threshold?}
% \yzcomment{I tried. The figures don't look pretty. the curves for most of $u_a$ are flat, with one curve having a very large gradient. I slightly prefer the current form.}

\begin{proposition}\label{lem:QA_eta}
Consider the continuous model with a fixed candidate threshold $\theta$ and a discount factor $\eta$. Let $\tau$ be the acceptance threshold that induces $\theta$ as the de facto threshold, as defined in \cref{prop:gap-invariant}. Then, for any paper of quality $q \ge \theta$, the author's marginal utility with respect to $\eta$ is positive
% \fangcomment{increasing for all $\eta\in (0,1)$?  Based on the proof, it is only increasing locally: there exists a small enough $\epsilon$ so that $u^{(a)}(q,\eta)>u^{(a)}(u,\eta')$ if $\eta<\eta'<\eta+\epsilon$} 
if $h(q) < h(\theta)$, negative if $h(q) > h(\theta)$, and zero if $h(q) = h(\theta)$, where
\[
h(q) = \frac{f^{(r)}(\tau - q)}{F^{(r)}(\tau - q) \cdot \left(1 - F^{(r)}(\tau - q)\right)},
\]
and $f^{(r)}$ and $F^{(r)}$ are the pdf and cdf of the review noise distribution, respectively.
\end{proposition}

% \begin{lemma}\label{lem:QA_eta}
% In the continuous model, fix a feasible de facto threshold $\theta$ and a discount factor $\eta$. Let $\tau$ be the acceptance threshold corresponding to $\theta$ as given by \cref{prop:gap-invariant}. Then, the utility of an author with a paper of quality $q$ increases, decrease, or doesn't change with $\eta$ if 
% $h(q) = \frac{f^{(r)}(\tau-q)}{\REVNOISEDIST(\tau-q)\left(1-\REVNOISEDIST(\tau-q)\right)} <, =, \text{or } > h(\theta)$. 
% \end{lemma}
We defer the proof of \cref{lem:QA_eta} to \cref{app:proof-QA-eta}.
The proposition highlights that whether an author’s utility increases with the discount factor $\eta$ depends on the behavior of the function $h(q)$ for $q > \theta$, which in turn depends on the shape of the review noise distribution. 
Notably, even for standard noise models such as the normal distribution, there is no general guarantee that $h(q) > h(\theta)$ or $h(q) < h(\theta)$, meaning that an author's utility is in general non-monotonic in $\eta$.

For example, when $f^{(r)}$ follows a zero-mean normal distribution, $h(q)$ is symmetric at $\tau$ and convex, with $h(q) \to \infty$ as $q \to \pm\infty$, and minimized at $q = \tau$. 
This implies that when $\tau > \theta$\textemdash as is typically the case\textemdash there exists a threshold $\bar{q} > \theta$ such that author utility increases with $\eta$ for $\theta < q < \bar{q}$, but decreases for $q > \bar{q}$. 
Intuitively, borderline papers near the de facto threshold benefit more from a higher discount factor as they usually experience more rounds of resubmissions.
In contrast, authors with high-quality papers usually experience fewer rounds of rejections and thus can benefit less from an increase in $\eta$. 
Their utilities decrease in $\eta$ because the conference has to raise $\tau$ so as to preserve the conference quality in response to the increase in $\eta$.

% This pattern is typically unique for single-peak thin-tailed noise distributions. 
However, for heavy-tailed review noise distributions such as the Cauchy distribution, the function $h(q)$ behaves differently: it increases from $0$ to its maximum value as $q$ increases toward $\tau$, and then decreases back to $0$ as $q \to \infty$. As a result, the pattern is reversed: authors’ marginal utility is positive in $\eta$ for high-quality papers, but is negative in $\eta$ for papers with relatively lower quality.
Intuitively, under heavy-tailed review noise, even high-quality papers have a non-trivial probability of being rejected in a single round of submission.

Because author welfare is the integral of individual author utilities across the quality distribution, it can both increase or decrease with $\eta$, depending on the review noise distribution and the paper quality prior. Therefore, attempts that aim to reduce the cost of resubmissions not only exacerbate the review burden but can also, in some cases, reduce overall author welfare.

\subsection{Acceptance Rate}\label{sec:acc_rate}

One may suspect that the higher the threshold, the more selective the conference, so the lower the acceptance rate.
But this is not always true. The reason is self-selection by authors of weaker papers, who may not submit in the first place.  As a result, those papers will not be rejected. 

We first develop some mathematical tools to help us reason about the interaction between the selectivity of the conference (the de facto threshold) and the acceptance rate.

Consider the continuous model.
Let $\tau$ be a non-trivial acceptance threshold, and $\theta$ the corresponding de facto threshold. As before, let $\AccP{\tau}{q}$ be the probability that a paper of quality $\Qual=q$ is accepted at the conference.  
In round $t$, the total resubmission ``density''
of papers with quality equal to $q\ge \theta$ is equal to 
$\QualDens{q}\cdot \sum_{j=0}^{t} (1-\AccP{\tau}{q})^{t-j}.$
As $t$ gets larger, 
% \footnote{Taking $t$ large enables us to take into account the full history of previously rejected papers that are resubmitted.} 
% \yichicomment{Removed a footnote here.}
this converges to $\QualDens{q}/\AccP{\tau}{q}$. 
Of these papers, a $\AccP{\tau}{q}$ fraction will be accepted in each round.
Hence, the acceptance rate converges to
\begin{equation}\label{eq:accept-rate}
    \AccRate = \frac{\text{Number of papers accepted this round}}{\text{Number of papers submitted this round}}  = \frac{\int_{\theta}^{\infty}\QualDens{q} dq}{%
        \int_{\theta}^{\infty} \QualDens{q}/\AccP{\tau}{q} \; dq}. 
\end{equation}

We now use \cref{eq:accept-rate} to intuitively reason about how the prior distribution $\QualDist$ affects the acceptance rate $\AccRate$.  Notice that it is the papers with low acceptance probabilities that disproportionally decrease the acceptance rate.  This is because they add only their mass to the numerator, but add their mass scaled by $1/\AccP{\tau}{q}$ to the denominator.  Thus, intuitively, for papers with quality at least $\theta$, if a $z$ fraction are borderline with quality  ``near'' $\theta$, and the other $1-z$ fraction has a very high acceptance probability,  the acceptance rate can be approximated by $\frac{1}{z/\AccP{\tau}{\theta} + (1 - z)}$.  Notice that this quantity is decreasing with $z$: the larger $z$, the smaller the acceptance rate.   

The measurement of $z$ intuitively resembles the \emph{hazard rate} of a distribution\fangcomment{Can I say the hazard rate of $\QualDist$?}, which is defined as $\frac{f(x)}{1-F(x)}$, where $f$ is the probability density function, and $F$ is the cumulative distribution function.  Similar to $z$, the hazard rate measures the probability of a paper on a boundary at $x$, $f(x)$, relative to the mass of papers larger than $x$, $1 - F(x)$.   The hazard rate is known to be monotone for thin-tailed distributions, like the Gaussian distribution. 
% Conversely,  the hazard rate is known to be non-monotone for heavy-tailed distributions, like the Cauchy distribution.  Finally, the hazard rate is known to be (eventually) constant for the Laplace distribution.  

The acceptance rate also depends on the acceptance probability of the borderline papers, $\AccP{\tau}{\theta}$. Based on \cref{lem:author_response}, this probability must be equal to $1/\rho(\theta)$, where the attractiveness factor $\rho$ is a function of $\theta$. As we argued before, $\rho$ is increasing in $\theta$, which means that $\AccP{\tau}{\theta}$ is decreasing in $\theta$. Therefore, the effect of $\AccP{\tau}{\theta}$ on the acceptance rate is monotone: fixing $z$, a larger acceptance probability of borderline papers always leads to a larger acceptance rate of the conference.

Using the above intuition, we might expect the acceptance rates to decrease as $\theta$ increases for the Gaussian prior.  This is because Gaussian distributions have a monotone (increasing) hazard rate; thus, both $z$ and $\AccP{\tau}{\theta}$ change in the direction to drive down the acceptance rate as $\theta$ increases.  
However, we might expect the acceptance rates to increase for quality distributions with non-monotone hazard rate. This is because for a quality prior with a non-monotone hazard rate, $z$ may decrease at some point and drive up the acceptance rate, since raising the acceptance threshold can substantially reduce the proportion of borderline papers among all submissions.

\cref{fig:Acc_rate_continuous} gives evidence to support this intuition: we observe that the acceptance rate is non-monotone in $\theta$ when the quality prior is a two-peaked Gaussian distribution (panel (c)). 
However, a non-monotone hazard rate in the quality prior does not necessarily imply a non-monotone acceptance rate. As illustrated in panel (b), the Cauchy prior, despite having a non-monotone hazard rate, yields a monotonically decreasing acceptance rate.

We also observe that compared with settings with smaller $\sigma$, settings with noisier review quality tend to have a monotone decreasing acceptance rate. This is because by \cref{prop:gap-invariant}, a larger review noise leads to a larger resubmission gap. Therefore, with the same increase in $\theta$, a larger rise in $\tau$ is required, and thus borderline papers are rejected with a larger probability.  
\gscomment{I don't think we have the energy/time to do this now, but I think this is a very nice insight if it formallizes.  In particular, this seems to be driving the results of section 4.1.   Overall, connecting observations to a few technical drivers helps to organize the results.} 
In other words, the derivative of $\AccP{\tau}{\theta}$ in $\theta$ is larger when the review is noisier. This enhances the original intuition that a more selective conference has a lower acceptance rate.

\begin{figure}[htb]
     \FIGURE
     {\begin{subfigure}[b]{0.32\textwidth}
         \centering
         \includegraphics[width=\textwidth]{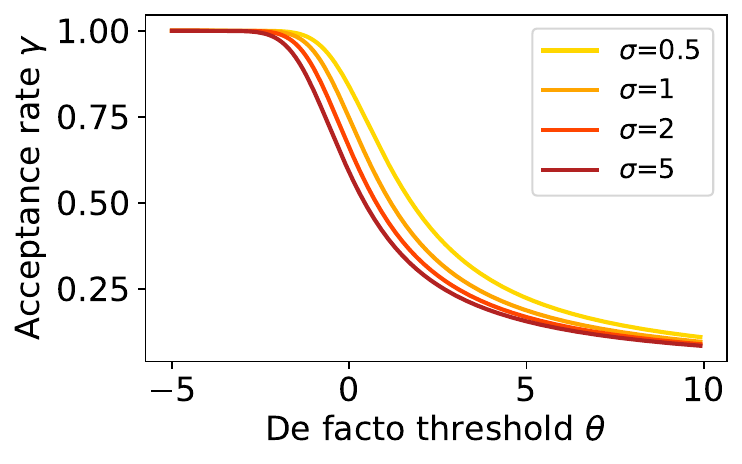}
         \captionsetup{size=}
         \caption{Gaussian prior.}
     \end{subfigure}
     \hfill
     \begin{subfigure}[b]{0.32\textwidth}
         \centering
         \includegraphics[width=\textwidth]{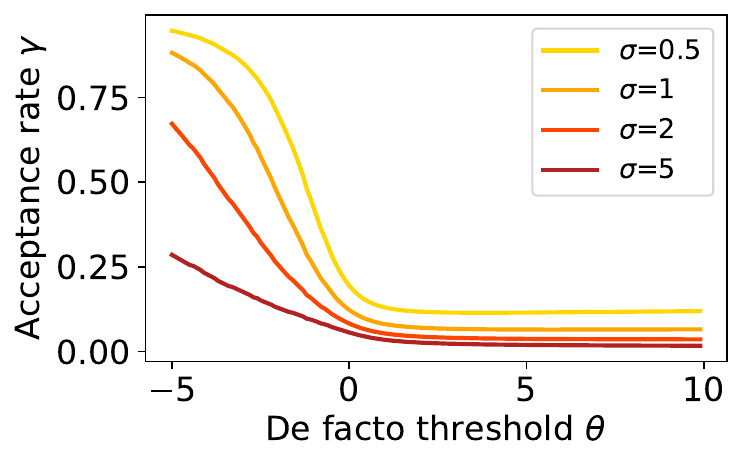}
         \captionsetup{size=}
         \caption{Cauchy prior.}
     \end{subfigure}
     \hfill
     \begin{subfigure}[b]{0.32\textwidth}
         \centering
         \includegraphics[width=\textwidth]{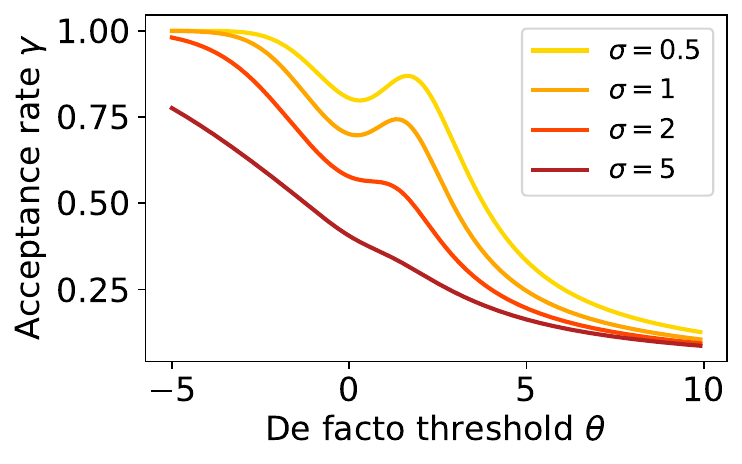}
         \captionsetup{size=}
         \caption{Mixture Gaussian prior.}
     \end{subfigure}
     \hfill
     }
     {The Acceptance Rate vs.~De Facto Threshold under Different Prior Distributions. \label{fig:Acc_rate_continuous}}
     {The noise distribution $\REVNOISEDIST$ is fixed as a normal distribution with zero mean and standard deviation of $1$. Three types of prior distributions of the paper quality are considered: (a) the Normal distribution with $\mu_q = 0$ and $\sigma_q = 1$; 
     (b) the Cauchy distribution with $\mu_q = 0$ and $\sigma_q = 1$;
     (c) the mixture of two Normal distributions with $\lambda = 0.5$, $\mu_{q} = 0$, $\mu_q' = 4$, and $\sigma_q = \sigma_q' = 1$.}
\end{figure}

Another interesting observation is that the quality distribution of submitted papers is not a good reflection of the prior quality distribution of papers, even conditional on being above the de facto threshold.  The reason is that papers nearer the de facto threshold need to be submitted more times (on average) before being accepted than higher-quality papers.  Therefore, they are over-represented among the submitted papers. This aligns well with many reviewers' observation in the real world that many of their assigned papers seem to be borderline. However, by \cref{eq:accept-rate}, the quality distribution of accepted papers is an accurate reflection of the prior quality distribution of papers conditional on being above the de facto threshold.

%% file: OR-resbumission/sections/Conclusion.tex
\subsection{Implications}\label{sec:implications}
We suggest some possible interpretations and lessons that might be learned from our analysis.

\emph{Resubmission Gap.}
Perhaps the cleanest result concerns the resubmission gap. If a conference operates like our model with noiseless authors, then we would observe the following:  \emph{Every paper ever submitted to the prestigious conference is eventually accepted; however, at any given conference, many papers will be rejected, and thus need to be submitted multiple times.}  With many parameters, the vast majority of papers are rejected. 
On the surface, this sounds like a dystopian bureaucracy. 

The frequently proposed and obvious reaction is to accept all acceptable papers the first time, without making them resubmit multiple times.  This would allegedly decrease the reviewing load and increase author welfare (because each paper would only be reviewed once), without affecting the quality of the conference because all of the papers were going to be accepted anyway. Superficially, this seems very reasonable.

However, our model warns that this is an unlikely outcome.  Instead, by lowering its acceptance threshold, the conference would also lower its de facto threshold.  While the papers currently being submitted could be overwhelmingly accepted in one round, it would invite more, lower-quality submissions.
These lower-quality submissions would themselves be repeatedly submitted, and so the review burden would not necessarily decrease, and could even increase.
Furthermore, as these lower-quality papers dilute the conference value, author welfare may decline, as illustrated in \cref{fig:QA_Pareto_frontier_continuous}.

Of course, this is not a perfect reflection of reality. In particular, because authors are not all aware of their papers' qualities, some submissions are of low quality and are very unlikely to ever be accepted.  However, the experience of this article's authors is that for some prestigious conferences, the situation is similar to a large extent: the pool of submitted papers has been self-selected to those that will eventually appear in a good venue. 

Additional relevant modeling ``blind spots'' are that the review noise for some papers may be different from others, agents may have different levels of patience and different utilities for their outside options, and the distribution of paper qualities may react to the acceptance policy.  However, it is not clear that any of these model limitations fundamentally challenge the insight that there is a gap between the quality of papers implied by the threshold acceptance policy and the types of papers submitted and eventually accepted.

\emph{Threshold Policy and tradeoffs.}
We investigated how the conference's acceptance policy influences the tradeoffs faced by stakeholders in the review system, with particular emphasis on how conference quality trades off against review burden and author welfare.
In the preceding discussion of the resubmission gap, it was not clear what happens to the review burden as the acceptance policy becomes more strict or lenient.  This is because it depends on the prior distribution of paper qualities and the review quality.  This effect was studied in \cref{subsec:QB-tradeoff}, where we observe the following. First, extremely stringent policies are always Pareto optimal, since the review burden is minimized when the acceptance threshold is set to infinity. Second, lenient policies tend to be Pareto optimal when reviews are highly noisy and the prior expected quality is high. In such cases, the conference benefits from accepting more papers sooner, which corresponds to a low acceptance threshold.
% \dkcomment{Did we explicitly discuss/name/describe them this way?}\yzcomment{No, how confusing is this?}\dkcomment{I found it confusing.}\gscomment{I think the "3"types is from the previous write up.  We don't have the three types in the figure anymore, so this does not make any sense.  Currently, it is a bit of a mess.  But we identify two dicotomies.  First, if the expected quality is positive, then if we accept all papers, people still submit.  Second, is whether, on the margins, increases or decreasing the acceptance give better review burden for a similar marginal decrease in conference quality.}
% (1) accept only a few top papers; (2) accept all worthy papers; (3) accept nearly all papers, only attempting to weed out the worst. 
% \yichiedit{only accept positive-quality papers by setting a relatively high acceptance threshold, or accept some or even all negative-quality papers by setting a low acceptance threshold.}
All of these policies can be easily identified in practice.

% The advantage of policies (1) and (3) is that they tend to require less review.  For policy (1), this is because few papers are submitted; for policy (3), it is because few papers are submitted more than once. The advantage of policy (2) is that it maximizes the conference quality, though often at the expense of a high review burden.
% In our analysis, however, policy (3) is Pareto optimal only if the prior of paper qualities contains mostly positive papers.  If, instead, the paper quality prior is a unimodal distribution centered at 0, then such a policy with a negative de facto threshold will never be Pareto optimal.  
% We note that in the data learned from ICLR (displayed in \Cref{fig:tradeoff_categorical}), all the acceptance policies which accept negative-utility papers are Pareto-dominated. (However, this is not always the case for the models learned from ICLR data for different years or parameter settings). 

\emph{Discount Factor.}
As shown in \cref{sec:dominating-value-discount}, in the noiseless author setting, increasing the discount factor creates a strictly worse QB-tradeoff curve. This implies that an intervention that decreases the discount factor and burdens the authors may improve the QB-tradeoff.  Examples include long review times, rebuttal periods, or onerous formatting requirements. A smaller discount factor for resubmission will intuitively allow the conference to decrease its acceptance threshold while keeping the same de facto threshold, and thus may decrease the review load without impacting conference quality. Essentially, such interventions artificially internalize the negative externality of imposing reviews upon others.
In fact, such impositions on the authors may even increase the authors' utilities, though this outcome is far from universal.

Conversely, our model predicts that well-meaning efforts to reduce the resubmission burdens may very well worsen the QB-tradeoff, and may eventually harm the authors' utility in some cases.  Proposed and executed reforms include: decreasing the required time to review papers for a given conference, ``fast-track'' resubmissions of papers recently rejected with sufficiently high scores, and a ``desk reject'' phase, where papers that appear subpar are quickly returned to authors without review.  These reforms artificially increase the time discount $\TD$ by decreasing the time between submissions.  This forces the conference to increase its acceptance policy threshold if it would like to maintain the same de facto threshold (and thus maintain conference quality). Consequently, the overall review burden is increased and, in some instances, author welfare is harmed as well.

% \gsedit{Conversely, conferences may put in additional hurdles for authors with little or no direct benefit to the conference.  While our model predicts that these may reduce the number of reviews while keeping a constant conference quality, they will sometimes, though not always, hurt author welfare.}
% Of course, it should be emphasized that these observations are not to be taken as recommendations: in particular, as we will discuss in \cref{sec:limitations}, our analysis essentially ignores the authors' utilities, and a longer time until acceptance or more burden to submit leads to a decrease of this utility. The tradeoff should naturally include all three concerns, and our observations should be interpreted as emphasizing one aspect that may not have been considered enough in the past.

\emph{Acceptance Rate.}
In \cref{sec:acc_rate}, we showed that whether the acceptance rate increases, decreases, or remains steady as the acceptance threshold increases depends on the paper quality prior distribution, in particular, a quantity resembling the hazard rate of the quality prior.
This warns against using the acceptance rate as a signal of quality. 
At issue is that for certain priors of paper quality, as the selectivity increases, the fraction of papers near the boundary may increase, leading to a larger acceptance rate.  

Additional factors may complicate this picture. For example, higher-quality papers may have different distributions of review noise than lower-quality papers.

\emph{Quality vs.~Quantity of Reviews.}
Our empirical results here discourage the strategy of soliciting a large number of reviews per paper. As shown in \cref{sec:noisy_abm}, any number of solicited reviews larger than $3$ greatly burdens the review system but is unlikely to bring enough benefits to the conference quality. 
% \dkcomment{Again, this result is only presented in the appendix, so a reader may lack context here. For these kind of results, we might explicitly point to the appendix, \`{a} la ``As we show in Appendix ..., any number of solicited ...''.}
Instead, our model predicts that a small number of solicited reviews, even with one review per paper, can be optimal if the conference is able to find the optimal acceptance threshold. The intuition is that when authors know the quality of their papers well enough, any number of solicited reviews, as long as it is combined with the optimal acceptance threshold, can take advantage of authors' self-selection such that only the desired papers are submitted, and eventually accepted. 

\emph{Institutional Memory.}
Our results in \cref{sec:memory_policy} indicate that having historical reviews follow submissions, or allowing the conference to limit the number of times a paper can be submitted, can help improve the QB-tradeoff (\cref{fig:memory_policies}). However, given that the improvement in the maximum conference quality is rather marginal, its main effect is to reduce the review burden. The intuition is that with such a policy, the conference can be strict in the first few rounds and relax the acceptance threshold for repeated resubmissions so that only the good papers will be submitted and accepted more quickly.

In summary, our results provide the following insight: the design of the acceptance threshold should consider the de facto threshold it may induce. Being aware of the authors' best response, the conference can optimally set its acceptance threshold to achieve a near-optimal quality with a small review burden, even with a small number of solicited reviews per paper and a simple memoryless acceptance policy.
% For any de facto threshold a conference would like to implement, a first-order concern is what acceptance policy to employ. 

\subsection{Limitations and Future Directions} \label{sec:limitations}

\emph{A Single Prestigious Conference.}
One limitation of this work is that we assume a single prestigious conference.  As mentioned in \cref{sec:model}, this can model several prestigious conferences that are more or less cooperating to uphold community standards.

Of course, in reality, there is an ecosystem of conferences, and not all of them are either top-tier or a side option. In such a setting, our analysis could model the decision to submit to a top-tier or second-tier conference. The utility for submitting to the second tier could be normalized to 1. The issue with this, however, is that the second-tier conference still needs to review its submissions. Furthermore, having multiple outside options can increase heterogeneity in authors’ utilities of outside options --- authors with higher-quality papers may have more attractive alternatives. This heterogeneity could lead to more complex, non-threshold author best responses.
There are other analyses (see related work in \cref{sec:related-work}) that have considered venues of different values. 

Furthermore, our model omits competition between conferences.  For example, conferences may compete to attract more papers by attempting to increase their quality, making their acceptance policy more predictable, creating a faster turnaround time, etc. 
Future work could extend our model to these settings. 

\emph{Heterogeneity.}
We also did not model various sources of heterogeneity in the process. For example, the qualities of reviews are not uniform, and different authors have different levels of patience (e.g., a Postdoc who will be on the job market vs.~a first-year student or a tenured faculty member). We are also not modeling the effects of biases that may impact different researchers disproportionately.  The uneven impact on different author populations would be made more complex by co-authorship.  As with the previously mentioned endogenous review quality, a main difficulty would be that this model would have more parameters, and as such require learning/setting them, which could lead to arbitrary choices.

\emph{Additional Feedback Loops.}
There are several feedback loops that we disregard.  We model paper quality as exogenous; however, in reality, it is largely determined by authors who decide how much effort or time to expend improving their papers.  Authors often write a paper to target particular venues. As the venues change their policies, the underlying distribution of paper qualities is likely to change.  Moreover, authors may improve their papers in response to reviewer feedback. Additionally, in our model, the review burden does not impact the quality of the reviews or the amount of time authors (who are typically also the reviewers) spend writing papers. 

As mentioned in the discussion of related work, several past papers do try to model these complexities and provide insights with agent-based simulations \citep{kovanis2016complex, bianchi2018peer, squazzoni2012saint}.

Here, we focus on a simpler and cleaner model than can be afforded when including these complexities.  Apart from the analysis being more difficult, it is often difficult to know precisely how these feedback loops function, which may lead to even greater uncertainty regarding the accuracy of a model, and any insights derived from it.  

\emph{Modeling Paper Quality.}
Finally, it should be noted that papers do not have an ``objective'' quality that can be projected to a single dimension or even multiple dimensions.  One approach, which adds minimal complexity along these lines, is to distinguish between different types of poor-quality papers.  Some papers may be deemed low quality because they are methodologically flawed;  others because their contributions may be incremental.  The first of these might be easier for reviews to detect and agree on than the second.

\subsection{Broader Impact and Conclusion} 

Despite the fact that there may be no agreed-upon objective metric, still, the rigor of peer review is important to a healthy future of academic research.
We avoid concrete recommendations, as many quantities and observed trends depend on model parameters and on unmodeled real-life effects. We hope that our theory (and simulations) can steer the discussion, uncover parameters to focus on, and inform decision makers. We believe that the focus on the resubmission gap, and the importance of reviewing quality over quantity, are important points to start a discussion in the community which may not have been as easily identified without studying a model like ours. 
We thus envisage our findings not merely as an academic contribution, but as catalysts for changing the conversation concerning conferences' acceptance and review policies.
Even more broadly, the impact of our work can reach beyond the context of conference peer review, finding relevance across diverse applications, including grant reviewing,  school admissions, and company recruitment.

%% file: OR-resbumission/sections/Proofs.tex
\subsection{Proof of Proposition~\ref{prop:monotone-prob}}\label{app:proof-monotone-prob}

In order to prove the proposition, we first note the well-known fact [\cite{whitt1979note,milgrom1981good,shaked2007stochastic}] that monotone likelihood ratio (MLR) implies first-order stochastic dominance (FOSD) of the distribution of the signal conditioned on a higher parameter (as well as for the posterior distribution of the parameter conditioned on a higher signal).

\begin{lemma} \label{lem:FOSD}
Assume that the family of review signal distributions has MLR.
Then, whenever $q' > q$, the signal distribution for $q'$ first-order stochastically dominates the distribution for $q$; that is, the distributions satisfy that $\Prob[\RevSig{} \sim {\RevSigDist[q']}]{\REVSIG \geq x} > \Prob[\REVSIG \sim {\RevSigDist[q]}]{\REVSIG \geq x}$ for all $x \in (\inf \SigSet, \sup \SigSet)$. 
\end{lemma}

We mentioned above that a higher signal also implies FOSD of the posterior quality distributions. The following lemma captures the stronger property that under the MLR property, this holds even for vectors of signals.

\begin{lemma} \label{lem:monotone_expected_quality}
Suppose $\RevSigV'$ and $\RevSigV$ are two vectors of signals that have MLR and satisfy $\RevSigV'\ge \RevSigV$ component-wise, and the inequality is strict for at least one of the components. 
Then, $U(\RevSigV')> U(\RevSigV)$ holds for any prior $\QualDist$.
\end{lemma}

For continuous distributions, this lemma is proved by \citet{torres2005multivariate}.
We give a self-contained proof for the categorical case, which is largely analogous, in \cref{app:FOSD-proof}.
We are now ready to prove Proposition~\ref{prop:monotone-prob}.

\proof{Proof of Proposition~\ref{prop:monotone-prob}.}
We give the proof in the categorical model; it can be straightforwardly generalized to the continuous model.

Let $q' > q$. 
For each reviewer $i$, couple the draws of $\RevSig{i}$ from $\RevSigDist[q]$ and $\RevSigP{i}$ from $\RevSigDist[q']$ by drawing a (common) uniformly random quantile $x$ in $[0,1]$ and letting $\RevSig{i}, \RevSigP{i}$ be the respective signals at quantile $x$ of the corresponding CDFs. 
Because, by \cref{lem:FOSD}, $\RevSigDist[q']$ (strictly) first-order stochastically dominates $\RevSigDist[q]$, this coupling ensures that $\RevSigP{i} \ge \RevSig{i}$;
furthermore, the inequality is strict with positive probability unless $\RevSig{i}=\max_s\in\SigSet$.
By applying this coupling to each individual review (which, recall, is drawn independently of other reviews), we obtain a coupling of vectors of reviews such that $\RevSigVP \geq \RevSigV$ always holds component-wise, and the inequality is strict for at least one of the components with positive probability. By Lemma~\ref{lem:monotone_expected_quality}, this coupling has the property that $U(\RevSigVP) \geq U(\RevSigV)$ always holds, and, conditional on \RevSigV, the inequality is strict with positive probability unless every component of $\RevSigV$ is the maximum signal (if the maximum exists).
Therefore, if \ACCMAP is a monotone acceptance policy, $\AccP{\ACCMAP}{q}$ can never decrease in $q$.

It remains to show that $\AccP{\ACCMAP}{q}$ is \emph{strictly} increasing in $q$ for non-trivial threshold policies $\ACCMAP[\tau,r]$.
First, we may assume w.l.o.g.~that $0 < r < 1$.
For if $r=0$, the policy is equivalent to the policy $\ACCMAP[\tau',\half]$ with any $\tau' \in (\max \Set{U(\RevSigV)}{U(\RevSigV) < \tau}, \tau)$, 
and if $r=1$, it is equivalent to the policy $\ACCMAP[\tau',\half]$ with any $\tau' \in (\tau, \min \Set{U(\RevSigV)}{U(\RevSigV) > \tau})$. Here, the minimum and maximum will be finite because the policy is assumed to be non-trivial.

By Lemma~\ref{lem:monotone_expected_quality}, there must exist \RevSigV, \RevSigVH with $U(\RevSigVH) \geq \tau \geq U(\RevSigV)$ such that at least one of the two inequalities is strict. Let \RevSigV be a vector of reviews maximizing $U(\RevSigV)$ subject to $U(\RevSigV) \leq \tau$.
Because $U(\RevSigVH) > U(\RevSigV)$, the vector \RevSigV cannot be maximal in all components.
Therefore, by the preceding coupling argument, when $\RevSigV$ is drawn with quality $q$, the corresponding vector $\RevSigVP$ drawn with quality $q'$ satisfies $U(\RevSigVP) > U(\RevSigV)$. 
By definition of \RevSigV, the review vector \RevSigVP gives rise to strictly higher acceptance probability than \RevSigV. For if $U(\RevSigV) < \tau$, then \RevSigV always leads to rejection, whereas (by maximality of $U(\RevSigV)$) \RevSigVP leads to acceptance with probability at least $r > 0$.
And if $U(\RevSigV) = \tau$, then \RevSigV leads to acceptance with probability $r < 1$, whereas \RevSigVP leads to acceptance with probability 1.

% \gs{I am having trouble with this proof.  It seems to require that "$\RevSigV$ is drawn with quality $q$ and the corresponding vector $\RevSigVP$ drawn with quality $q'$" with some positive probability, but I don't see where we show that.}
% \dkcomment{I thought that this is because we couple by quantile. You draw $s$. Each component has a quantile. You draw $s'$ in that component as having the same quantile according to the distribution parametrized with $q'$. I don't see how that leads to problems.}\fangcomment{I agree with David. We can first sample a quantile that determines $s$ then $s'$ that ensures $U(s')\ge U(s)$. 
%  However, it is unclear to me why $U(s')>U(s)$ instead of $U(s')\ge U(s)$ and strict with positive probability (but this should not be an issue).  This may be stated in the second paragraph.}

Therefore, for non-trivial threshold policies, the coupling ensures that a paper of quality $q'$ is accepted at least whenever a paper of quality $q$ is accepted, and with strictly positive probability, only the paper with quality $q'$ is accepted.
This completes the proof. \Halmos
\endproof

\subsection{Proof of Lemma~\ref{lem:monotone_expected_quality}}\label{app:FOSD-proof}

Here, we prove Lemma~\ref{lem:monotone_expected_quality}. 

\proof{Proof of \Cref{lem:monotone_expected_quality}.}

% \dkreplace{The proof follows by induction on the size of quality set $\QualSet$. Note that the proof is based on the categorical model, but can be straightforwardly generalized to the continuous model. For a more detailed proof for the continuous model, one can refer Theorem 1 of \cite{torres2005multivariate}.}{
Here, we provide a proof for the categorical model. It can be easily modified for the continuous model, and the result also is shown as Theorem~1 by \citet{torres2005multivariate}.
We show the result by induction on the size of the quality set $\QualSet$.

\emph{Base case:} We show that the inequality holds for any binary quality set $\QualSet=\SET{q_1, q_2}$ with $q_2>q_1$.  For $\RevSigV \in \SigSet^{\NumReviews}$, let $\gamma(\RevSigV)=\frac{\RevSigProb[q_2]{\RevSigV}}{\RevSigProb[q_1]{\RevSigV}}$ be the likelihood ratio. We first rewrite the ex-post expected quality:
\begin{align*}
    U(\RevSigV) & = \ExpectC{Q}{\RevSigV}\\
    &= \sum_q q\cdot \ProbC{Q=q}{\RevSigV}\\
    &= \sum_q q \cdot \frac{\Prob{q}\cdot \ProbC{\RevSigV}{Q=q}}{\sum_{q'}\Prob{q'} \cdot \ProbC{\RevSigV}{Q=q'}}\\
    &= q_1 \cdot \frac{\QualProb{q_1}\RevSigProb[q_1]{\RevSigV}}{\QualProb{q_1}\RevSigProb[q_1]{\RevSigV}+\QualProb{q_2}\RevSigProb[q_2]{\RevSigV}} + q_2\cdot \frac{\QualProb{q_2}\RevSigProb[q_2]{\RevSigV}}{\QualProb{q_1}\RevSigProb[q_1]{\RevSigV}+\QualProb{q_2}\RevSigProb[q_2]{\RevSigV}}\\
    &=  \frac{q_1 \cdot \QualProb{q_1}\RevSigProb[q_1]{\RevSigV} + q_2\cdot\QualProb{q_2}\RevSigProb[q_2]{\RevSigV}}{\QualProb{q_1}\RevSigProb[q_1]{\RevSigV}+\QualProb{q_2}\RevSigProb[q_2]{\RevSigV}}\\
    &= \frac{q_1\cdot\QualProb{q_1} + q_2\cdot\QualProb{q_2}\gamma(\RevSigV)}{\QualProb{q_1}+\QualProb{q_2}\gamma(\RevSigV)}.
\end{align*}
    
Now, we consider the difference $U(\RevSigVP) - U(\RevSigV)$, write it using a common denominator, and cancel common terms. Then, 
% \dkdeletecomment{we'd also need to specify the one strict inequality, no?}{because $\RevSigVP \ge \RevSigV$ component-wise,} 
by \cref{def:informative}, we obtain that $\gamma(\RevSigVP)>\gamma(\RevSigV)$, which we substitute:
\begin{align*}
    U(\RevSigVP) - U(\RevSigV)
    &=  \frac{\QualProb{q_1}\QualProb{q_2} \cdot \left(q_1\gamma(\RevSigV) + q_2\gamma(\RevSigVP) - q_1\gamma(\RevSigVP) - q_2\gamma(\RevSigV)\right)}{\left(\QualProb{q_1}+\QualProb{q_2}\gamma(\RevSigVP)\right) \cdot \left(\QualProb{q_1}+\QualProb{q_2}\gamma(\RevSigV)\right)}\\
    &= \frac{\QualProb{q_1}\QualProb{q_2} \cdot (\gamma(\RevSigVP) - \gamma(\RevSigV)) \cdot (q_2-q_1)}{\left(\QualProb{q_1}+\QualProb{q_2}\gamma(\RevSigVP)\right) \cdot \left(\QualProb{q_1}+\QualProb{q_2}\gamma(\RevSigV)\right)}\\
    &>0.
\end{align*}

\emph{Induction step: } We show that if the inequality holds for all categorical models with support size $|\QualSet|=n-1$, it also holds for categorical models with support size $|\QualSet|=n$. 
Let $\QualSet=\SET{q_1, q_2, \ldots, q_n}$ with $q_n>\cdots>q_1$. Using the same derivation as in the base case, 
\begin{align*}
    U(\RevSigVP) - U(\RevSigV)
    &= \frac{\left(\sum_{i=1}^n q_i\QualProb{q_i}\RevSigProb[q_i]{\RevSigVP}\right) \cdot \left(\sum_{i=1}^n \QualProb{q_i}\RevSigProb[q_i]{\RevSigV}\right) - \left(\sum_{i=1}^n q_i\QualProb{q_i}\RevSigProb[q_i]{\RevSigV}\right) \cdot \left(\sum_{i=1}^n \QualProb{q_i}\RevSigProb[q_i]{\RevSigVP}\right)}{\left(\sum_{i=1}^n \QualProb{q_i}\RevSigProb[q_i]{\RevSigVP}\right) \cdot \left(\sum_{i=1}^n \QualProb{q_i}\RevSigProb[q_i]{\RevSigV}\right)}.
\end{align*}

Because the denominator is strictly positive, we now focus on showing that the numerator is as well.
We show that the difference between the actual numerator, and the version where the upper bound of each summation is $n-1$, is positive. Then, because the latter is positive by induction hypothesis, we will have shown the claim.
\begin{align*}
& \left( \left(\sum_{i=1}^{n} q_i\QualProb{q_i}\RevSigProb[q_i]{\RevSigVP}\right) \cdot \left(\sum_{i=1}^{n} \QualProb{q_i}\RevSigProb[q_i]{\RevSigV}\right) - \left(\sum_{i=1}^{n} q_i\QualProb{q_i}\RevSigProb[q_i]{\RevSigV}\right)\left(\sum_{i=1}^{n} \QualProb{q_i}\RevSigProb[q_i]{\RevSigVP}\right) \right)
\\ - & 
 \left( \left(\sum_{i=1}^{n-1} q_i\QualProb{q_i}\RevSigProb[q_i]{\RevSigVP}\right) \cdot \left(\sum_{i=1}^{n-1} \QualProb{q_i}\RevSigProb[q_i]{\RevSigV}\right) - \left(\sum_{i=1}^{n-1} q_i\QualProb{q_i}\RevSigProb[q_i]{\RevSigV}\right)\left(\sum_{i=1}^{n-1} \QualProb{q_i}\RevSigProb[q_i]{\RevSigVP}\right) \right)
\\ = & 
 q_n \QualProb{q_n} \RevSigProb[q_{n}]{\RevSigVP} \cdot \left( \sum_{i=1}^{n-1} \QualProb{q_i}\RevSigProb[q_i]{\RevSigV} \right) 
 + \QualProb{q_n} \RevSigProb[q_n]{\RevSigV} \cdot \left( \sum_{i=1}^{n-1} q_i \QualProb{q_i}\RevSigProb[q_i]{\RevSigVP} \right)
\\ - & q_n \QualProb{q_n} \RevSigProb[q_n]{\RevSigV} \cdot \left( \sum_{i=1}^{n-1} \QualProb{q_i} \RevSigProb[q_i]{\RevSigVP} \right) 
- \QualProb{q_n} \RevSigProb[q_n]{\RevSigVP} \cdot \left(\sum_{i=1}^{n-1} q_i \QualProb{q_i}\RevSigProb[q_i]{\RevSigV} \right)
\\ = & 
\sum_{i=1}^{n-1} \QualProb{q_i} \QualProb{q_n} \cdot \left( \RevSigProb[q_n]{\RevSigVP}\RevSigProb[q_i]{\RevSigV} - \RevSigProb[q_n]{\RevSigV} \RevSigProb[q_i]{\RevSigVP} \right) \cdot \left( q_n - q_i \right).
\end{align*}

By \cref{def:informative}, 
% \dkreplace{for any $i<k+1$ and $\RevSigVP\ge \RevSigV$ component-wise with strict inequality for at least one of the components, we know that 

% \begin{align*}
%     &\frac{\RevSigProb[q_{k+1}]{\RevSigVP}}{\RevSigProb[q_i]{\RevSigVP}}>\frac{\RevSigProb[q_{k+1}]{\RevSigV}}{\RevSigProb[q_i]{\RevSigV}}\\
%     \Leftrightarrow\qquad & \RevSigProb[q_{k+1}]{\RevSigVP}\RevSigProb[q_{i}]{\RevSigV} > \RevSigProb[q_{k+1}]{\RevSigV}\RevSigProb[q_i]{\RevSigVP}.
%     \end{align*}

% This implies that $P(k+1)>0$, establishing the induction step.
% }{
$\RevSigProb[q_n]{\RevSigVP} \RevSigProb[q_i]{\RevSigV} > \RevSigProb[q_n]{\RevSigV}\RevSigProb[q_i]{\RevSigVP}$, proving that the difference is strictly positive. Thus, by induction hypothesis, the utility difference is strictly positive, completing the proof. \Halmos
\endproof

\subsection{Proof of \cref{lem:author_response}}\label{app:proof-author_resp}

Let $\ACCMAP$ be the acceptance policy. 
The probability of acceptance for a paper of quality $\Qual=q$ is $\AccP{\ACCMAP}{q}$. By Proposition~\ref{prop:monotone-prob}, this probability is non-decreasing in $q$ for monotone acceptance policies, and strictly increasing for non-trivial threshold policies.
The author submits the paper if the expected utility of submitting is greater than 1 (the utility of the outside option), does not submit the paper if the expected utility is less than 1, and is indifferent between submitting or not if the expected utility is equal to 1. 
Because a noiseless author does not learn any new information from rejection in previous rounds, she will make the same decision in future rounds, implying that she will submit until acceptance.
% Since the author will face the same tradeoff in future rounds\footnote{Crucially, a noiseless author does not learn any new information from rejection in previous rounds.}, she will make the same decision, so she will submit until acceptance. 
Let $\ConfValue$ be fixed. The expected utility can be obtained as the time-discounted sum of the utility from acceptance:
\begin{align}
\AUTHUTIL(\ACCMAP,q)
& = 
\sum_{t\ge 1} \ConfValue \cdot \TD^{t-1} \AccP{\ACCMAP}{q} \cdot (1-\AccP{\ACCMAP}{q})^{t-1} 
\; = \; \frac{\ConfValue \cdot \AccP{\ACCMAP}{q}}{1-\TD \cdot (1-\AccP{\ACCMAP}{q})}.
\label{eqn:submission-inequality}
\end{align}
Solving the inequalities $\AUTHUTIL(\ACCMAP,q) > 1$, $\AUTHUTIL(\ACCMAP,q) < 1$, and $\AUTHUTIL(\ACCMAP,q) = 1$ for $q$, the author submits the paper if $\AccP{\ACCMAP}{q} > 1/\rho$, does not submit if $\AccP{\ACCMAP}{q} < 1/\rho$, and is indifferent between submitting or not if $\AccP{\ACCMAP}{q} = 1/\rho$, respectively. 
This completes the proof of the lemma.

\subsection{Proof of \cref{prop:de_facto}}\label{app:proof-de_facto}
By \cref{lem:author_response}, an author with a paper of quality $Q=q$ decides whether to submit based on whether the acceptance probability $\AccP{\ACCMAP}{q}$ is greater than the inverse of the attractiveness factor, $1/\rho$.
Recall that $\AccP{\ACCMAP}{q}$ is weakly increasing in $q$ by \cref{prop:monotone-prob}, while $1/\rho(q,r)$ is strictly decreasing in $q$.
Since the conference's policy is responsive, there must exist some $\bar{q}$ such that $\AccP{\ACCMAP}{\bar{q}} > 1/\rho(\bar{q},1)$.
If $\AccP{\ACCMAP}{q} \ge 1/\rho(q,1)$ for all $q\in \QualSet$, then every author prefers to submit and $\theta = -\infty$ is a de facto threshold.
Therefore, we only have to consider the case where there exists some $\ubar{q}$ such that $\AccP{\ACCMAP}{\ubar{q}} < 1/\rho(\ubar{q},1)$.
% Moreover, under our model's assumption 
% \dkcomment{I did not see that assumption anywhere in either the model or the statement of the lemma. Where is it?}\yzcomment{It was commented out for some reason. I put it back in section 2.2} 
% that the conference's value is less than 1 if all papers are accepted, there exists some $\ubar{q}$ such that $\AccP{\ACCMAP}{\ubar{q}} < 1 < 1/\rho(\ubar{q},1)$.

Combining these observations, we conclude that either there exists some $\theta'\in \QualSet$ such that $\AccP{\ACCMAP}{\theta'} = 1/\rho(\theta',1)$, or there exist two qualities $\ubar{\theta}, \bar{\theta}\in \QualSet$, with no intermediate qualities between them, such that $\AccP{\ACCMAP}{\ubar{\theta}} < 1/\rho(\ubar{\theta},1)$ and $\AccP{\ACCMAP}{\bar{\theta}} > 1/\rho(\bar{\theta},1)$.
We show that a threshold best response exists for either case.

First, if there exists a $\theta'\in \QualSet$ such that $\AccP{\ACCMAP}{\theta'} = 1/\rho(\theta',1)$, then it is a threshold best response:  every author submits and keeps resubmitting if her paper has quality $Q = q\ge \theta'$ and takes the outside option if $q < \theta'$.
To see this, note that authors with a paper of quality $q\ge \theta'$ all have an acceptance probability $\AccP{\ACCMAP}{q} \ge 1/\rho(\theta', 1)$, meaning that they are weakly happier to submit; on the other hand, the acceptance probability $\AccP{\ACCMAP}{q} \le 1/\rho(\theta',1)$ for authors with $q< \theta'$, meaning that they are weakly happier to take the outside option.

If there exist two adjacent qualities $\ubar{\theta}, \bar{\theta}\in \QualSet$ such that $\AccP{\ACCMAP}{\ubar{\theta}} < 1/\rho(\ubar{\theta},1)$ and $\AccP{\ACCMAP}{\bar{\theta}} > 1/\rho(\bar{\theta},1)$, we further distinguish two cases.
In the first case,  $\AccP{\ACCMAP}{\ubar{\theta}}\le 1/\rho(\bar{\theta},1)$. This means that $\AccP{\ACCMAP}{\ubar{\theta}}\le 1/\rho(\ubar{\theta},r)$ for any $r\in [0,1]$, because $\ubar{\theta}<\bar{\theta}$.\fangcomment{Minor point: we use $r$ for three different meaning 1) superscript for reviewer $F^{(r)}$, 2) threshold acceptance policy $\phi_{\tau, r}$, and 3) author's threshold strategy.} \dkcomment{Good point! At least the review noise should perhaps be $R$? Unfortunately, we haven't been consistent in using macros, so this will now be hard to fix without missing one. But we should probably do it after we submit, before the next round of reviews/acceptance.}
In other words, if authors with $Q \ge \bar{\theta}>\ubar{\theta}$ all decide to submit and authors with $Q < \ubar{\theta}$ all decide to take the outside option, then authors with $Q = \ubar{\theta}$ all (weakly) prefer to take the outside option no matter with what probability the other authors with $Q= \ubar{\theta}$ submit.
Therefore, the following is an equilibrium for authors: every author with $Q \ge \bar{\theta}$ submits, and every author with $Q < \bar{\theta}$ takes the outside option.

In the second case, $\AccP{\ACCMAP}{\ubar{\theta}}> 1/\rho(\bar{\theta},1)$. 
This corresponds to the following case: if authors with paper quality $Q \ge \bar{\theta}$ submit and those with $Q < \ubar{\theta}$ opt for the outside option, then authors with paper quality $Q = \ubar{\theta}$ prefer submitting if no one at that quality submits but prefer the outside option when everyone at that quality submits.
% Therefore, there exists some probability $r$ such that it is a best response for authors with $Q = \ubar{\theta}$ to submit with probability $r$, while authors with $Q > \ubar{\theta}$ submit with probability 1, and those with $Q < \ubar{\theta}$ choose the outside option.

At equilibrium, the probability $r$ with which such authors submit must make authors with paper quality $Q = \ubar{\theta}$ indifferent between submitting or not submitting. Therefore, $r$ must be is the solution to $\AccP{\ACCMAP}{\ubar{\theta}} = 1/\rho(\ubar{\theta},r)$, i.e.,
\begin{equation*} 
        \frac{r\cdot \ubar{\theta} \cdot \QualProb{\ubar{\theta}} + \sum_{q \in \QualSet, q \geq \bar{\theta}} q \cdot \QualProb{q}}{%
        r\cdot\QualProb{\ubar{\theta}} + \sum_{q \in \QualSet, q \geq \bar{\theta}}\QualProb{q}} 
        = \frac{1-\TD \cdot (1-\AccP{\ACCMAP}{\ubar{\theta}})}{\AccP{\ACCMAP}{\ubar{\theta}}}.
\end{equation*}
Note that the solution must exist because by the relationship between $P_{\text{acc}}$ and $1/\rho$ in this case, when $r = 0$, the left-hand side of the above equation is larger than the right-hand side, while if $r = 1$, the left-hand side is smaller than the right-hand side.
By the continuity of the left-hand-side as a function of $r$, a solution must exist.
This completes the proof of the first part.

As noted above, by \cref{prop:monotone-prob}, under a monotone acceptance policy, the acceptance probability is weakly increasing in $q$.
As a result, multiple adjacent quality levels may share the same acceptance probability, which can lead to non-threshold best responses.
For example, suppose that there are three adjacent quality levels $q_1<q_2<q_3$ such that $\AccP{\ACCMAP}{q_1} = \AccP{\ACCMAP}{q_2} = \AccP{\ACCMAP}{q_3} = 1/\rho(q_2,1)$.
Based on the earlier argument, there exists a threshold best response at $\theta = q_2$, where the authors submit if and only if $Q\ge q_2$ and the corresponding conference value is $\ConfValue(q_2,1)$.
However, it is also possible for a non-threshold best response to exist: for example, there may be a submission strategy where authors with qualities $q_1, q_2$, and $q_3$ submit with probabilities $r_1, r_2$, and $r_3$, respectively, and $r_1, r_3 \in (0,1)$. As long as the induced conference value remains at $\ConfValue(q_2)$, such a mixed strategy can also constitute a non-threshold best response.

However, when the conference applies a responsive threshold acceptance policy, $\AccP{\ACCMAP}{q}$ is \emph{strictly} increasing in $q$.
This implies that no matter what the conference value is, there exists at most one quality level at which authors are indifferent between submitting and not submitting. 
This immediately guarantees that every best response must be a threshold strategy.
If such a quality level $\theta'$ where authors are indifferent between submitting or not submitting exists, then $\theta = \theta'$ is a de facto threshold.
If not, there again exist two adjacent qualities $\ubar{\theta}, \bar{\theta}\in \QualSet$ such that $\AccP{\ACCMAP}{\ubar{\theta}} < 1/\rho(\ubar{\theta},1)$ and $\AccP{\ACCMAP}{\bar{\theta}} > 1/\rho(\bar{\theta},1)$.
Then, based on the previous arguments, there are again two cases: if $\AccP{\ACCMAP}{\ubar{\theta}}\le 1/\rho(\bar{\theta},1)$, any $\theta \in [\ubar{\theta}, \bar{\theta}]$ is a de facto threshold; if $\AccP{\ACCMAP}{\ubar{\theta}}> 1/\rho(\bar{\theta},1)$, then $\theta = \ubar{\theta}$ is a de facto threshold.
This completes the second part.

% We next prove the third part, so we assume that the conference's policy is a non-trivial threshold acceptance policy.
% By \cref{prop:monotone-prob}, $\AccP{\ACCMAP}{q}$ is strictly monotone in $q$. 
% Therefore, in the categorical model, there exists at most one $\hat{q} \in \QualSet$ such that $\AccP{\ACCMAP}{\hat{q}}=1/\rho$.
% If such a $\hat{q}$ exists, it is a de facto threshold. If not, then because the threshold acceptance policy is non-trivial, there exist qualities $\ubar{q} < \bar{q}$ such that $\AccP{\ACCMAP}{\ubar{q}} < 1/\rho$, $\AccP{\ACCMAP}{\bar{q}} > 1/\rho$, and there are no qualities in $\QualSet$ between $\ubar{q}$ and $\bar{q}$.  
% Then, any $\theta \in [\ubar{q}, \bar{q}]$ is a de facto threshold. To complete the third part of the proof, notice that if $\theta$ is a de facto threshold, then every best response by the author is a $\theta$-threshold strategy.
% {the author is strictly better off to submit any paper with quality $q\ge \bar{q}$ and not submit any paper with quality $q< \bar{q}$, which results in a threshold best response. In this case, any $\theta \in [\ubar{q}, \bar{q}]$ is a de facto threshold.}

Finally, we prove the third part. 
In the continuous model, the distribution of the review noise is assumed to be a bijection, which implies that $\AccP{\ACCMAP}{q}$ is continuously increasing in $q$. 
Moreover, as the density of $\QualDist$ has full support, $1/\rho(q)$ is continuously decreasing in $q$.
Because the acceptance policy is responsive, there must exist a unique intersection such that $\AccP{\ACCMAP}{\theta} = 1/\rho(\theta)$, meaning that $\theta$ is a unique de facto threshold. \Halmos

\subsection{Proof of \cref{prop:threshold-policy}}
\label{app:proof-threshold-policy}

% \dkcomment{We are restating some of the lemmas/propositions from the main text when we prove them in the appendix, but not others. Is there a concrete reason why? If not, should we be consistent? E.g., restate all of them?}

To prove the proposition, we want to relate the acceptance threshold to how ``strict'' the policy is. We begin by defining a comparison between the strictness of two policies:

\begin{definition} \label{def:stricter}
An acceptance policy \ACCMAP['] is \emph{(weakly) stricter} than another policy $\ACCMAP$ if it accepts every paper with a (weakly) smaller probability, i.e., $\AccP{\ACCMAP[']}{q} < \AccP{\ACCMAP}{q}$ for all $q$ (resp., $\AccP{\ACCMAP[']}{q} \leq \AccP{\ACCMAP}{q}$ for all $q$ for the weak version).
\end{definition}

Being stricter appears to be a very demanding requirement, in that it requires an inequality for all paper qualities. 
We next show that for threshold policies, it in fact follows from a strictly smaller acceptance probability for just one paper.

\begin{lemma}\label{lem:stricter_policy}
Let $\ACCMAP$ and $\ACCMAP[']$ be two threshold acceptance policies. If there exists a $q\in \QualSet$ such that $\AccP{\ACCMAP[']}{q} < \AccP{\ACCMAP}{q}$, then $\ACCMAP[']$ is stricter than $\ACCMAP$.
\end{lemma}

\proof{Proof of \Cref{lem:stricter_policy}.}

  We want to show that if a threshold policy accepts one type of paper with strictly smaller probability, it accepts \emph{every} paper with strictly smaller probability. 

  First, recall that we assumed that the conditional signal distribution has full support on the signal space. Because multiple reviews are i.i.d., conditioned on any paper quality $q$, the review signal distribution has full support over all vectors of review signals.

  Because \ACCMAP, \ACCMAP['] are \emph{threshold} policies, any review vector that leads to acceptance under \ACCMAP['] must lead to acceptance under \ACCMAP with at least the same probability. And because a paper of quality $q$ is accepted with strictly higher probability by \ACCMAP, there must exist a set $S$ of review vectors \RevSigV which are all accepted with strictly higher probability under \ACCMAP than under \ACCMAP['], such that $S$ has strictly positive probability mass under the combination of the distributions $\QualDist$ and $\RevSigDist[q]$.
  
  Because $S$ occurs with positive probability for every paper quality $q'$ (by the full-support assumption on the $\RevSigDist[q]$), every paper is accepted with strictly higher probability by \ACCMAP than by \ACCMAP['], completing the proof. \Halmos
\endproof

The following lemma relates the strictness of a threshold acceptance policy with its acceptance threshold.

\begin{lemma} \label{prop:monotone-prob-threshold}
In the categorical model, let $\ACCMAP[\tau,r]$ and $\ACCMAP[\tau',r']$ be two threshold acceptance policies with $r, r'\in (0,1]$.
Then, if either $\tau'>\tau$ or $\tau'=\tau$ and $r'< r$, $\ACCMAP[\tau',r']$ is weakly stricter than $\ACCMAP[\tau,r]$. 

 In the continuous model, let $\tau$ and $\tau'$ be the thresholds for two non-trivial threshold policies. Then, $\ACCMAP[\tau']$ is stricter than $\ACCMAP[\tau]$ if and only if $\tau' > \tau$.
\end{lemma}

\proof{Proof of \Cref{prop:monotone-prob-threshold}.}

We begin by proving the first part of the lemma, regarding the categorical model. By the definition of $\ACCMAP[\tau,r]$ and $\ACCMAP[\tau',r']$, whenever a paper with some review vector $\RevSigV$ is accepted by $\ACCMAP[\tau',r']$ with positive probability, it will be accepted by $\ACCMAP[\tau,r]$ with at least the same probability. This means that $\ACCMAP[\tau,r]$ accepts every paper with at least the same probability as  $\ACCMAP[\tau',r']$.

% \dkcomment{Switching order of ``if'' vs.~``only if'', because ``if'' takes more work.}
Next, we prove the second statement, regarding the continuous model. For the ``only if'' direction, we know that a stricter threshold policy accepts every paper with strictly smaller probability. In the continuous model, $\ACCMAP[\tau']$ must reject some review vectors that are accepted under $\ACCMAP[\tau]$. This implies $\tau'>\tau$. 

For the ``if'' direction, we begin with some basic observations. First, by \cref{lem:monotone_expected_quality}, $U(\RevSigV)$ is monotone increasing in \RevSigV.
Second, by standard Bayesian updating formulas and using that the review noise is additive and independent of $q$, the expected quality conditioned on the signal \RevSigV can be written as 
\begin{align*}
    U(\RevSigV)
    & = \int q \cdot \frac{\QualDens{q} \cdot \prod_{i=1}^{\NumReviews} \RevSigProb[q]{\RevSig{i}}}{\int \QualDens{q'} \cdot \prod_{i=1}^{\NumReviews} \RevSigProb[q']{\RevSig{i}}\, d q'} \, d q
    = \int q \cdot \frac{\QualDens{q} \cdot \prod_{i=1}^{\NumReviews} f^{(r)}(\RevSig{i}-q)}{\int \QualDens{q'} \cdot \prod_{i=1}^{\NumReviews} f^{(r)}(\RevSig{i}-q') \, d q'} \, d q.
\end{align*}

In this form, because we assumed $f^{(r)}$ to be a continuous function, and $\QualDens{p}$ to be strictly positive, it is easy to see that $U(\RevSigV)$ is also a continuous function of \RevSigV.

Because neither threshold acceptance policy is trivial by assumption, there exist review vectors $\RevSigV, \RevSigVP$ with $U(\RevSigV) < \tau < \tau' < U(\RevSigVP)$, and in particular, because $U$ is monotone, writing $\mathbf{1}$ for the all-1 vector, $\lim_{y \to - \infty} U(y \cdot \mathbf{1}) < \tau$, and $\lim_{y \to \infty} U(y \cdot \mathbf{1}) > \tau'$.
By continuity and monotonicity of $U$, the function $U(y \cdot \mathbf{1})$ is continuous and monotone as a function of $y$, so there exists a $y$ with $U(y \cdot \mathbf{1}) = \frac{\tau+\tau'}{2}$.

Again by continuity of $U$, there is a sufficiently small $\epsilon > 0$ such that for every $\RevSigV$ with $|| \RevSigV -  y \cdot \mathbf{1} ||_2 \leq \epsilon$, we have that 
$U(\RevSigV) \in (\tau, \tau')$.
Let $S = \Set{\RevSigV}{|| \RevSigV - y \cdot \mathbf{1} ||_2 \leq \epsilon}$.
$S$ has positive volume, and because we assumed that $f^{(r)}(x) > 0$ for all $x$, for every quality $q$, the event of obtaining a review vector $\RevSigV \in S$ has positive probability.
Furthermore, whenever the review vector $\RevSigV \in S$, the paper is accepted by $\ACCMAP[\tau]$, but rejected by $\ACCMAP[\tau']$.
Thus, every paper quality $q$ is strictly more likely to be accepted under $\ACCMAP[\tau']$ than under $\ACCMAP[\tau]$.
\Halmos
\endproof

\proof{Proof of \cref{prop:threshold-policy}.}

We prove the proposition based on two cases, depending on whether $\theta\in \QualSet$.
If $\theta \in \QualSet$, let $r\in [0,1]$ be any value such that $\ConfValue(\theta,r)> 1$; the existence of such an $r$ is guaranteed because $\theta$ is a candidate threshold.
In this case, let $\hat{\theta} = \theta$ and $\rho = \frac{\ConfValue(\theta, r)-\TD}{1-\TD}$.
If $\theta \notin \QualSet$, let $\hat{\theta} = \inf\Set{q \in \QualSet}{q > \theta}$ and $\rho = \frac{\ConfValue(\hat{\theta}, 1)-\TD}{1-\TD}$.
Note that because $\theta$ is a candidate threshold, $V(\hat{\theta},1) > 1$. 
This, in both cases (discrete and continuous), we obtain that $1/\rho  < 1$. 

To prove the proposition, it is sufficient to show that, in either case, there exists a threshold acceptance policy such that $\AccP{\ACCMAP}{\hat{\theta}} = 1/\rho$.
To see this, when $\theta \in \QualSet$, it is a best response for authors with $Q > \theta$ to submit, while authors with $Q < \theta$ take the outside option, and authors with $Q = \theta$ submit with probability $r$.
When $\theta \notin \QualSet$, it is a best response for authors with $Q > \theta$ to submit, while authors with $Q < \theta$ take the outside option.
This means that as long as $\AccP{\ACCMAP}{\hat{\theta}} = 1/\rho$, $\theta$ is a de facto threshold.

% Let $\hat{\theta}$ be a value in $\QualSet$ closest to $\theta$, i.e.~$\hat{\theta} \in \arg\min_{q\in \QualSet} (q-\theta)^2$ (breaking ties arbitrarily).
We now show that there always exists a threshold acceptance policy such that $\AccP{\ACCMAP}{\hat{\theta}} = 1/\rho$.
Let $f(\tau) := \AccP{\ACCMAP[\tau,0]}{\hat{\theta}}$ be the probability that a paper of quality $\hat{\theta}$ is accepted under the policy $\ACCMAP[\tau, 0]$, which accepts the paper if and only if its expected posterior quality is greater than $\tau$.
Then, $\lim_{\tau \to -\infty} f(\tau) = 1 > 1/\rho > 0 = \lim_{\tau \to +\infty} f(\tau)$. Furthermore, by Lemma \ref{prop:monotone-prob-threshold}, $f(\tau)$ is a non-increasing function of $\tau$. 
Therefore, there must exist a $\hat{\tau}$ such that $\lim_{\tau \to \hat{\tau} \uparrow} f(\tau) \geq 1/\rho \geq \lim_{\tau \to \hat{\tau} \downarrow} f(\tau)$.
  
If $f$ is continuous at $\hat{\tau}$, then the threshold policy $\ACCMAP[\hat{\tau},0]$ has $\AccP{\ACCMAP[\hat{\tau},0]}{\hat{\theta}} = 1/\rho$ by definition.
Otherwise, let $z = \lim_{\tau \to \hat{\tau} \uparrow} f(\tau) - \lim_{\tau \to \hat{\tau} \downarrow} f(\tau) > 0$. We can infer that there must be a discrete probability of $z$ for the event that \fangcomment{do we use $U(\RevSigV)$?}$\ExpectC{\Qual}{\RevSigV} = \hat{\tau}$, i.e., that $\Prob[{\RevSigV \sim \RevSigDist[{\hat{\theta}}]}]{\ExpectC{\Qual}{\RevSigV} = \hat{\tau}} = z$.
We then consider the threshold policy $\ACCMAP[\hat{\tau},\hat{r}]$ with threshold $\hat{\tau}$ which conditioned on $\ExpectC{\Qual}{\RevSigV} = \hat{\tau}$ accepts a paper with probability $\hat{r} := \frac{1/\rho-\lim_{\tau \to \hat{\tau} \downarrow} f(\tau)}{z}$. The overall acceptance probability of $\ACCMAP[\hat{\tau},\hat{r}]$ for a paper with quality $\hat{\theta}$ is therefore
\begin{align*} 
 & \Prob[{\RevSigV \sim \RevSigDist[{\hat{\theta}}]}]{\ExpectC{\Qual}{\RevSigV} > \hat{\tau}} + \Prob[{\RevSigV \sim \RevSigDist[{\hat{\theta}}]}]{\ExpectC{\Qual}{\RevSigV} = \hat{\tau}} \cdot \frac{1/\rho-\lim_{\tau \to \hat{\tau} \downarrow} f(\tau)}{z}
\\ & = \lim_{\tau \to \hat{\tau} \downarrow} f(\tau)
+ z \cdot \frac{1/\rho-\lim_{\tau \to \hat{\tau} \downarrow} f(\tau)}{z}
\; = \; 1/\rho.
\end{align*}

Thus, under the threshold acceptance policy $\ACCMAP[\hat{\tau},\hat{r}]$, authors of papers with quality $\hat{\theta}$ are indifferent between submitting and not submitting.\fangcomment{do we need a proof for the moreover part?}
% Note that if $\theta\in\QualSet$, then $\hat{\theta}=\theta$. The result for the special case in the proposition statement straightforwardly follows.
% For the general case, assume w.l.o.g.~(the other case is symmetric) that $\hat{\theta} \geq \theta$.
% Because the author is indifferent between submitting and not submitting papers of quality $\Qual = \hat{\theta}$, we may assume that such papers are submitted, and hence all papers of quality at least $\hat{\theta}$. 
% By definition of $\hat{\theta}$, no papers have quality $\Qual \in [\theta, \hat{\theta})$, so all papers of quality at least $\theta$ are submitted. Thus, $\theta$ is also a de facto threshold for the author.

Finally, uniqueness of the conference's threshold policy in the continuous model when the author's submission decision is non-trivial follows directly because by Lemma~\ref{prop:monotone-prob-threshold}, the solution for $\hat{\tau}$ of $\AccP{\ACCMAP[\hat{\tau},0]}{\hat{\theta}} = 1/\rho$ is unique in the continuous model when agents are neither submitting all papers nor submitting no papers.
\Halmos
\endproof

\subsection{Proof of \cref{prop:gap-QB-dominance}}\label{app:proof-gap-QB-dominance}

In the continuous model, fixing the quality prior, the conference quality only depends on the de facto threshold.
Therefore, to prove the proposition, it is sufficient to show that for every candidate threshold $\theta$, the review burden of the first setting is larger than that of the second setting.

Recall that the review burden is given by $R(\theta) = m\int_\theta^\infty \QualDens{q}/\AccP{\ACCMAP[\tau(\theta)]}{q} dq$, suppress parameter dependence and denote the corresponding acceptance thresholds by $\tau(\theta)$ and $\tau'(\theta)$ for the two settings, respectively.
It is thus sufficient to show that $\AccP{\ACCMAP[\tau(\theta)]}{q} < \AccP{\ACCMAP[\tau'(\theta)]}{q}$ for every candidate threshold $\theta$.
This inequality follows directly from the fact that the resubmission gap in the first setting is larger, i.e., $\tau(\theta) > \tau'(\theta)$ for any $\theta$. 
By \cref{prop:monotone-prob-threshold}, the claim follows, completing the proof.

\subsection{Proof of \cref{prop:blackwell}}\label{app:proof-blackwell}

We will show that any memoryless acceptance policy $\ACCMAP[']: \SigSet' \to [0,1]$ in the second setting can be simulated in the first setting with the same expected conference quality and review burden. It follows that both the QB-tradeoff and the QA-tradeoff in the first setting weakly dominate those in the second setting.

Because $\RevSigDist$ Blackwell dominates $\RevSigDistP$, there exists a garbling $\gamma$ from $\SigSet$ to $\SigSet'$.  
We define a memoryless acceptance policy $\ACCMAP$ in the first setting: for any review signal $\REVSIG$, we set
\begin{align*}
\AccMap{\REVSIG} & = \sum_{\REVSIGP \in \SigSet'} 
\AccMap[']{\REVSIGP} \cdot \gamma (\REVSIG, \REVSIGP). 
\end{align*}

First, because $(\gamma(\REVSIG, \REVSIGP))_{\REVSIGP \in \SigSet'}$ is a distribution on $\SigSet'$ and $\AccMap[']{\REVSIGP} \in [0,1]$ for all $\REVSIGP$, the output of $\ACCMAP$ is in $[0,1]$. 
Moreover, for any paper quality $\qual \in \QualSet$,
\begin{align*}
\AccP{\ACCMAP}{\qual} 
& = \sum_{\dkreplace{\RevSigV}{\REVSIG} \in \SigSet} 
\RevSigProb[\qual]{\REVSIG}
\cdot \AccMap{\REVSIG} \\
& = \sum_{\dkreplace{\RevSigV}{\REVSIG} \in \SigSet} 
 \RevSigProb[\qual]{\REVSIG}
\cdot \sum_{\REVSIGP \in \SigSet'} \AccMap[']{\REVSIGP} \cdot
\gamma(\REVSIG, \REVSIGP) \tag{Definition of $\ACCMAP$}\\
& = \sum_{\REVSIGP \in \SigSet'} \left( \sum_{\REVSIG \in \SigSet}
\RevSigProb[\qual]{\REVSIG} \cdot \gamma(\REVSIG, \REVSIGP) \right)
\cdot \AccMap[']{\REVSIGP} \tag{Changing order of summation} \\
& = \sum_{\REVSIGP \in \SigSet'}  \RevSigProbP[\qual]{\REVSIGP} \cdot \AccMap[']{\REVSIGP} \tag{$\gamma$ is a garbling} \\
& = \AccP{\ACCMAP[']}{\qual}.
\end{align*}

Thus, the acceptance policies $\ACCMAP$ and $\ACCMAP[']$ have identical acceptance probabilities for each paper quality; this makes them indistinguishable to authors. In particular, both acceptance policies have the same expected conference quality and review burden. \Halmos

\subsection{Proof of \cref{prop:blackwell-threshold}}\label{app:proof-blackwell-threshold}

We first prove \cref{claim:blackwell-RB-better}, which is the key to the proof of \cref{prop:blackwell-threshold}.
We note that our proof is written for the categorical model. However, the proof for the continuous model is analogous, simply by replacing the summations with integrals.
\yzcomment{Note for future: better to check if this is precisely true.}
% \fangdelete{
% \restatedlemma{claim:blackwell-RB-better}{
%     Consider two threshold acceptance policies $\ACCMAP$ and $\ACCMAP[']$ which accept papers of quality $\bar{q}$ with \yichiedit{equal} probability in the first and the second setting, respectively.  
%     Then, under the author's $\bar{q}$-threshold strategy, 
%     % the review burden of $\ACCMAP$ in the first setting is no larger than the review burden of $\ACCMAP'$ in the second setting.  
%     \yichiedit{the acceptance probability of a paper of quality $q$ in the first setting is no less than that in the second setting, for any $q>\bar{q}$.}
%     }}

\proof{Proof of \cref{claim:blackwell-RB-better}.}
    Let $(\tau, r)$ and $(\tau', r')$ be the parameters corresponding to the threshold policies $\ACCMAP$ and $\ACCMAP'$, respectively.
    Given that the acceptance probability of $\bar{q}$ is the same under $\ACCMAP[\tau,r]$ and $\ACCMAP[\tau', r']$, we want to show that for every $q > \bar{q}$, the acceptance probability is weakly higher under $\ACCMAP[\tau,r]$ than under $\ACCMAP[\tau', r']$.
    By decomposing the acceptance probability into the individual signals, we need to show that
    \begin{equation} \label{eq:blackwell_ineq}
        r \cdot \RevSigProb[q]{\tau} + \sum_{s>\tau} \RevSigProb[q]{\REVSIG} 
        - \left(r'\cdot \RevSigProbP[q]{\tau'} + \sum_{\REVSIGP>\tau'} \RevSigProbP[q]{\REVSIGP} \right) 
        \ge 0,
    \end{equation}
    for any $q > \bar{q}$.

    Let $\gamma$ be the garbling from $\RevSigDist$ to $\RevSigDistP$ (see \cref{def:blackwell}), so that $\RevSigProbP[q']{\REVSIGP} = \sum_{\REVSIG} \RevSigProb[q']{\REVSIG} \cdot \gamma(\REVSIG, \REVSIGP)$ for any $q'$ and $\REVSIGP$. 
    Substituting the garbling-based characterization and changing the order of summation, the left-hand side of \cref{eq:blackwell_ineq} can be rewritten as
    \begin{align*}
        \lefteqn{r \cdot \RevSigProb[q]{\tau} + \sum_{s>\tau} \RevSigProb[q]{\REVSIG}
        - \left( 
        r' \cdot \sum_{\REVSIG} \RevSigProb[q]{\REVSIG} \cdot \gamma(\REVSIG, \tau')
        + \sum_{\REVSIGP>\tau'} \sum_{\REVSIG} \RevSigProb[q]{\REVSIG} \cdot \gamma(\REVSIG, \REVSIGP)  \right)}
        \\ & =
        r \cdot \RevSigProb[q]{\tau} + \sum_{s>\tau} \RevSigProb[q]{\REVSIG} 
        - \sum_{\REVSIG} \left( r'\cdot\gamma(\REVSIG, \tau') + \sum_{\REVSIGP>\tau'} \gamma(\REVSIG, \REVSIGP) \right) \cdot \RevSigProb[q]{\REVSIG}.
    \end{align*}
    
    For notational simplicity, let $h(\REVSIG)=  r'\cdot\gamma(\REVSIG, \tau') + \sum_{\REVSIGP>\tau'} \gamma(\REVSIG, \REVSIGP)$. 
    Because $h(s)$ is a probability (the probability of observing an instantiation of the Blackwell dominated signal that is accepted by $\ACCMAP[\tau', r']$ when the instance of the Blackwell dominating signal is $s$), $0\le h(\REVSIG)\le 1$. 
    We now break the summation over $s$ in the above equation into three summations, namely, the summation over $s<\tau$, $s=\tau$, and $s>\tau$. Combining the summations over the same subset of signals allows us to rewrite the preceding equation as
    \begin{align}
        \lefteqn{r \cdot \RevSigProb[q]{\tau} + \sum_{s>\tau}\RevSigProb[q]{\REVSIG} 
        - \sum_{\REVSIG} h(\REVSIG) \cdot \RevSigProb[q]{\REVSIG}} \nonumber 
       \\ & =
          (r-h(\tau)) \cdot \RevSigProb[q]{\tau} 
        + \sum_{s>\tau} (1-h(\REVSIG)) \cdot \RevSigProb[q]{\REVSIG} 
        - \sum_{s<\tau} h(\REVSIG) \cdot \RevSigProb[q]{\REVSIG}.\label{eq:blackwell_acc_prob_diff}
    \end{align}

    Next, we use the monotone likelihood ratio property.
    Let $\eta_s := \frac{\RevSigProb[q]{s}}{\RevSigProb[\bar{q}]{s}}$ denote the likelihood ratio for signal $s$. 
    Because signals have MLR, we have $\eta_{s'} > \eta_s$ for any $s'>s$.
    Using this monotonicity, we can bound \eqref{eq:blackwell_acc_prob_diff} as
    \begin{align*}
        \eqref{eq:blackwell_acc_prob_diff} 
        & = (r-h(\tau)) \cdot \eta_\tau \RevSigProb[\bar{q}]{\tau} 
        + \sum_{s>\tau} (1-h(\REVSIG)) \cdot \eta_s \RevSigProb[\bar{q}]{\REVSIG} 
        - \sum_{s<\tau} h(\REVSIG) \cdot \eta_s \RevSigProb[\bar{q}]{\REVSIG}
        \\ & \ge
        (r-h(\tau)) \cdot \eta_\tau \RevSigProb[\bar{q}]{\tau} 
        + \sum_{s>\tau} (1-h(\REVSIG)) \cdot \eta_\tau \RevSigProb[\bar{q}]{\REVSIG} 
        - \sum_{s<\tau} h(\REVSIG) \cdot \eta_\tau \RevSigProb[\bar{q}]{\REVSIG} 
        \\ & =
        \eta_\tau \cdot \left(
        r\cdot \RevSigProb[\bar{q}]{\tau} 
        + \sum_{s>\tau} \RevSigProb[\bar{q}]{\REVSIG} 
        - \sum_{\REVSIG}h(\REVSIG)\RevSigProb[\bar{q}]{\REVSIG}
        \right)
        \\ & =
        \eta_\tau \cdot 
        \left( \AccP{\ACCMAP[\tau,r]}{\bar{q}} - \AccP{\ACCMAP[\tau',r']}{\bar{q}} \right)
        \\ & = 0.
    \end{align*}
    Here, the inequality uses the monotone likelihood property separately for $s > \tau$ and $s < \tau$ (observing the signs of the multipliers of $\eta_s$), and the final step follows from the assumption that the acceptance probabilities of papers of quality $\bar{q}$ in both settings are equal. \Halmos
\endproof

With this lemma, we are able to prove \cref{prop:blackwell-threshold}.
\proof{Proof of \cref{prop:blackwell-threshold}.}  
For any responsive threshold acceptance policy $\ACCMAP[']$ in the second setting (the one that has weaker review signals), let $\vartheta = (\theta, r)$ be any of the author's equilibrium strategies to $\ACCMAP[']$.
% (As discussed in \cref{prop:de_facto}, there might be more than one best response.)
We will show the existence of a threshold acceptance policy $\ACCMAP$ in the first setting such that: 1) $\vartheta$ is also an equilibrium strategy to $\ACCMAP$; 2) the conference quality is the same in both settings; 3) the review burden induced by the policy-response pair $(\ACCMAP, \vartheta)$  in the first setting is at most the review burden induced by the policy-response pair  $(\ACCMAP', \vartheta)$ in the second setting; and 4) the author welfare induced by the policy-response pair $(\ACCMAP, \vartheta)$ in the first setting is at least the author welfare induced by the policy-response pair  $(\ACCMAP', \vartheta)$ in the second setting. If the above is true, then every point on the QB-tradeoff curve of the second setting is weakly dominated by a point on the QB-tradeoff curve of the first setting, which completes the proof. The same arguments apply to QA-tradeoff curves.
Note that it is sufficient to consider only responsive threshold policies because any policy that is not responsive either attracts no submissions (where the claim trivially holds) or accepts all submissions with a probability of 1 (which leads to a conference value smaller than 1 and thus is not a candidate threshold).

% First note that if $\ACCMAP[']$ always accepts or always rejects, then the proposition trivially holds.  
% Thus, we can assume that $\ACCMAP[']$ is non-trivial. 
By \cref{prop:de_facto}, the author's equilibrium strategy under a responsive threshold acceptance policy is a threshold strategy with some threshold $\theta$ (and probability $r$): the author will always submit a paper with quality $\Qual > \theta$, not submit if $\Qual < \theta$, and submit with probability $r$ if $\Qual = \theta$.
% Additionally, if $\vartheta$ is the strategy under which the author always submits everything, then letting $\ACCMAP$ be the policy that accepts everything minimizes the review burden without changing the conference quality. Therefore, without loss of generality, we can assume that not every paper is submitted under $\vartheta$.

Now, we consider two cases based on whether there are authors who are indifferent between submitting and not submitting.

\begin{enumerate}
\item The first case is that there exists a quality $\bar{q}$ such that an author with threshold strategy $\vartheta$ submits papers of quality $\bar{q}$ with probability $r$ in equilibrium. Because $\vartheta$ is a threshold strategy, there can be at most one such quality. Let $\rho = \rho(\theta, r)$ be the conference attractiveness under this setting. By \cref{prop:de_facto}, the acceptance probability at $\bar{q}$ must be exactly $1/\rho$.\fangcomment{Minor issue: this cannot be directly derived from \cref{prop:de_facto} }

Now, consider a threshold acceptance policy $\ACCMAP$ in the first setting that also induces $\vartheta$ as the author's best response. By construction, $\bar{q}$ is a candidate threshold. Therefore, the existence of $\ACCMAP$ is guaranteed by \cref{prop:threshold-policy}.
% Note that $\theta$ is a best response to $\phi'$ because by \gs{fill in ref here.} the acceptance rate is strictly monotone in the paper quality. 
Note that the conference quality and the attractiveness $\rho$ remain unchanged: papers with qualities strictly greater than $\bar{q}$ are all accepted, no paper with quality strictly less than $\bar{q}$ is accepted; and papers with quality $\bar{q}$ are submitted with probability $r$, and these papers are accepted i.i.d.~with probability $1/\rho$ in each setting. 
Therefore, by \cref{claim:blackwell-RB-better}, all submitted papers have a weakly larger acceptance probability under $\ACCMAP$ in the first setting, resulting in a weakly smaller review burden and a weakly larger author welfare.

\item In the second case, for any possible paper quality $q$, an author with threshold strategy $\vartheta$ either always or never submits papers of quality $q$.  
% Such a $\vartheta$ exists because we are assuming the categorical model, and we assumed that not every paper is submitted. 
Let $\bar{q}$ be the highest paper quality which is never submitted and let $\rho = \rho(\bar{q}, 0)$ be the conference attractiveness.

Let $\hat{\phi}'$ denote the threshold policy which accepts papers of quality $\bar{q}$ with probability $1/\rho$ in the second setting. 
Note that both $\ACCMAP'$ and $\hat{\phi}'$ induce $\vartheta$ as the author's best response, because we can assume that authors with papers of quality $\bar{q}$ never submit under $\hat{\phi}'$, given that $\AccP{\hat{\phi}'}{\bar{q}} = 1/\rho$, i.e., the authors are indifferent between submitting papers of quality $\bar{q}$ or not. 
This implies that $\hat{\phi}'$ induces the same conference quality as $\ACCMAP'$. Fixing the author's strategy $\vartheta$, the review burden of $\hat{\phi}'$ is no larger than that of $\ACCMAP'$ and the author welfare of $\hat{\phi}'$ is no less than that of $\ACCMAP'$. This is because by definition, $\ACCMAP'$ accepts papers of quality $\bar{q}$ with probability at most $1/\rho$, which is the acceptance probability of papers of quality $\bar{q}$ under $\hat{\phi}'$. Then, by \cref{lem:stricter_policy}, the acceptance probability of papers of any quality is weakly larger under $\hat{\phi}'$ than $\ACCMAP'$, resulting in a weakly smaller review burden and a weakly larger author welfare under $\hat{\phi}'$.

% First, we observe that the review burden of $\hat{\phi}$ is at most that of $\phi$ in setting one with author strategy $\theta$.  This is because the threshold of $\hat{\phi}$ is at most that of $\phi$.  Notice that $\phi$ accepts papers with quality $\bar{q}$ with probability at most $\rho$, because otherwise, $\theta$ would not be a best response.  If $\phi$ accepts papers with quality $\bar{q}$ with probability exactly $\rho$, then by definition, the threshold of $\hat{\phi}$ is at most that of $\phi$.   If $\phi$ accepts papers with quality $\bar{q}$ with probability strictly less than  $\rho$, then, the threshold of $\hat{\phi}$ must be strictly less than that of $\phi$ because acceptance probability is monotone in threshold.  \gs{should we cite something here?}

Finally, in the first setting, let $\ACCMAP$ be a threshold policy that induces a threshold best response with $(\bar{q}, 0)$ for the author such that authors with paper quality $\bar{q}$ prefer not to submit and $\AccP{\ACCMAP}{\bar{q}} = 1/\rho(\bar{q}, 0)$. This implies that the conference attractiveness is $\rho = \rho(\bar{q}, 0)$, the same as in the second setting. Because $\ACCMAP$ is a responsive threshold policy, by \cref{prop:de_facto}, $\bar{q}$ is a de facto threshold. Therefore, the existence of $\ACCMAP$ is guaranteed by \cref{prop:threshold-policy}. By this construction, first note that $\vartheta$ is also a best response to $\ACCMAP$.
Second, the conference quality under $\ACCMAP$ is identical to the conference quality under $\ACCMAP'$ because the same set of papers is submitted, and all are eventually accepted. Finally, by \cref{claim:blackwell-RB-better}, because papers of quality $\bar{q}$ are accepted with probability $1/\rho$ under $\ACCMAP$ in the first setting and under $\hat{\phi}'$ in the second setting, all papers with quality larger than $\bar{q}$ have a larger acceptance probability in the first setting. Therefore, the review burden of $\ACCMAP$ in the first setting is no larger than the review burden of $\hat{\phi}'$ which is no larger than the review burden of $\ACCMAP'$ in the second setting. The same argument holds analogously for author welfare. This completes the proof. \Halmos
\end{enumerate}
\endproof

\subsection{Proof of \cref{lemma:query burden value}}\label{app:proof-QB-tradeoff}

% To verify the definition of the dominance of QB-tradeoff curves, we will show that for any point on $\mathcal{C'}$ that neither corresponds to accepting all papers nor to rejecting all papers, there is a point on $\mathcal{C}$ that has the same conference quality but strictly smaller review burden.

We first consider the continuous model. Fixing a de facto threshold $\theta$, whose existence as a candidate threshold is ensured by \cref{prop:de_facto}), the conference value $\ConfValue(\theta)$ is fixed. Therefore, the conference attractiveness factor in the setting with a larger $\TD$, which is $\rho = \frac{\ConfValue(\theta) - \TD}{1-\TD}$, is larger than that in the setting with a smaller $\rho'$.

Next, we show that the setting with a larger attractiveness factor $\rho$ has a QB-tradeoff curve dominated by the setting with a smaller factor $\rho' < \rho$.

By \cref{prop:threshold-policy}, both settings have respective unique corresponding non-trivial threshold acceptance policies with thresholds $\tau$ and $\tau'$. By \cref{prop:de_facto}, a paper with quality $\theta$ is accepted with probabilities $1/\rho$ and $1/\rho'$, respectively, in the two settings.  
This means that papers with quality $\theta$ are accepted with strictly smaller probability by $\ACCMAP[\tau]$ than $\ACCMAP[\tau']$. 
By \cref{lem:stricter_policy}, $\ACCMAP[\tau]$ is stricter than $\ACCMAP[\tau']$.
Thus, for papers of every quality $q$, the expected number of rounds that a paper of quality $q$ has to be resubmitted is strictly larger in the setting with attractiveness factor $\rho$, which implies a strictly larger review burden.
Therefore, the QB-tradeoff curve in the setting with a larger attractiveness factor, is dominated by the QB-tradeoff curve in the other setting.

The argument for categorical models is essentially the same. We fix the authors' threshold strategy $(\bar{q}, r)$ such that authors with papers of quality $\Qual>\bar{q}$ submit, authors with $\Qual<\bar{q}$ take the outside option, and authors with $\Qual = \bar{q}$ submit with probability $r$. This fixes the conference quality and the conference value $\ConfValue(\bar{q}, r)$. 

We want to show that the acceptance policy which induces such an author's best response has a smaller review burden in the setting with $\rho'$ than in the setting with $\rho>\rho'$. Without loss of generality, suppose $r<1$. If $r\in (0,1)$, by \cref{lem:author_response}, the acceptance probability $\AccP{\ACCMAP}{\bar{q}} = 1/\rho < \AccP{\ACCMAP'}{\bar{q}} = 1/\rho'$. This means that the acceptance policy $\ACCMAP$ is stricter than $\ACCMAP'$. Based on the same argument as in the continuous model, this means that $\ACCMAP$ leads to a larger review burden. 
If $r = 0$, there might be multiple acceptance policies that induce the same threshold equilibrium. Let $\ACCMAP$ be any one of them. Because papers of quality $\bar{q}$ are not submitted, we must have $\AccP{\ACCMAP}{\bar{q}} \leq 1/\rho$. Now, let $\ACCMAP[']$ be a threshold acceptance policy in the setting with attractiveness $\rho'$ which accepts papers of quality $\bar{q}$ with probability $1/\rho'$. Such a policy $\ACCMAP[']$ exists by Proposition~\ref{prop:threshold-policy}, i.e., $\ACCMAP[']$ induces $\bar{q}$ as a de facto threshold. Because $1/\rho' > 1/\rho$, by \cref{lem:stricter_policy}, $\ACCMAP$ is stricter than $\ACCMAP[']$. Again, based on the same argument, this completes the proof.
\Halmos

\subsection{Proof of \cref{lem:QA_eta}}\label{app:proof-QA-eta}

% \dkcomment{Also, why are we writing $\partial$? Why not just $d$? Isn't everything just scalar?}\yzcomment{chanegd}

  Define $g(\TD) = \frac{1-\TD}{\AccP{\tau(\TD)}{q}} + \TD$. 
  Because $u^{(a)}(q,\TD) = \frac{\ConfValue}{g(\TD)}$, we know that $u^{(a)}$ is increasing in $\TD$ if and only if $g$ is decreasing in $\TD$.
  Note that $\AccP{\tau(\TD)}{q} = 1 - \REVNOISEDIST(\tau(\TD) - q)$ so that $\frac{d \AccP{\tau(\TD)}{q}}{d \TD} = - f^{(r)}(\tau(\TD) - q) \cdot \frac{d \tau(\TD)}{d \TD}$.
    
    Recall that $\tau(\TD)$ is chosen so that authors with a paper of quality $\theta$ are indifferent between submitting and taking the outside option, i.e., $\AccP{\tau(\TD)}{\theta} = \frac{1-\TD}{\ConfValue-\TD}$. Taking the derivative of both sides of the equation w.r.t.~$\TD$, and combining with the previous derivative calculation, we have $\frac{d \tau(\TD)}{d \TD} = \frac{\ConfValue-1}{(\ConfValue-\TD)^2 \cdot f^{(r)}(\tau(\TD) - \theta)}$.

    Then, taking the derivative of $g$ w.r.t.~$\TD$ yields
    \begin{align*}
        g'(\TD) &= 1 + \frac{-(1-\REVNOISEDIST(\tau(\TD) - q)) + (1-\TD) \cdot \frac{(\ConfValue-1)f^{(r)}(\tau(\TD) - q)}{(\ConfValue-\TD)^2 \cdot f^{(r)}(\tau(\TD) - \theta)}}{(1-\REVNOISEDIST(\tau(\TD) - q))^2}.
    \end{align*}
    $g'(\TD)$ is positive if and only if 
    \begin{align*}
        (1-\REVNOISEDIST(\tau(\TD) - q))^2 -(1-\REVNOISEDIST(\tau(\TD) - q)) + (1-\TD) \cdot \frac{(\ConfValue-1) \cdot f^{(r)}(\tau(\TD) - q)}{(\ConfValue-\TD)^2 \cdot f^{(r)}(\tau(\TD) - \theta)} 
        & > 0,
    \end{align*}
    which can be rearranged to
    \begin{align*}
    \frac{(1-\TD) \cdot (\ConfValue-1) \cdot f^{(r)}(\tau(\TD) - q)}{(\ConfValue-\TD)^2 \cdot f^{(r)}(\tau(\TD) - \theta)} 
    & > \REVNOISEDIST(\tau(\TD) - q) \cdot (1-\REVNOISEDIST(\tau(\TD) - q)).
    \end{align*}
    Because $\frac{1-\TD}{\ConfValue-\TD} = 1 - \REVNOISEDIST(\tau(\TD) - \theta)$ and $\frac{\ConfValue-1}{\ConfValue-\TD} = \REVNOISEDIST(\tau(\TD) - \theta)$, the preceding inequality can be further rearranged to
    \begin{align*}
       \frac{\REVNOISEDIST(\tau(\TD) - \theta) \cdot (1-\REVNOISEDIST(\tau(\TD) - \theta))}{f^{(r)}(\tau(\TD) - \theta)} 
       & > \frac{\REVNOISEDIST(\tau(\TD) - q) \cdot (1-\REVNOISEDIST(\tau(\TD) - q))}{f^{(r)}(\tau(\TD) - q)},
    \end{align*}
    which holds if and only if $h(q) > h(\theta)$.

Therefore, the author's utility is increasing in $\TD$ if and only if $h(q) < h(\theta)$. By reversing all inequalities in the preceding calculations (i.e., comparison of the derivative with 0), we can show that the author's utility is decreasing in $\TD$ if and only if $h(q) > h(\theta)$, and remains unchanged in $\TD$ if $h(q) = h(\theta)$.\Halmos

%% file: OR-resbumission/sections/When_optimal.tex
Our analysis of QB-tradeoffs in \cref{subsec:gap-tradeoffs} is largely based on threshold acceptance policies. To contextualize this focus, we first investigate when using such policies is optimal in terms of achieving Pareto frontiers of the QB-tradeoff.
%In \cref{prop:de_facto}, we showed that every monotone acceptance policy induces a de facto threshold above which all papers are submitted (and eventually accepted), and below which no papers are submitted. Then, in \cref{prop:threshold-policy}, we showed that every de facto threshold can be implemented with a threshold acceptance policy. In particular, this applies to the optimal de facto threshold of $\theta=0$\textemdash see \cref{prop:max-general}.
%While this implies that every desired conference quality \yichidelete{(including the optimal conference quality)} that is achievable by a monotone policy can be achieved by a threshold policy, it raises the question whether insisting on a threshold acceptance policy might come at a cost in terms of the \emph{review burden}. 
%Here, we answer this question. We show that when only $\NumReviews = 1$ review is obtained, there always is an optimal threshold acceptance policy. 
%However, even for $\NumReviews = 2$, there are in general instances for which monotone non-threshold policies strictly outperform threshold policies in terms of the review burden, while achieving the same (optimal) conference quality.
While analyzing monotone acceptance policies, we showed that they always induce a threshold strategy for authors as one of the best responses (see \cref{prop:threshold-policy}). Furthermore, under threshold acceptance policies, non-threshold best responses do not exist; and even for monotone policies, they exist only in knife-edge cases, where authors are indifferent between submitting or not submitting for multiple paper qualities. 
We therefore restrict our attention to threshold best responses for the authors in the following discussions. 

In \cref{prop:threshold-policy}, we showed that every candidate threshold can be implemented with a threshold acceptance policy. In particular, this applies to the optimal de facto threshold of $\theta=0$; see \cref{prop:max-general}. Therefore, every conference quality that is achievable by a monotone policy can be achieved by a threshold policy. 
It remains to understand when a threshold acceptance policy minimizes the corresponding review burden (among all monotone policies) for every conference quality. Our results suggest that when only $\NumReviews = 1$ review is obtained, there always is an optimal threshold acceptance policy. 
However, even for $\NumReviews = 2$, there are in general instances for which monotone non-threshold policies strictly outperform threshold policies in terms of the review burden, while achieving the same (optimal) conference quality.

\begin{proposition}\label{prop.threshold_opt}
When $\NumReviews = 1$ and authors always respond with threshold strategies if possible, the QB-tradeoff for threshold policies weakly dominates that of all monotone acceptance policies.
\end{proposition}
Our proof will use the following theorem about uniformly most powerful tests in statistical testing.

\begin{theorem}[Karlin–Rubin Theorem (Theorem~12.9 of \cite{keener2010theoretical})] \label{thm.ump}  Consider two sets $\QualSet$ and $\SigSet\subseteq \R$ and a family of pdfs or pmfs $\{g(\REVSIG |\qual):\qual\in \QualSet\}$ on $\SigSet$ possessing the monotone likelihood property (\cref{def:informative}).  
For any $\theta\in \QualSet$, $\tau\in \R$, $r\in [0,1]$, and two (randomized) tests $h, h^*: \SigSet\to [0,1]$, if $h^*$ is a threshold function, i.e., of the form
$$h^*(\REVSIG) = \begin{cases} 1\text{ if }\REVSIG>\tau\\
r \text{ if }\REVSIG = \tau\\
0 \text{ if }\REVSIG<\tau\end{cases},$$ and 
$\sup_{\qual\in \Theta_0}\E_{\REVSIG\sim g(\cdot|\qual)}[h^*(\REVSIG)]\ge \sup_{\qual\in \Theta_0}\E_{\REVSIG\sim g(\cdot|\qual)}[h(\REVSIG)]$ with $\Theta_0:=\{\qual\in \QualSet: \qual\le \theta\}$, then 
$$\E_{\REVSIG\sim g(\cdot|\qual)}[h^*(\REVSIG)]\ge \E_{\REVSIG\sim g(\cdot|\qual)}[h(\REVSIG)]\text{, for all } \qual\in \QualSet\setminus \Theta_0.$$

We call $h^*$ \emph{uniformly most powerful} among all tests $h$ with 
$$\sup_{\qual\in \Theta_0}\E_{\REVSIG\sim g(\cdot|\qual)}[h^*(\REVSIG)]\ge \sup_{\qual\in \Theta_0}\E_{\REVSIG\sim g(\cdot|\qual)}[h(\REVSIG)].$$
\end{theorem}

\proof{Proof of Proposition~\ref{prop.threshold_opt}.}

As discussed after \cref{prop:threshold-policy}, by assuming that authors break ties in favor of using threshold strategies, any conference quality that can be achieved with a monotone policy is 
% $\int^\infty_\theta q\,dp(q)$ for some de facto threshold $\theta\in \R$, 
determined by some de facto threshold $\theta$ (or a threshold-probability pair $(\theta, r)$ in the categorical model), and the quality can be achieved with a threshold acceptance policy. 
Thus, it is sufficient to show that for any de facto threshold $\theta$ (or $(\theta, r)$ in the categorical model) there exists a threshold policy that minimizes the review burden over all monotone policies.  

Given an attractiveness $\rho$ induced by the threshold best response, by \cref{prop:de_facto}, a monotone policy with de facto threshold $\theta$ can be considered as
a statistic test on $\{q \in \QualSet: q \le \theta\}$ vs.~$\{q \in \QualSet: q > \theta\}$. 
% \dkcomment{I think we're interested in the actual values here, not the random variable.}
% \dkreplace{so that 1) the size is less than $1/\rho$ and 2) the power is greater than $1/\rho$ for all papers with quality greater than $\theta$.  When the review signal is informative (monotone likelihood ratio property), by Karlin–Rubin theorem~(e.g., Theorem 8.3.17 of \cite{keener2010theoretical}),
% there exists a uniformly most powerful \fangreplace{test $\ACCMAP^*$ at size $1/\rho$}{ size $1/\rho$ test $\ACCMAP^*$} so that for any acceptance policy $\ACCMAP$ with size less than or equal to $1/\rho$ and any paper quality $q>\theta$, $\AccP{\ACCMAP^*}{q}\ge \AccP{\ACCMAP}{q}$.  Moreover, $\ACCMAP^*$ has a critical value $c$ and accepts a paper when the signal is greater than $c$.  Therefore, $\ACCMAP^*$ has the same de facto threshold $\theta$ and minimizes review burden.  Finally, by the definition of monotone likelihood ratio property with $m = 1$, the $\ACCMAP^*$ is a threshold policy.}
To ensure that authors do not submit papers of quality less than $\theta$, the false positive rate (accepting any paper of quality $q < \theta$) must be at most $1/\rho$. Subject to those constraints, the false negative rate (rejecting any paper of quality $q > \theta$) should be minimized.  By \cref{prop:threshold-policy}, there exists a threshold policy $\ACCMAP^*$ with de facto threshold $\theta$.
The Karlin–Rubin Theorem (as given in \cref{thm.ump}) states that when the observed signal (in our case: the review) has the monotone likelihood ratio property, 
% \fangreplace{there exists a uniformly most powerful test $\ACCMAP^*$ which is a threshold test subject to the constraint on false positive, and it satisfies}
the above threshold test (policy) $\ACCMAP^*$ is uniformly most powerful so that $\AccP{\ACCMAP^*}{q}\ge \AccP{\ACCMAP}{q}$ for all $q > \theta$.  
Thus, $\ACCMAP^*$ has the desired de facto threshold $\theta$ and minimizes the probability of false negatives (rejections) for all paper qualities that should be eventually accepted.
Thus, it minimizes the review burden. \Halmos
\endproof

The optimality of threshold policies ceases to hold when $\NumReviews > 1$, as we illustrate next with a counter-example.

\begin{table}[ht]
\caption{Likelihood function in \cref{ex:threshold-counterexample}. \label{tab:threshold-counterexample}}
    \centering
     \renewcommand{\arraystretch}{1.2}
    \begin{tabular}{|l|c|c|c|c|c|c|}
    \hline
 $\RevSigV$ & $(L,L)$ &     
 \begin{tabular}{@{}c@{}}$(L,M)$ or\\ $(M,L)$ \end{tabular}
 & \begin{tabular}{@{}c@{}}$(L,H)$ or\\ $(H,L)$ \end{tabular} 
 & $(M,M)$ 
 &  \begin{tabular}{@{}c@{}}$(M,H)$ or\\ $(H,M)$ \end{tabular} 
 & $(H,H)$ \\
         \hline
        $\ProbC{\RevSigV}{q = -2}$ & $\frac{4}{9}$ & $\frac{2}{9}$ & $\frac{2}{9}$ & $\frac{1}{36}$ & $\frac{1}{18}$ & $\frac{1}{36}$\\
        $\ProbC{\RevSigV}{q = 1}$ & $\frac{1}{9}$ & $\frac{1}{9}$ & $\frac{1}{3}$ & $\frac{1}{36}$ & $\frac{1}{6}$ & $\frac{1}{4}$\\
        $\ProbC{\RevSigV}{q = 5}$ & $\frac{1}{36}$ & $\frac{1}{18}$ & $\frac{2}{9}$ & $\frac{1}{36}$ & $\frac{2}{9}$ & $\frac{4}{9}$\\
        \hline
    \end{tabular}
\end{table}

% \begin{table}
%      \TABLE
%      {Likelihood function in \cref{ex:threshold-counterexample}.\label{tab:threshold-counterexample}}
%      \begin{tabular}{|c|c|c|c|c|c|c|}
%           \hline
%          $\RevSigV$ & $(L,L)$ &     
%          \begin{tabular}{@{}c@{}}$(L,M)$ or\\ $(M,L)$ \end{tabular}
%          & \begin{tabular}{@{}c@{}}$(L,H)$ or\\ $(H,L)$ \end{tabular} 
%          & $(M,M)$ 
%          &  \begin{tabular}{@{}c@{}}$(M,H)$ or\\ $(H,M)$ \end{tabular} 
%          & $(H,H)$ \\
%          \hline
%         $\ProbC{\RevSigV}{q = -2}$ & $\frac{4}{9}$ & $\frac{2}{9}$ & $\frac{2}{9}$ & $\frac{1}{36}$ & $\frac{1}{18}$ & $\frac{1}{36}$\\
%         $\ProbC{\RevSigV}{q = 1}$ & $\frac{1}{9}$ & $\frac{1}{9}$ & $\frac{1}{3}$ & $\frac{1}{36}$ & $\frac{1}{6}$ & $\frac{1}{4}$\\
%         $\ProbC{\RevSigV}{q = 5}$ & $\frac{1}{36}$ & $\frac{1}{18}$ & $\frac{2}{9}$ & $\frac{1}{36}$ & $\frac{2}{9}$ & $\frac{4}{9}$\\
%         \hline
%     \end{tabular}
%     \vspace{2mm}
%      {Columns are ordered based on expected conference quality (from low to high) conditioned on the review vector.}
% \end{table}

\begin{example} \label{ex:threshold-counterexample}
The set of paper qualities is $\{-2,1,5\}$, with uniform prior $\QualDist = (1/3,1/3,1/3)$.  The review signal set is $\SigSet = \{L,M,H\}$, and the number of reviews is $\NumReviews = 2$.  The conditional distributions of review signals are $\RevSigDist[-2] = (2/3,1/6,1/6), \RevSigDist[1] = (1/3,1/6,1/2)$, and $\RevSigDist[5] = (1/6,1/6,2/3)$; it can be verified that this information structure satisfies the monotone likelihood ratio (MLR) property. 
\Cref{tab:threshold-counterexample} shows the resulting likelihood of each vector of review signals, with $\NumReviews = 2$. The columns of the table are ordered based on expected conference quality (from low to high) conditioned on the review vector.

When only positive-quality papers ($q = 1, 5$) are submitted, the conference value is $V=3$. Under an author discount factor of $\eta = 9/19$, the conference's attractiveness is then $\rho = 24/5$.

Consider the following policy $\ACCMAP[']$. 
(1) When the review vectors are $(H,H), (M,H)$ or $(H,M)$,
accept the paper with probability 1;
(2) when the review vectors are $(M,M), (L,H)$ or $(H,L)$,
accept the paper with probability $1/2$; and
(3) otherwise, reject the paper with probability 1.
Notice that this policy is monotone (because the acceptance probability is non-decreasing in the conditional expected quality), but it is not a threshold policy. 
This is because the conditional expected paper qualities of the signal vectors $(M,M)$ and $(L,H)$ (or $(H,L)$) are
$U(M,M) = \frac{4}{3} > \frac{9}{7} = U(H,L) = U(L,H)$,
yet $(H,L)$ and $(L,H)$ lead to acceptance with positive probability, while $(M,M)$ does not lead to acceptance with probability 1.
Next, we compute the acceptance probabilities of papers with different qualities:
\begin{align*}
\AccP{\ACCMAP[']}{5} & = 19/24 > 1/\rho
& \AccP{\ACCMAP[']}{1} & = 43/72 > 1/\rho
& \AccP{\ACCMAP[']}{-2} & = 5/24 = 1/\rho.
\end{align*}

Thus, all papers of positive quality are submitted, while the papers of negative quality are not submitted. As a result, $\ACCMAP[']$ is a monotone policy with de facto threshold 0, maximizing the conference quality.  

Next, consider any threshold policy $\ACCMAP$ implementing the de facto threshold of $\theta=0$. It cannot accept papers with review vectors $(L,H), (H,L)$ with probability exceeding $7/16$. 
Otherwise, by virtue of being a threshold policy, $\ACCMAP$ would have to accept all papers with higher review vectors with probability 1; as a result, the acceptance probability of a paper with quality $-2$ would exceed $7/16 \cdot 2/9 + 1 \cdot (1/36 + 1/18 + 1/36) = 5/24$.

Again by virtue of being a threshold policy, $\ACCMAP$ must reject all papers with review vectors $(L,L), (L,M), (M,L)$. 
Thus, a paper of quality $5$ is accepted with probability at most $1 \cdot (4/9 + 2/9 + 1/36) + 7/16 \cdot 2/9 = 19/24$, while a paper of quality $1$ is accepted with probability at most $1 \cdot (1/4+1/6+1/36) + 7/16 \cdot 1/3 = 85/144$.
Thus, papers of quality $5$ are accepted with the same probability as under the policy $\ACCMAP[']$, while papers of quality $1$ are accepted with strictly smaller probability. As a result, a strictly higher review load is required.
\end{example}

We further note that \cref{ex:threshold-counterexample} also shows that the combination of independent review signals with MLR do not necessarily have MLR. To see this, in \cref{tab:threshold-counterexample}, $(M, M)$ is a better review signal than $(L, H)$ or $(H, L)$ in the sense that the expected paper quality is higher conditioned on the former. However, it is not hard to observe that these two signals violate the definition of MLR. In particular, $\frac{\RevSigDist[1]((M,M))}{\RevSigDist[-2]((M,M))}= 1 < 1.5 = \frac{\RevSigDist[1]((L,H)\text{ or }(H,L))}{\RevSigDist[-2]((L,H)\text{ or }(H,L))}$.

%% file: OR-resbumission/sections/Noisy_ABM.tex
Thus far, we have examined how parameters and acceptance policies shape the conference’s QB tradeoff under the assumptions that authors perfectly observe their papers’ quality and may resubmit indefinitely until accepted. Under these conditions, we show that authors adopt a straightforward threshold best response, which enables a rigorous theoretical analysis of the problem.

When authors receive noisy signals themselves, they \emph{will} update their beliefs about their papers' quality based on the reviews they receive. For example, an author who initially received a signal indicating that her paper was of high quality may revise this estimate downward after receiving several negative reviews.
As a result, authors may not make the same decision in each iteration; while it may initially be utility-maximizing to submit the paper, after several negative reviews, the author may instead take the side option.

In this section, we focus on a more realistic setting where authors have noisy signals and the model is a real-data estimated categorical model.
After all, understanding the behavior of more realistic authors is also an important robustness check on our results.
Unfortunately, computing the author's posterior belief of the paper quality, conditioned on the reviews in each round, is analytically intractable. 
Therefore, we use agent-based simulations (agent-based modeling (ABM)) to evaluate the impact of conference policies on outcomes.

\subsection{Agent-Based Model Setup}
\label{subsec:noisy_setup}

\subsubsection{Categorical Models Studied}

The \emph{$(\lambda_A, \lambda_R, \NumReviews, \TD)-\text{ICLR}^{y, L}$ model} is a categorical model learned from data.
Specifically, the prior of paper quality $\QualDist^*$ and the review signal distribution $\RevSigDist^*$ are learned from the ICLR OpenReview datasets of year $y$ [\cite{iclr2020review,iclr2021review}]; for each dataset, a model is learned with paper quality sets of sizes $L=6$. 
We use the expectation-maximization (EM) algorithm to estimate and update the model parameters. 
Details of the method used for learning are described in \cref{sec:learning-parameters}, where the learned parameters are shown in \cref{tab:learned_para}. 
% \dkcomment{Is it a good idea to not put them here? Is this just a remnant of the earlier structure where this section was in the main body, or do we think that having basically an appendix to an appendix is better? I think that a subsection at the end of the present section would be better --- readers can still skip it if they want.}
% \yzcomment{I'll mark this, but may address this with a lower priority.}
% \dkcomment{I moved it into Subsection 3.5. Does this work?}

When varying the signal quality in models based on ICLR data, we use accuracy parameters $\lambda_A \in [0, 1]$ for the author's signal and $\lambda_R \in [0, 1]$ for the reviewer's signal; the parameters control the probability with which the signal is drawn from the learned parameters $\RevSigDist^*$ vs.~a uniform distribution. $\lambda_A=1$ ($\lambda_R=1$) implies that the signal is drawn from $\RevSigDist^*$ (the same for both authors and reviewers) while $\lambda_A=0$ ($\lambda_R=0$) implies that the signal is uniformly random.
% When we model noiseless authors, we simply remove the entry $\lambda_A$ from the tuple and call this the $(\lambda_R, \NumReviews, \ConfValue, \TD)-\text{ICLR}^{y, L}$ noiseless-author model.

We remark that the learned distributions $\RevSigDist^*$ of reviewer signals do not strictly satisfy the monotone likelihood ratio property as defined in \cref{def:informative}. However, as we will see, the properties that we are interested in still (approximately) follow.

\subsubsection{Myopic Authors and Agent-Based Simulations}

% \fangcomment{I try to use our existing macros to avoid inconsistency.}
In our ABM experiments, we simulate the submission-review process for $T = 15$ rounds with $\NumNewPapers = 1000$ new papers per round. We give a brief recap of our model from \cref{sec:model}; recall that the dynamics has two major phases: submission and reviewing.

In the first phase, each author updates her belief about her paper based on her private signal that is generated from the distribution $\AuthSigDist[q]$, and the reviews from the previous rounds (if any). At each round, given the belief, the author makes the decision of either submitting to the conference or taking the outside option, depending on the expected utility of each option.
We assume that authors are myopic, meaning that they compute the expected utility of submitting to the conference while assuming that they will take the outside option in the next round immediately if the paper is rejected this round. That is, a myopic author in round $t$ does not foresee the future after round $t+1$. The myopic strategy asks: is it better to submit one more round before giving up or to give up now? While it can happen that submitting two more rounds before giving up is better than giving up now which in turn is better than giving up after one round of submission, such cases are quite rare when the same threshold acceptance policy is used in each round.

The author's utility depends on the conference value $\ConfValue$ which in turn depends on the authors' response. In our ABM, we initialize the conference value $\ConfValue_1 = 2$ in the first round and iteratively update $\ConfValue_t = (1-\lambda_V) \cdot \ConfValue_{t-1} + \lambda_V \cdot\left(1 + \ExpectC{Q_i}{\text{paper } i \text{ was accepted in round } t}\right)$. We set the updating rate $\lambda_V = 0.5$ and compute the expected quality directly using the average quality of the papers accepted in the previous round.

In the second phase, the conference obtains $\NumReviews$ reviews for each submission. These reviews are drawn i.i.d.~from the review signal distribution $\RevSigDist[\dkedit{q}]$, conditioned on the ground truth quality \dkedit{$q$} of the submission. 
Then, for each submission, the conference makes a decision of acceptance or rejection based on a threshold policy with threshold $\tau$. Given the $\NumReviews$ i.i.d.~reviews, denoted as $\RevSigV$, the conference can infer the expected quality of the submission $U(\RevSigV) = \ExpectC[\Qual]{\Qual}{\RevSigV, \QualDist, \RevSigDist}$ (conditioned only on the reviews from the current round) and accept or reject the paper based on the simplified threshold policy for the categorical model described in \cref{app:simplified_threshold}.

% These two phases are repeated for $N = 15$ rounds, giving time for the conference value to be updated based on the average quality of accepted papers. After $N$ rounds, we compute the conference quality by summing all the papers that are accepted in round $N$ and taking the average over $n$; we compute the review burden by summing the total number of reviews assigned to each paper that is submitted or resubmitted in round $N$ and averaging over $n$; we compute author social welfare by summing the utility of all authors who have submitted a paper in round $N$ and averaging over $n$. All metrics exclude any papers accepted (or withdrawn via the outside option) prior to round $T$.

We repeat these two phases for $N = 15$ rounds, allowing the conference value and the author submission strategy to dynamically update based on each other. At the end of round $N$, we compute:
\begin{itemize}
    \item \textbf{Conference quality}, as the sum of the quality of papers accepted in round $N$, normalized by $\NumNewPapers$;
    \item \textbf{Review burden}, as the total number of reviews assigned to papers submitted or resubmitted in round $N$, normalized by $\NumNewPapers$;
    \item \textbf{Author welfare}, as the sum of author utilities who submitted in round $N$, normalized by $\NumNewPapers$. Note that this counts the utilities of authors whose papers are accepted by the prestigious conference, excluding the utilities of authors who decide to take the outside option.
\end{itemize}
In all metrics, we exclude papers accepted or withdrawn (via the outside option) before round $T$.

\subsection{QB-tradeoffs With Noisy Authors}

Fixing the model parameters, we vary the acceptance threshold $\tau\in [-2,2]$ with a step size of $0.1$. Each acceptance threshold leads to a corresponding conference quality and review burden, allowing us to trace the QB-tradeoff curve.

As shown in \cref{fig:QB_tradeoff_noisy}(b), when $\tau$ is close to $-2$, the conference accepts every paper with almost a single round of review, leading to a review burden of approximately $\NumReviews$ (slightly higher than $\NumReviews$ because $\tau = -2$ is still larger than the minimum quality, meaning that clearly bad papers still have a positive probability to be rejected).
These thresholds correspond to the dots close to $(\NumReviews, 0.22)$, where $0.22$ is the conference quality when all papers are submitted and accepted.
Similar to \cref{fig:Pareto_frontier_continuous}, as we raise the acceptance threshold, both the conference quality and the review burden increase until the conference quality is maximized. 
Increasing $\tau$ even further will reduce the conference quality and increase the review burden, resulting in a dominated point on the QB-tradeoff curve.
Interestingly, this pattern is similar to panel (b) of \cref{fig:Pareto_frontier_continuous} with a large review noise, where the Pareto optimal policies are achieved by adopting low acceptance thresholds.

Note that\fangcomment{The reader cannot notice this because we do not display $\tau$.  It would be good if we can also show the location of $\tau = -2, -1, 0, 1, 2$} the quality-maximizing $\tau$ does not necessarily induce a de facto threshold of 0 when authors are noisy. This is because the authors' strategy is more complicated when they observe noisy signals, and thus, the original definition of de facto threshold is not suitable for the new setting.

\begin{figure}
     \FIGURE
     {\begin{subfigure}[b]{0.46\textwidth}
         \centering
         \includegraphics[width=\textwidth]{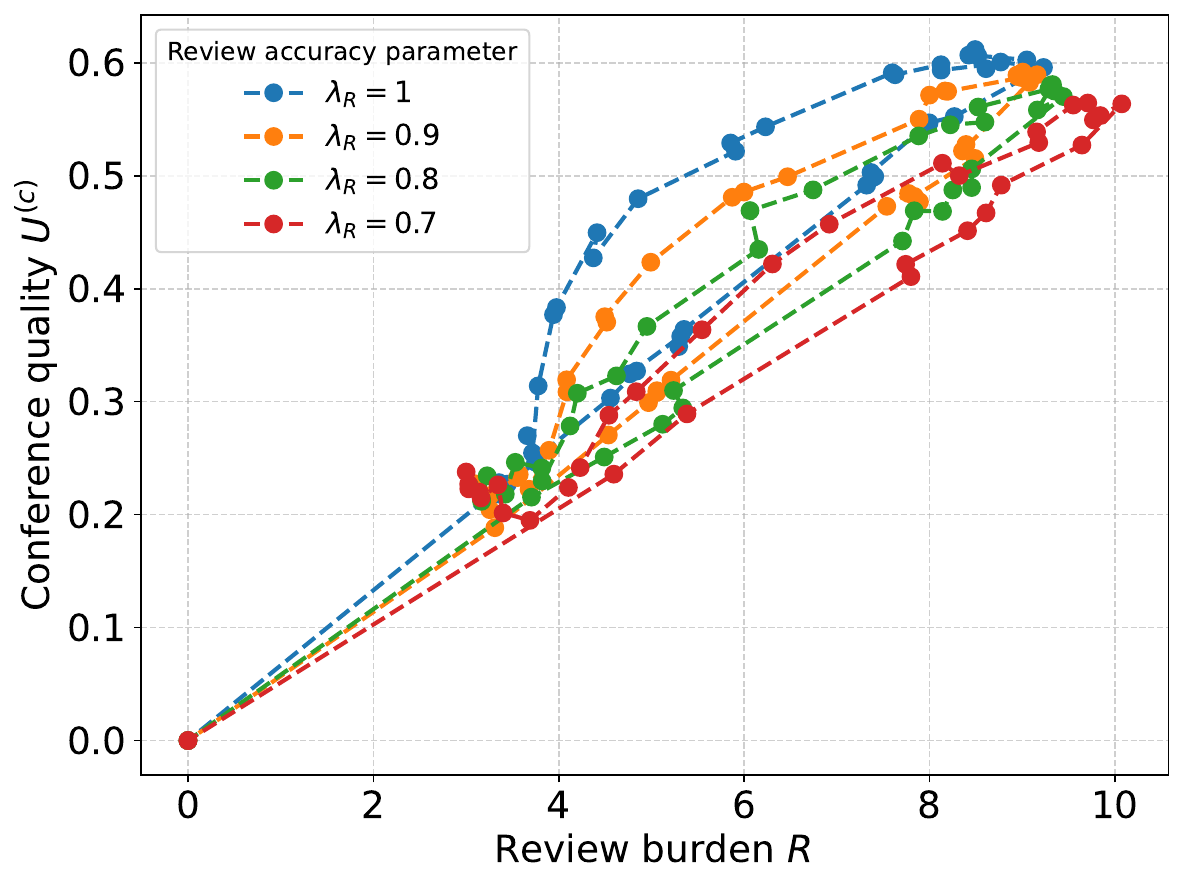}
         \captionsetup{size=}
         \caption{QB-tradeoff varying $\lambda_R$.}
     \end{subfigure}
     % \begin{subfigure}[b]{0.32\textwidth}
     %     \centering
     %     \includegraphics[width=\textwidth]{Plots/QB_tradeoff_noisy_L6_vary_alpha.pdf}
     %     \captionsetup{size=}
     %     \caption{QB-tradeoff varying $\lambda_A$.}
     % \end{subfigure}
     \begin{subfigure}[b]{0.46\textwidth}
         \centering
         \includegraphics[width=\textwidth]{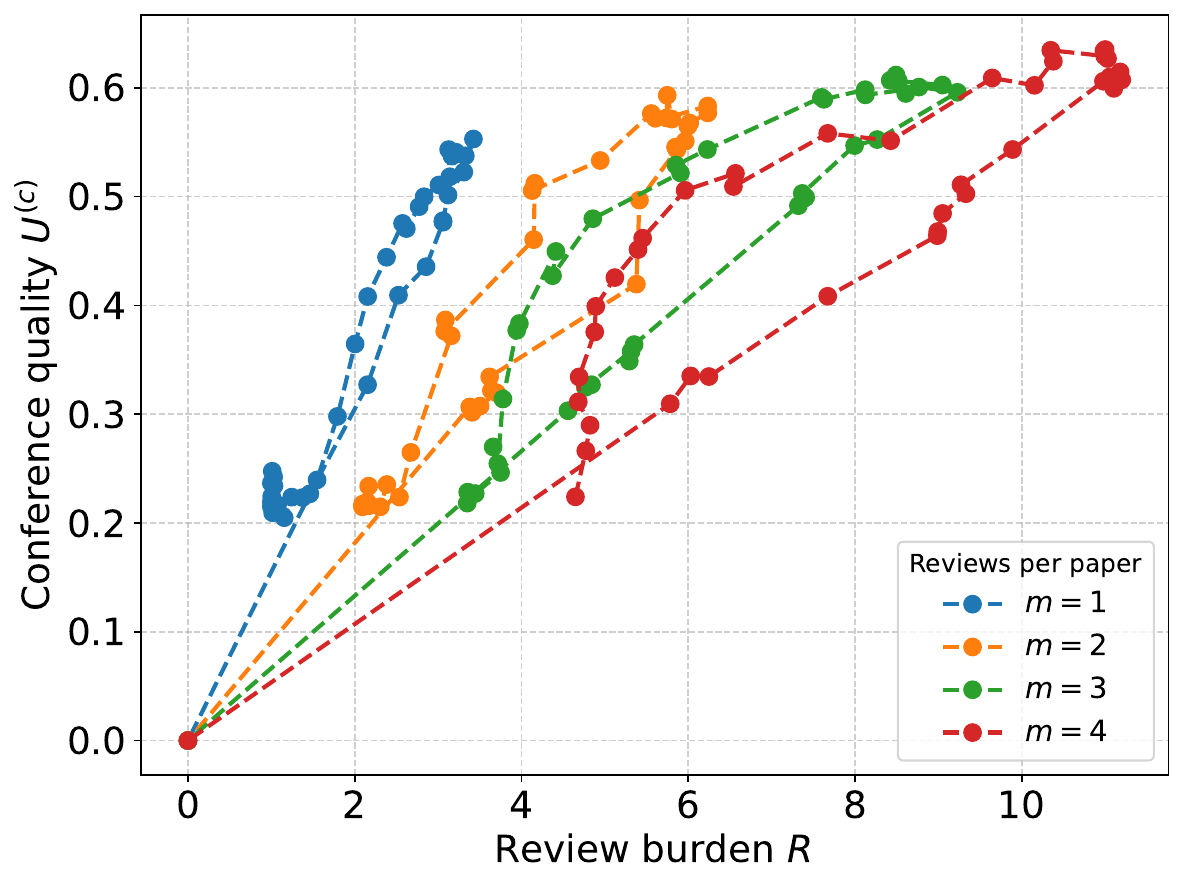}
         \captionsetup{size=}
         \caption{QB-tradeoff varying $m$.}
     \end{subfigure}
     }
     {QB-tradeoff Curves under Reviewer Accuracy Parameter $\lambda_R$ and Number of Reviews Per Paper $m$. \label{fig:QB_tradeoff_noisy}}
     {The figures show the QB-tradeoff curves under different review qualities and numbers of reviews per paper, resulting from different acceptance thresholds in (a) the \emph{${(\lambda_A = 1, \lambda_R, \NumReviews=3, \TD=0.7)}$-ICLR$^{2020, 6}$ model} and (b) the \emph{$(\lambda_A = 1, \lambda_R = 1, \NumReviews, \TD=0.7)$-ICLR$^{2020, 6}$ model}.}
\end{figure}

% \dkcomment{Given that higher $\lambda$ means better reviews, calling it ``noise'' is very misleading. I am trying to change to ``accuracy parameter'' everywhere, but if you think of a better name, please fix. Also, if I forgot any. Most importantly, fix it in the captions (in the PDFs, which I cannot fix)}

We next examine how review noise affects the QB–tradeoff curves.
\cref{fig:QB_tradeoff_noisy} (a) plots these curves for a fixed author‐signal accuracy $\lambda_A = 1$ varying review-signal accuracy levels $\lambda_B\in \{0.7,0.8,0.9,1\}$.
Consistent with our observations in \cref{sec:QB-trade-noise}, improving review quality yields a strictly dominating tradeoff curve. Furthermore, unlike \cref{fig:Pareto_frontier_continuous} --- where the maximum conference quality remains unchanged across noise levels --- we find that when authors themselves are noisy, better reviews actually raise the attainable maximum conference quality. Intuitively, noisy authors tend not to persist in resubmitting after some negative reviews; under high review‐noise, there are more high–quality papers that receive unlucky low scores, which prompts their authors to take the outside option rather than resubmit, and thereby reducing overall conference quality.

% In \cref{fig:QB_tradeoff_noisy} (b), we present how the QB-tradeoff curves change with the noise of authors' signals. We observe that when we raise the acceptance threshold, starting from a very low $\tau$, the difference between different curves is very small at the beginning. This is because our estimated review signal distribution is very likely to assign a low review score to papers with negative quality. Therefore, the authors with negative quality papers roughly decide to quit upon receiving a few negative review no matter how noisy their review signals are. This results in a similar QB-tradeoff curves for various $\lambda_A$. However, when $\tau$ is large, only positive-quality papers are submitted. Authors who are more confident in their papers' quality (when $\lambda_A$ is large) are more persistent in resubmitting their papers even when they observe a few bad reviews. Consequently, there are more rounds of resubmissions, resulting in a larger review burden.

We further present how the number of reviews per paper affects the QB-tradeoff curves in \cref{fig:QB_tradeoff_noisy} (b).
When we fix the review quality at $\lambda_R = 1$ and increase $m$, the QB-tradeoff curve shifts upward and to the right: for $\NumReviews \in \{1,2,3,4\}$, the maximum conference quality is $\CONFUTIL = 0.56, 0.6, 0.61, 0.63$, which is reached at $\PaperReviews = 3.4, 5.8, 8.5, 11$, respectively. This means that a larger $\NumReviews$ greatly increases the review burden but has the benefit of reaching a larger maximum conference quality. However, the marginal benefit of having an additional review is minimal relative to its extra review burden when $\NumReviews \ge 3$. This implies that a small number of high-quality reviews paired with a carefully chosen acceptance threshold can reach a desirable QB-tradeoff. 

\subsection{QA-tradeoffs With Noisy Authors}

We next investigate author welfare. As shown in \cref{fig:QA_tradeoff_noisy}, the QA-tradeoff curves in the real-data estimated categorical model look similar to those in the continuous model, where both the conference quality and author welfare first increase and then decrease as the acceptance threshold increases from $-2$ to $2$; eventually, both the quality and welfare reach 0. Intuitively, a high acceptance threshold rejects too many (both high-quality and low-quality) papers, while a low acceptance threshold accepts too many negative-quality papers and results in a low conference value; both lead to a decrease in the author welfare.
% \fangcomment{we need to update the figure because $\AUTHUTIL$ is always greater or equal to 1}
% \yzcomment{If only the very top papers are accepted, the sum of author utilities is very small. We consider the sum, not the average.}\fangcomment{Does the author get utility by submitting to the side option? Though not always greater than 1, it should not be zero.}\yzcomment{our definition of welfare only counts papers accepted by the prestigious conference.}
\begin{figure}
     \FIGURE
     {\begin{subfigure}[b]{0.46\textwidth}
         \centering
         \includegraphics[width=\textwidth]{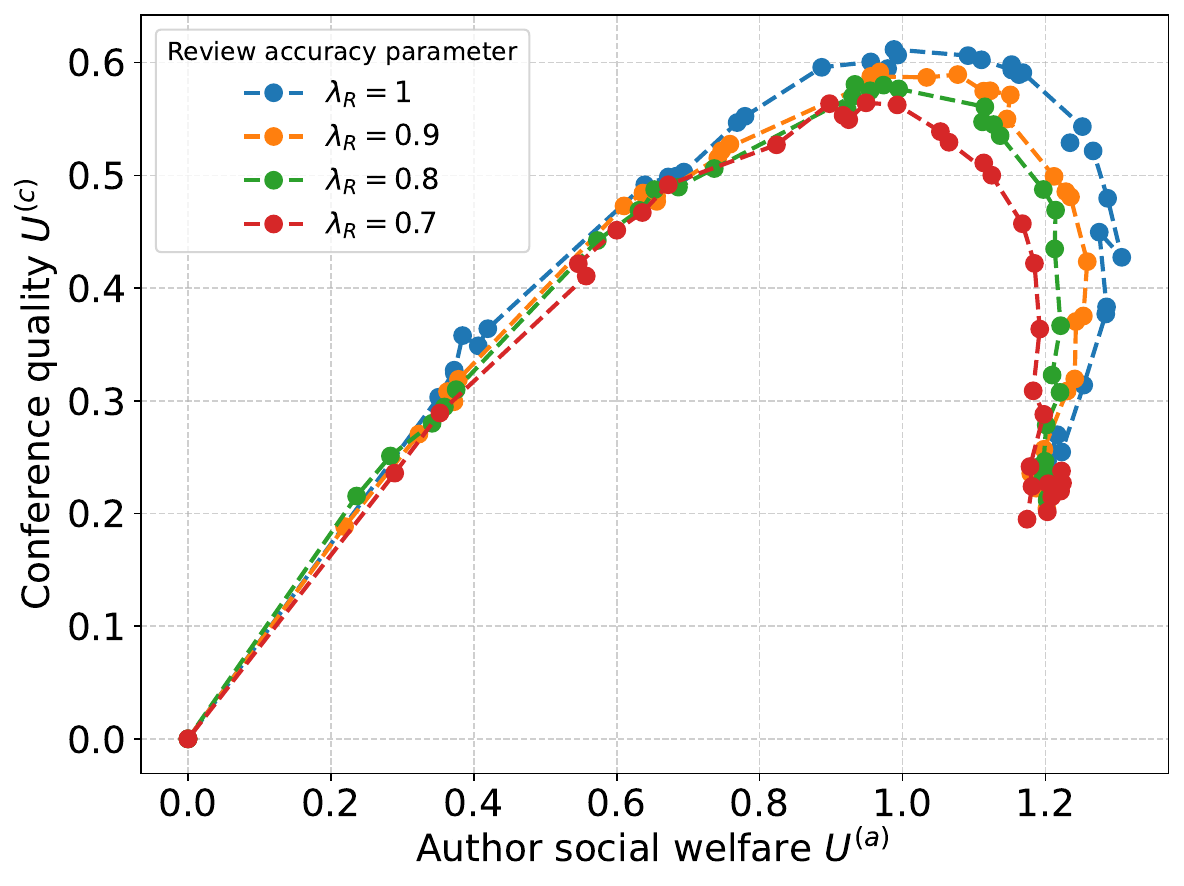}
         \captionsetup{size=}
         \caption{QB-tradeoff varying $\lambda_R$.}\label{fig:QA_tradeoff_noisy_a}
     \end{subfigure}
     \begin{subfigure}[b]{0.46\textwidth}
         \centering
         \includegraphics[width=\textwidth]{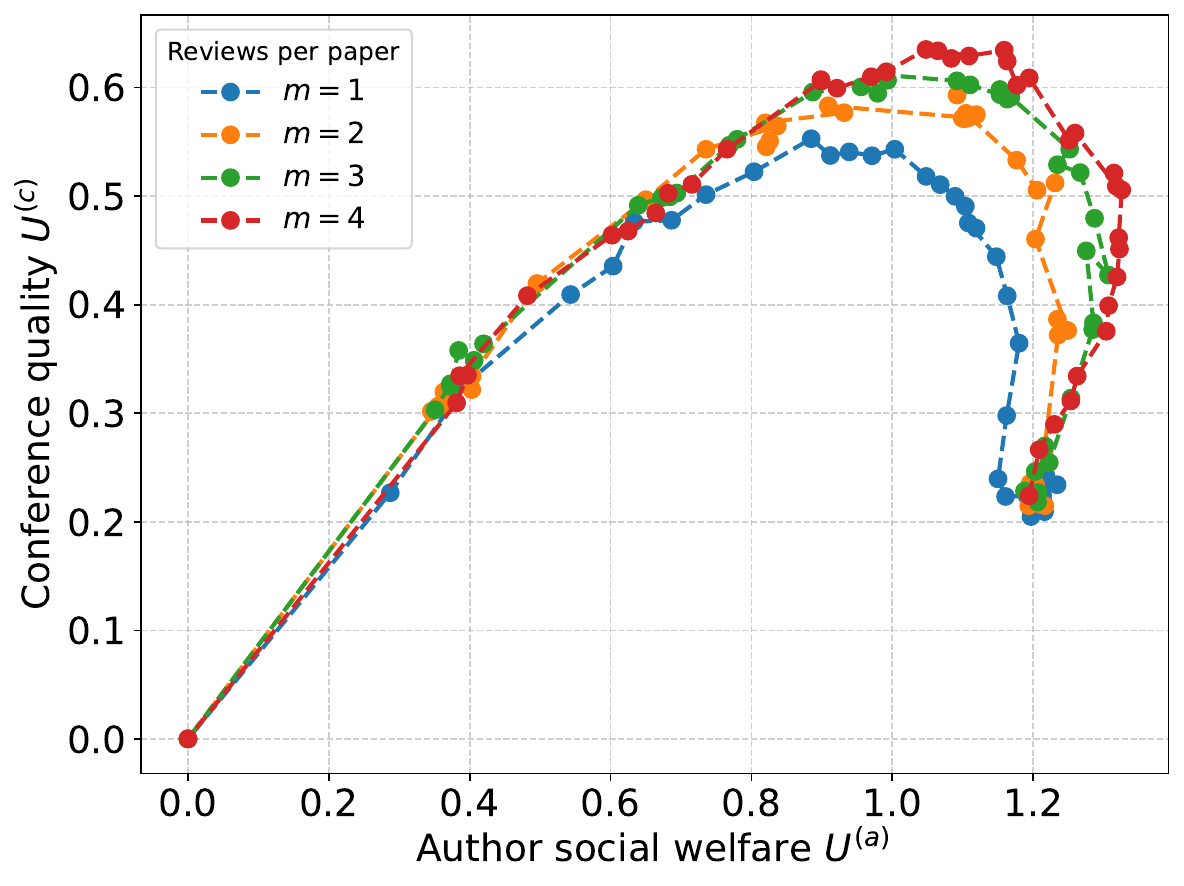}
         \captionsetup{size=}
         \caption{QA-tradeoff varying $m$.}\label{fig:QA_tradeoff_noisy_b}
     \end{subfigure}
     }
     {QA-tradeoff Curves under Reviewer Accuracy Parameter $\lambda_R$ and Number of Reviews Per Paper $m$. \label{fig:QA_tradeoff_noisy}}
     {The figures show the QA-tradeoff curves under different review qualities and numbers of reviews per paper, resulting from different acceptance thresholds in (a) the \emph{${(\lambda_A = 1, \lambda_R, \NumReviews=3, \TD=0.7)}$-ICLR$^{2020, 6}$ model} and (b) the \emph{$(\lambda_A = 1, \lambda_R = 1, \NumReviews, \TD=0.7)$-ICLR$^{2020, 6}$ model}.}
\end{figure}

In \cref{fig:QA_tradeoff_noisy}(a), we confirm our theoretical results in \cref{sec:QB-trade-noise} that a better review quality leads to a dominating QA-tradeoff. With a better review quality, the conference can not only achieve the same quality with a larger author welfare, but can also achieve a larger maximum conference quality.

We next investigate the effect of the number of reviews per paper on the QA-tradeoffs.
Although increasing $\NumReviews$ typically increases the review burden conditioned on the same conference quality, it improves \author welfare. 
This can be observed in \cref{fig:QA_tradeoff_noisy}(b), where the QA-tradeoff curves corresponding to larger $\NumReviews$ dominate those corresponding to smaller $\NumReviews$.
Intuitively, this is because a larger number of i.i.d.~peer reviews work similarly to a decrease in review noise. Therefore, the desired batch of papers can be accepted with fewer rounds of resubmissions, resulting in larger author welfare.

\subsection{Acceptance Rate With Noisy Authors}

In \cref{sec:acc_rate}, we show that whether the acceptance rate is increasing or decreasing in the acceptance threshold can be roughly determined by a quantity resembling the hazard rate of the paper quality prior. The acceptance rate tends to be monotone decreasing in $\tau$ if the hazard rate of the paper quality distribution is monotone. Our experiments on real data echo this observation. The quality prior we estimated from the real ICLR review data has a monotone increasing hazard rate. As expected, in \cref{fig:Acc_rate_noisy}, we observe that the acceptance rates are monotone decreasing in $\tau$.

We also observe that, compared to a weaker review system (i.e., one with a smaller $\lambda_R$ or fewer reviews per paper $\NumReviews$), a stronger review system tends to have a lower acceptance rate when the threshold $\tau$ is small, and a higher acceptance rate when $\tau$ is large. This pattern arises because a small $\tau$ admits many low-quality submissions. A stronger system, being more accurate, is better at filtering out these low-quality papers, leading to fewer acceptances. Conversely, when $\tau$ is large and most submissions are high-quality, a better review system can more efficiently identify good papers\textemdash often with fewer rounds of resubmission\textemdash resulting in a higher acceptance rate.

In practice, top conferences usually face a large number of low-quality submissions and thus have acceptance rates lower than 35\%. In this case, our results suggest that the acceptance rate is typically increasing with review quality and decreases with the acceptance threshold.

\begin{figure}
     \FIGURE
     {\begin{subfigure}[b]{0.47\textwidth}
         \centering
         \includegraphics[width=\textwidth]{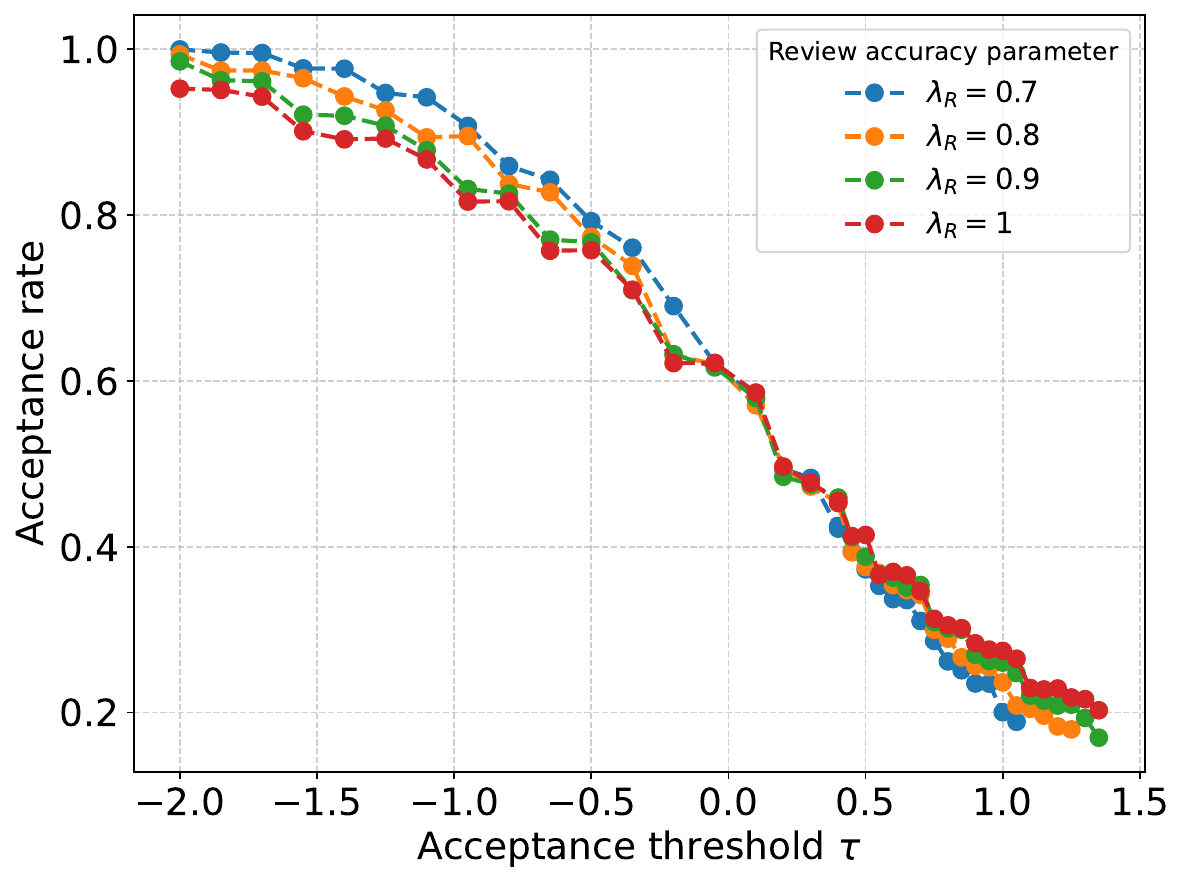}
         \captionsetup{size=}
         \caption{Acceptance rate varying $\lambda_R$.}
     \end{subfigure}
     \begin{subfigure}[b]{0.48\textwidth}
         \centering
         \includegraphics[width=\textwidth]{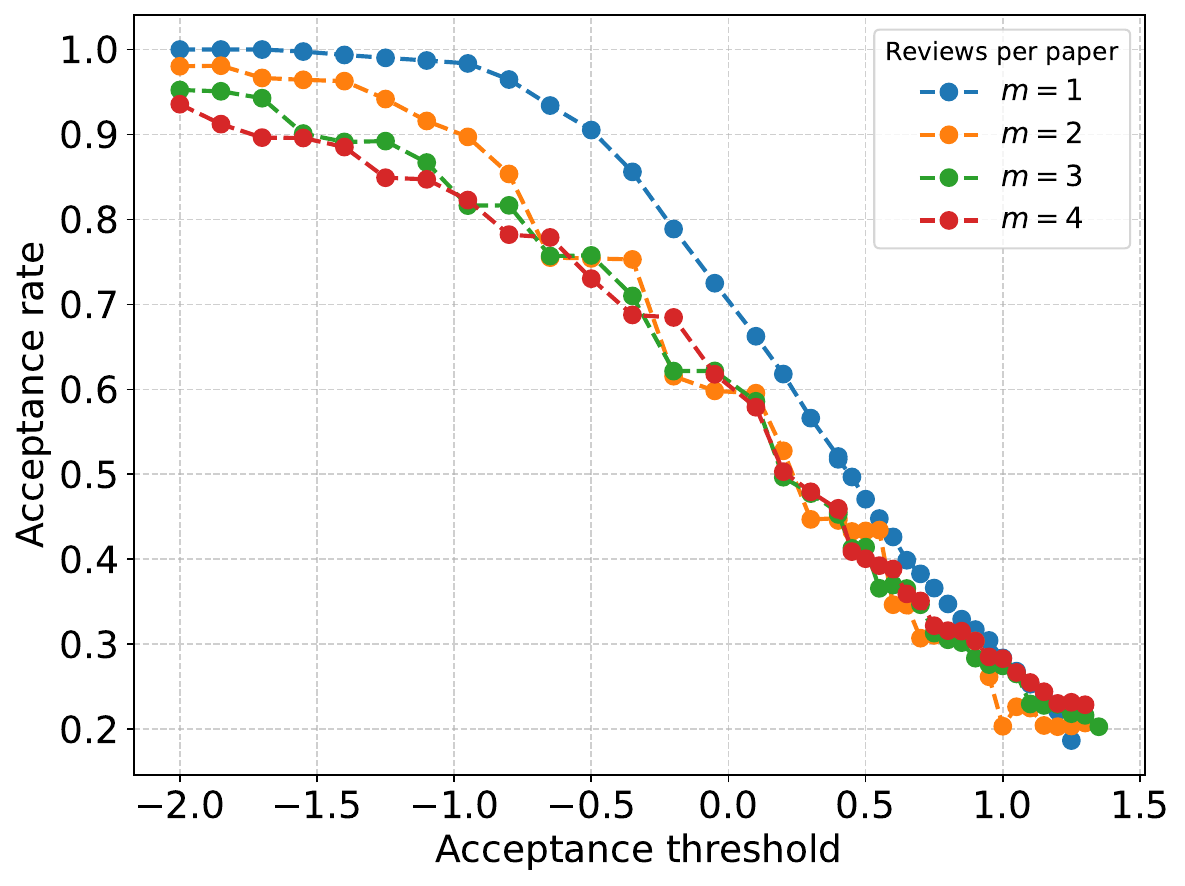}
         \captionsetup{size=}
         \caption{Acceptance rate varying $m$.}
     \end{subfigure}}
     {Acceptance Rates under Reviewer Accuracy Parameter $\lambda_R$ and number of reviews per paper $m$. \label{fig:Acc_rate_noisy}}
     {The figures show the relationship between the acceptance rate and the acceptance threshold under different review qualities and numbers of reviews per paper, resulting from different acceptance thresholds in (a) the \emph{${(\lambda_A = 1, \lambda_R, \NumReviews=3, \TD=0.7)}$-ICLR$^{2020, 6}$ model} and (b) the \emph{$(\lambda_A = 1, \lambda_R = 1, \NumReviews, \TD=0.7)$-ICLR$^{2020, 6}$ model}.}
\end{figure}

\begin{remark}[Discussions and limitations of ABM experiments and results]
First, we note that in our noiseless model, the quality distribution of accepted papers equals the prior paper quality distribution conditional on being greater than the de facto threshold. However, our agent-based model presented in this section was instead learned from the \emph{submitted} papers, so it may overcount borderline papers while under-counting low-quality papers. 
% \dkcomment{Given that we don't discuss anything about ICLR data in the main paper (except saying that we do have experiments on it in the appendix), this sentence may lack context for a reader.}\yichicomment{addressed}\gscomment{it has context, but still seems out of place.  Could we move it to the appendix somewhere?  Ideally the main body reads well without the appendix, but this is just something that does not apply.  Okay, I see this is particially contradicted below.  However, because is something that is still a relevant take-away.  this takeway is completely not interesting because it discribes a weakenss of results that we never discuss.  Right?}

Second, we have seen that in our experiments, a small number of reviews per paper can achieve appealing tradeoffs. However, in practice, a larger number of reviews may nonetheless be desirable due to several other considerations:
% \dkcomment{Given that we now analyze author utility explicitly, it might not just be ``real-world'' any more.}\yichicomment{addressed}\gscomment{they are not unmodeled either.  We model them and the results are in the appendix.}
first, more reviews may decrease the average number of times a paper needs to be resubmitted, helping authors; second, our reviewer error model does not capture the situation where there may be a ``fatal flaw'' that only an astute reviewer observes; and finally, more reviewers can provide more feedback.
% Relatedly, in \Cref{fig:tradeoff_noisy_m}, we observe that soliciting only one review per paper may slow the progress of authors learning about the quality of their papers based on historical reviews;
% this results in low conference quality. 
Moreover, our model assumes that reviews are i.i.d., which is not true in reality when the reviewers can communicate (after the rebuttal). The integrated review signals after communication may become much more informative than aggregating each of them as an i.i.d.~review; this may significantly benefit the strategy of soliciting a large number of reviews.
\end{remark}

%% file: OR-resbumission/sections/Additional_appendix.tex
% \dkreplace{\section{Additional Details of ABM Experiments}}{
\subsection{Additional Details of ABM Experiments}

In this section, we provide additional details about our agent-based model experiments under the categorical model. This includes how we simplify the threshold acceptance policy in the categorical model and how we learn the model parameters from real data.

%\subsection{Simplified Threshold Acceptance Policy}
\subsubsection{Simplified Threshold Acceptance Policy}

\label{app:simplified_threshold}

In the categorical model, a threshold $\tau$ may not uniquely determine a policy, namely, when there is a combination of reviewer signals which occurs with positive probability and induces posterior expected quality exactly $\tau$. 
For convenience of notation and visualization, instead of specifying the additional parameter $r$, we use the following convention to associate a unique policy $\ACCMAP[\tau]$ with each $\tau \in \R$.
Let $\min(\QualSet) \leq U(\RevSigV[1]) < \cdots <  U(\RevSigV[M]) \leq \max(\QualSet)$ be the expected posterior qualities of all possible combinations of review signals that occur with positive probability.
Note that the assumption that all inequalities between $U(\RevSigV[1])$ and $U(\RevSigV[M])$ are strict is basically without loss of generality. If there are multiple review vectors with the same expected posterior quality, our proofs can be adjusted by replacing one vector \RevSigV[i] with the set of all vectors giving rise to the same $U(\RevSigV[i])$. This change is merely syntactic.
\begin{enumerate}
\item If $\tau <  U(\RevSigV[1])$, accept everything.
\item If $\tau \geq \max(\QualSet)$, reject everything.
\item If $ U(\RevSigV[M]) \leq \tau < \max(\QualSet)$, accept papers of posterior expected quality $U(\RevSigV[M])$ with probability $\frac{\max(\QualSet) - \tau}{\max(\QualSet) -  U(\RevSigV[M])}$ and reject all other papers.
\item For the intermediate cases $\tau \in [ U(\RevSigV[1]),  U(\RevSigV[M]))$,
let $i$ be such that $\tau \in [ U(\RevSigV[i]),  U(\RevSigV[i+1]))$. Then, we interpret a threshold of $\tau$ as the policy which accepts all papers of expected posterior quality at least $U(\RevSigV[i+1])$, rejects all papers of expected posterior quality strictly less than $U(\RevSigV[i])$, and accepts papers of expected posterior quality exactly $U(\RevSigV[i])$ with probability $\frac{U(\RevSigV[i+1]) - \tau}{ U(\RevSigV[i+1]) -  U(\RevSigV[i])}$.
\end{enumerate}

%\subsection{Learning Parameters of the Categorical Model from Data}
\subsubsection{Learning Parameters of the Categorical Model from Data}
\label{sec:learning-parameters}

We set the parameters of our model based on the OpenReview datasets of submissions and reviews for ICLR 2020 \citep{iclr2020review} and ICLR 2021 \citep{iclr2021review}. The datasets contain about $1500$ and $2500$ submissions, respectively; typically, each submission is reviewed thrice, with scores in $\SigSet = \SET{0,1,\dots,9}$.
We apply the same learning algorithm to each of the two datasets separately, yielding two plausible parameter settings for evaluation.

\emph{The Number of Paper Quality Scores and Signals.}
While the set $\SigSet = \SET{0,1,\dots,9}$ of available review scores is known, the number (or set) of different paper qualities is not. Thus, our goal is to simultaneously learn the number of paper qualities, the paper quality distribution, and the distribution of reviewer signals conditioned on the paper's quality.

To do so, we exhaustively try all numbers $\QualSetSize$ of paper quality scores in $\SET{2, \ldots, 10}$ for each dataset; for each, we apply a variant of the EM algorithm described below. Once the EM algorithm has converged, we evaluate the likelihood of the learned model for the held-out test data, and retain the model(s) with the highest likelihood scores.

Given a choice of $\QualSetSize$, to learn the paper quality distribution $\QualDist$ and conditional review distribution $\RevSigDist$, we apply the EM algorithm (adding some noise in each iteration for smoothness) and cross-validation to avoid overfitting [\cite{dawid1979maximum}].
% \fang{original version commented out below}
% Specifically, suppose the size of quality space is $|\QualSet|=\QualSetSize$. After randomly dividing the dataset into a training set and a test set with the $80/20$ rule, we evaluate the likelihood of the test data using the model learned with the training set, for which the process is repeated five times and averaged. The regularization is applied in the way that after each iteration, the new confusion matrix is linearly combined with a matrix of size $\QualSetSize\times 10$ and each row is a uniform distribution with a step size of $0.001$. 
% Finally, we choose the $\QualSetSize$ that maximizes the likelihood after $100$ iterations and use the corresponding model to set our DS parameters.
% The results illustrate that $\QualSetSize=4$ or $5$ tend to fit the data well in both of the datasets. For robustness, we conduct our experiments for both of the models. The resulting parameters are shown in \ref{app:DS_parameters} \yichi{put in appending}. 

Specifically, by cross-validation, we first randomly divide the dataset into five subsets of approximately equal size. We choose one of the subsets as the test set while the remaining $80\%$ of data form the training set. This step is repeated five times: each time, a different one of the five subsets is used as the test set.
Given the training and test dataset, for each $\QualSetSize \in \SET{2,3,\ldots,10}$, we run the EM algorithm for $100$ iterations on the training set to estimate the quality of each paper and the confusion matrix for reviewers (i.e., the matrix of review score probabilities conditional on ground truth quality); the EM algorithm alternates between updating the quality distribution with fixed confusion matrix, and updating the confusion matrix with fixed quality distribution.
To avoid overfitting, we repeat the following steps in every iteration: given an estimated confusion matrix $\RevSigDist^{(k)}$ after the $k$-th iteration, we perturb $\RevSigDist^{(k)}$ with a small amount of noise so that each row $i$ becomes the convex combination $0.99 \cdot \RevSigDist[i]^{(k)}+0.01 \cdot \frac{1}{|\SigSet|}\mathbf{1}$;
here, $\frac{1}{|\SigSet|}\mathbf{1}\in \R^{|\SigSet|}$ is the uniform distribution on signals. After each iteration, we evaluate the likelihood of the test data given the trained model, i.e., $\QualDist$ and $\RevSigDist$. For every $\QualSetSize$, this gives us a sequence of models for each iteration. Finally, the model that corresponds to the greatest likelihood on the test data is selected for use.

To choose the value of $\QualSetSize$, we judge the learned model based on the likelihood averaged over five times of cross-validation. 
The paper quality space size $\QualSetSize$ that has the maximum averaged likelihood is selected. Finally, we output the paper quality distribution $\QualDist$ as well as the confusion matrix $\RevSigDist$ as the average of the learned parameters for each of the five runs with different test sets. 
We find that $\QualSetSize \in \SET{4,5,6}$ tends to fit the data well. In our experiments, we choose $\QualSetSize = 6$ to enable a richer set of submissions and thus smoother curves.

The resulting parameters are shown in \cref{tab:learned_para}.
For experiments in which the authors receive \emph{noisy} signals (instead of the ground truth quality), we set the confusion matrix for the authors $\AuthSigDist$ to be the same as the one for reviewers, $\RevSigDist$; this is because unfortunately, no data are available that show how authors evaluate their own papers.

We note that although only the results for ICLR 2020 and $\QualSetSize = 6$ are presented in the paper, all of our qualitative results hold for all of the learned models.

% In our work, we tested four models: $\QualSetSize=4, 5$ for each of the two ICLR datasets (see \cref{tab:learned_para}). In the paper, only the results for the model with $\QualSetSize=4$ and learned from the ICLR 2020 review data are presented. 

In some of our experiments, we want to explicitly evaluate the impact of increasing the noise in reviews. To do so, we consider reviewer signal matrices which are a convex combination of the learned signal matrix $\RevSigDist$ with the uniform signal distribution $\frac{1}{|\SigSet|}\mathbf{1}$; this corresponds to a reviewer who assigns a uniformly random score with probability $1-\lambda_R$. The weight $\lambda_R \in [0,1]$ placed on the learned distribution $\RevSigDist$ then captures the quality of the signal. Similarly, $\lambda_A$ controls the weight of the confusion matrix of the authors' signal.

\emph{The Numerical Values of Paper Quality.} 
Next, we infer the values of the paper qualities, given the estimated parameters $\QualDistTilde, \RevSigDistTilde$ and the review scores $\boldsymbol{s}$. That is, given $\QualSetSize$, we want to learn a vector of values 
$(\qual_1,\ldots,\qual_\QualSetSize)$ of paper qualities. 

First, for each paper $i$ in the dataset, we take the average over the review scores, denoted by $\bar{s}_i$. Then, we set the value of the quality of paper $i$ as $\psi(\bar{s}_i)$, where $\psi: [0,9] \to \mathbb{R}$ is an increasing function that maps the average score to the quality of a paper. In our experiments, we set $\psi$ as a reversed (and shifted and scaled) sigmoid function, i.e., $\psi(x) = 3\cdot \log{\frac{x+0.01}{9.01-x}}$ for $x\in [0,9]$. We choose the reversed sigmoid function because it can assign ``convex'' weights on both very positive reviews and very negative reviews. Furthermore, it is parameterized by a small number of interpretable parameters so that it is not too complex to empirically set the parameters.

Next, given the ``artificial'' quality assigned to each paper in the dataset, we want to compute the average paper quality of each category. This requires us to learn a distribution of the category label for each paper. With the learned model $\QualDistTilde, \RevSigDistTilde$ and paper $i$'s review signals $\boldsymbol{s}_i$, we can infer this distribution based on Bayes' rule. Let $l_{i,k}=\ProbC{\text{paper } i\text{ belongs to the $k$-th category}}{\boldsymbol{s}_i, \QualDistTilde, \RevSigDistTilde}$ for $k\in \{1, 2, \ldots, \QualSetSize\}$.
We can then set the paper quality as a weighted average
\begin{equation*} %\label{eq:paper_quality_average}
    \qual_k = \frac{\sum_i^n l_{i,k}\cdot \psi\left(\bar{s}_i\right)}{\sum_i^n l_{i,k}}.
\end{equation*}

\subsubsection{The Learned Parameters}

\cref{tab:learned_para} summarizes our parameters for ICLR 2020 and $L=6$. The rows of the confusion matrix are ranked based on the expected scores (from low to high). That is, the $k$th row of the confusion matrix $\RevSigDist$ has a lower average score than the ($k+1$)-st row. Given this ranking, we observe that the value of paper qualities learned from our method is monotone increasing in $k$.

\begin{table}[htb]
\centering
\scriptsize
\caption{Learned Prior ($\QualDist$), Confusion Matrix ($\RevSigDist$), and Quality Levels ($\QualSet$) from the ICLR Dataset.}
\label{tab:learned_para}
\begin{tabular}{|c|>{\centering\arraybackslash}p{1.8cm}|>{\centering\arraybackslash}p{10.2cm}|>{\centering\arraybackslash}p{1.8cm}|}
\hline
 & $\QualDist$ & $\RevSigDist$ (Confusion Matrix) & $\QualSet$ \\
\hline
ICLR (Updated) $\QualSetSize = 6$ 
& $\begin{bmatrix}
0.111\\ 0.1736\\ 0.1958\\ 0.2081\\ 0.1856\\ 0.1259
\end{bmatrix}$ 
& $\begin{bmatrix}
0.0059 & 0.0867 & 0.2653 & 0.3958 & 0.1817 & 0.0541 & 0.0087 & 0.0006 & 0.0009 & 0.0003\\
0.0004 & 0.0109 & 0.1204 & 0.3564 & 0.2940 & 0.1441 & 0.0557 & 0.0152 & 0.0027 & 0.0003\\
0.0006 & 0.0074 & 0.0891 & 0.2546 & 0.3131 & 0.2206 & 0.0901 & 0.0168 & 0.0073 & 0.0004\\
0.0016 & 0.0084 & 0.0691 & 0.1935 & 0.2803 & 0.2557 & 0.1451 & 0.0356 & 0.0092 & 0.0015\\
0.0007 & 0.0069 & 0.0387 & 0.1485 & 0.2612 & 0.2811 & 0.1949 & 0.0614 & 0.0060 & 0.0007\\
0.0009 & 0.0007 & 0.0176 & 0.0930 & 0.1674 & 0.2593 & 0.3125 & 0.1038 & 0.0413 & 0.0036
\end{bmatrix}$ 
& $\begin{bmatrix}
-2.9101\\ -0.8668\\ -0.0772\\ 0.6178\\ 1.2811\\ 2.7132
\end{bmatrix}$ 
\\
\hline
\end{tabular}
\end{table}

%% file: OR-resbumission/sections/Memory.tex
Most conferences' acceptance policies are memoryless, in the sense that resubmissions are treated the same as new submissions.
However, in part to deal with the large number of papers that are repeatedly resubmitted, several conferences have experimented with models that have ``institutional memory.''  We consider the following types of policies which are not memoryless and contain some institutional memory.

\begin{description}
\item[Time Limited, Fixed Threshold:] The simplest way to incorporate memory into the submission process is to limit the number of times the same paper can be submitted.
We call such a policy a \emph{$T$-round fixed-threshold policy.} 
% \fangcomment{It seems we already use $T$ for the number of runs  }\yzcomment{changed to $N$}

\item[Time Limited, Variable Threshold:] A generalization of $T$-round fixed threshold policies is to allow different acceptance thresholds for different rounds. This allows a conference to set higher/lower standards for resubmissions.
However, we require the conference to solicit the same number of reviews for each round. 
We do so for two reasons. First, this reduces the policy space\textemdash this is significant in terms of computation when optimizing over policies with memory. Second, it excludes highly unrealistic policies with very specific dependency on model parameters. For example, when authors are noiseless, the following unrealistic policy can achieve maximum conference quality with minimum review burden: the conference rejects all submissions $T-1$ times without review. In round $T$, one review is solicited, and the paper is accepted if and only if the expected quality conditioned on the review is larger than $\tau$; finally, in round $T+1$, the submission is accepted without review. A careful choice of $T$ and $\tau$, taking advantage of authors' patience (or lack thereof) and knowledge of their own paper's quality, ensures that no negative-quality papers are submitted, yet all positive-quality papers are submitted and eventually accepted.
% \yichicomment{This footnote is quite long. Is there a better solution? Make it an endnote?} \dkcomment{I think it actually works pretty well as the bulk of this paragraph.}

Formally, a \emph{round-dependent threshold policy} is defined by a threshold vector $\boldsymbol{\tau}=\left(\tau^{(1)},\tau^{(2)}, \ldots, \tau^{(T)}\right)$; in round $t \leq T$,  a paper with reviews $\RevSigV$ is accepted if and only if its expected quality conditioned on the most recent review vector $\RevSigV$ (not including reviews from earlier rounds) is at least $\tau^{(t)}$.

\item[Review Following:]
Under a \emph{$T$-round review-following threshold memory policy}, not only does the conference track the number of resubmissions; it also considers all past reviews as equal (additional) reviews of resubmissions. That is, reviews are treated identically regardless of which round they were provided in.
% \footnote{As such, they do not serve the purpose of verifying whether authors addressed concerns about previous versions of their paper.}
(As such, they do not serve the purpose of verifying whether authors addressed concerns about previous versions of their paper.)
Again, we have the conference obtain the same number of reviews in each round of resubmissions, i.e., $\NumReviews[t]=\NumReviews$ for all $t$.
The conference commits to a number $T$ of rounds and rejects any paper that has been submitted $T$ times. The conference also commits to a sequence of thresholds $\boldsymbol{\tau}=\left( \tau^{(1)}, \tau^{(2)}, \ldots, \tau^{(T)}\right)$, such that in round $t$, a paper with historical reviews $\left(\RevSigV[1], \ldots, \RevSigV[t]\right)$ is accepted if and only if its expected quality conditioned on $\left(\RevSigV[1], \ldots, \RevSigV[t]\right)$, $U(\RevSigV[1], \ldots, \RevSigV[t])$, is at least $\tau^{(t)}$.
\end{description}

All three policy types are time-limited, in that the number of times any particular paper can be submitted is capped.  This is similar to  certain National Science Foundation programs (e.g., CAREER), where the number of times a proposal (or sometimes author) can submit is limited. 
Time-limited policies with fixed threshold treat all submissions, whether initial or resubmitted, equally in each round, until they have reached their resubmission limit. In this section, we analytically investigate such review policies with a focus on how the limit on the number of resubmissions affects the QB-tradeoff. 
In contrast, the other two types of policies allow different thresholds in different rounds. In particular, the review-following model captures the increasingly popular policy of requiring resubmissions to be accompanied by previous reviews (e.g., at IJCAI).
These two generalizations are more complicated to analyze, and we use ABMs to investigate their QB-tradeoff.

Perhaps subtly, review-following policies do not fully subsume the other two policies, because past reviews cannot be treated differently from new reviews. For example, review-following policies cannot simulate a time-limited fixed threshold policy in which every round, a paper obtains two reviews and is accepted iff both reviews are positive. The reason is that a review-following policy cannot distinguish the case of having one positive review in each round (which should lead to rejection) from the case of having zero positive reviews in the first round and two positive reviews in the second round (which should lead to acceptance).

\subsection{Time-limited Policy With Fixed Acceptance Rule}
\label{sec:time_limited_policy}

We start by investigating policies under which in each round, the conference applies a fixed monotone acceptance rule. 
% We first show that when a conference merely restricts the number of times a paper can be submitted (to some value $T$), noiseless authors will respond exactly as if the conference allowed unlimited resubmissions. Next, based on this property, we show how $T$ affects the QB-tradeoff.
As for the case of unlimited resubmissions, we first establish theoretical results for noiseless authors, and then test their robustness to noise in the authors' signals with ABM experiments.

\subsubsection{Theoretical Results for Noiseless Authors}

We begin with a theoretical analysis for noiseless authors.
We first show that under this type of acceptance policy, a threshold strategy is still a best response for noiseless authors, i.e., there exists a de facto threshold $\theta$ for every monotone acceptance policy $\ACCMAP$. 

We begin by generalizing \cref{lem:author_response} to time-limited policies.

\begin{lemma}\label{lem:best-response-finite}
Consider a conference which allows a paper to be submitted at most $T$ times, and for each of these submissions independently decides whether to accept the paper, according to the same monotone policy $\ACCMAP$.
Suppose that the conference value $\ConfValue$ is fixed, and accordingly $\rho$ is fixed.
Then, the author's best response is to (re)submit the paper (in each round) if $\AccP{\ACCMAP}{\Qual} > 1/\rho$, and take the side option when $\AccP{\ACCMAP}{\Qual} < 1/\rho$.
The author is indifferent between submitting and not submitting if $\AccP{\ACCMAP}{\Qual} = 1/\rho$.
\end{lemma}

\proof{Proof of \cref{lem:best-response-finite}}
  The author will submit a paper with quality $\Qual = q$ if her expected utility is greater than 1, and not submit it if her expected utility is less than 1.
  We compute the expected utility for an author who submits the paper exactly $T$ times, akin to the proof of \cref{lem:author_response}:

\begin{align*}
    u^{(a)}(q,\ACCMAP, \ConfValue) & = 
    (\TD \cdot (1-\AccP{\ACCMAP}{q}))^T 
    + \ConfValue \cdot \sum_{t=1}^T \AccP{\ACCMAP}{q} \cdot  
    (\TD \cdot (1-\AccP{\ACCMAP}{q}))^{t-1}
    \\ & = 
    (\TD \cdot (1-\AccP{\ACCMAP}{q}))^T
    + \ConfValue \cdot \AccP{\ACCMAP}{q} \cdot 
    \frac{1 - (\TD \cdot (1-\AccP{\ACCMAP}{q}))^T}{1 - \TD \cdot (1-\AccP{\ACCMAP}{q})}
    \\ & = 
    \left(\frac{\ConfValue \cdot \AccP{\ACCMAP}{q}}{1 - \TD \cdot (1-\AccP{\ACCMAP}{q})} - 1 \right)
    \cdot \left( 1 - (\TD \cdot (1-\AccP{\ACCMAP}{q}))^T \right) + 1.
\end{align*}
To determine when this utility is strictly larger resp.~strictly smaller than 1, we need to determine when the product of the first two terms is positive resp.~negative. The second factor is always positive, and the first has the same sign as $\AccP{\ACCMAP}{q} - 1/\rho$, regardless of $T$.
This completes the proof. \Halmos
\endproof

\cref{lem:best-response-finite} shows that when the conference value is fixed, an author’s decision depends solely on the relationship between the acceptance probability $\AccP{\ACCMAP}{q}$ and the inverse of the conference attractiveness factor $1/\rho$, and is independent of the time limit $T$.
However, the attractiveness factor $\rho$ is itself determined by the authors' submission strategy as well as the conference acceptance policy and the time horizon $T$.

Consider a threshold strategy under which authors with quality $\Qual > \bar{q}$ submit and continue to resubmit up to $T$ times, those with $\Qual < \bar{q}$ immediately opt for the outside option, and those with $\Qual = \bar{q}$ submit and resubmit with probability $r$. We will show that, for a fixed conference policy $\ACCMAP$ and time horizon $T$, the attractiveness factor $\rho$ increases with the threshold $\bar{q}$. Consequently, $1/\rho$ decreases with $\bar{q}$, while the marginal acceptance probability $\AccP{\ACCMAP}{\bar{q}}$ increases with $\bar{q}$.
Building on this monotonicity, the results of \cref{prop:de_facto} will then extend naturally to the finite-$T$ setting using the same proof structure.

In order to prove monotonicity, we first introduce the following notation. 
If an author persistently submits a paper of quality $\qual$, the probability that it is eventually accepted under the acceptance policy $\ACCMAP$ is given by:
\begin{align*}
A(\qual, T, \ACCMAP) 
& = \sum_{t=0}^{T-1} (1-\AccP{\ACCMAP}{\qual})^t \cdot \AccP{\ACCMAP}{\qual}
= 1-(1-\AccP{\ACCMAP}{\qual})^T.
\end{align*}

As a result, $\lim_{T \to \infty} A(\qual, T, \ACCMAP) = 1$, corresponding to the fact that when an author resubmits a paper until acceptance, it will eventually be accepted.

The conference value under the threshold strategy \ACCMAP can be characterized as:
\begin{align*}
    \ConfValue(\bar{q}, r, T, \ACCMAP) 
    & = 1 + \frac{r\cdot \bar{q} \cdot \QualProb{\bar{q}} \cdot A(\bar{q}, T, \ACCMAP) 
    + \sum_{q \in \QualSet, q > \bar{q}} q \cdot \QualProb{q} \cdot A(q, T, \ACCMAP)}{r\cdot \QualProb{\bar{q}} \cdot A(\bar{q}, T, \ACCMAP) 
    + \sum_{q \in \QualSet, q > \bar{q}} \QualProb{q} \cdot A(q, T, \ACCMAP)} \tag{categorical model}\\
    \ConfValue(\bar{q}, T, \ACCMAP) 
    & = 1 + \frac{\int_{\bar{q}}^\infty q \cdot \QualProb{q} \cdot A(q, T, \ACCMAP) dq}{\int_{\bar{q}}^\infty \QualProb{q} \cdot A(q, T, \ACCMAP) dq}. \tag{continuous model}
\end{align*}

Fixing the acceptance policy $\ACCMAP$ and the time limit $T$, we see that the conference value is increasing in the submission threshold $\bar{q}$ and decreasing in the marginal submission probability $r$. Consequently, the conference attractiveness factor $\rho(\bar{q}, r, T, \ACCMAP)$ is also increasing in $\bar{q}$ and decreasing in $r$.

This implies that \cref{prop:de_facto} continues to hold for time-limited acceptance policies with a fixed acceptance rule.
The only difference is that in the characterization result, the conference's attractiveness $\rho$ now depends not only on the acceptance policy $\ACCMAP$, but also on $T$.
% \dkcomment{I found the preceding sentence very confusing. Either it holds, or it doesn't? I hope that I rephrased it correctly, i.e., that you only wanted to express the dependence on more parameters, not that the proof doesn't actually work after all.}
Furthermore, for threshold acceptance policies with threshold $\tau$ under the continuous model, the result from \cref{prop:gap-invariant} can also be directly generalized. 
That is, given a candidate threshold $\theta$, the acceptance threshold that induces it as a de facto threshold is the solution to the following equation:\fangcomment{need to fix $\tau_s$?}
\begin{align*}
  \tau 
  & = \theta + \left(\REVNOISEDIST\right)^{-1}{\left(1-\frac{1}{\rho(\theta, T, \ACCMAP[\tau])}\right)} 
  = \theta + \left(\REVNOISEDIST\right)^{-1}{\left(\frac{\ConfValue(\theta, T, \ACCMAP[\tau]) - 1}{\ConfValue(\theta, T, \ACCMAP[\tau]) - \TD}\right)},
\end{align*}
where $\ACCMAP[\tau]$ is the threshold acceptance policy with threshold $\tau$. This equation allows us to solve for either $\theta$ or $\tau$ given the other.

% \begin{proposition}\label{prop:best-response-finite}
% Consider a conference which allows a paper to be submitted at most $T$ times, and for each of these submissions independently decides whether to accept the paper, according to the same monotone policy \ACCMAP.

% Then, the author's best response is to submit the paper (in each round) if $\AccP{\ACCMAP}{\Qual} > 1/\rho$, and take the side option when $\AccP{\ACCMAP}{\Qual} < 1/\rho$.
% \end{proposition}

% Notice that while the author's expected \emph{utility} depends on $T$, her best \emph{response} does not. Hence, as with an unlimited number of submissions, an author will either submit as often as she is allowed to or not even submit once, and the author's threshold for doing so is the same as with an unlimited number of resubmissions. 

% The property we proved in \cref{prop:best-response-finite} allows us to characterize the expected value of the conference quality and the review burden.
We further characterize the expected value of the conference quality and the review burden.
For a given acceptance threshold $\tau$, let $\mathcal{S}_{\tau} \subseteq \QualSet$ be the set of paper qualities which an author will submit at this threshold.
When $\QualSet$ is discrete,
% \footnote{For continuous qualities, the corresponding quantities are obtained by replacing sums by integrals and probabilities by densities.}
the expected conference quality and review burden are as follows:
\begin{align*}
    \CONFUTIL(\tau) 
    & = \sum_{q\in \mathcal{S}_{\tau}} \QualProb{q} \cdot q \cdot 
    {\sum_{t=0}^{T-1} (1-\AccP{\tau}{q})^{t}\AccP{\tau}{q}}
%\\ & = 
\; = \; \sum_{q \in \mathcal{S}_{\tau}} 
    \QualProb{q} \cdot q \cdot A(q,T, \ACCMAP[\tau]).
    \end{align*}
\dkreplace{Note that $(1-\AccP{\tau}{q})^T$ is the probability of a paper with quality $q$ being rejected for all $T$ rounds; t}{T}hus, the expected conference quality is the expected value of papers in $\mathcal{S}_\tau$, weighted by the acceptance probabilities. For continuous qualities, the corresponding quantities are obtained by replacing sums by integrals and probabilities by densities.

To compute the review burden, recall that the expectation of a non-negative random variable $Z$ is $\Expect{Z} = \int_{0}^\infty \Prob{Z > z}\,dz$. For a paper of quality $q$, the number of submissions exceeds $z \geq 0$ with probability $(1-\AccP{\tau}{q})^{\lfloor z \rfloor}$ for $z < T$, and 0 otherwise. Therefore, the expected review burden is
\begin{align*}
    \PaperReviews(\tau) & = \NumReviews \cdot \sum_{q\in \mathcal{S}_{\tau}} \QualProb{q} \cdot 
    {\sum_{t=0}^{T-1} (1-\AccP{\tau}{q})^{t}}
    \; = \; \NumReviews \cdot \sum_{q\in \mathcal{S}_{\tau}} \QualProb{q} \cdot \frac{1}{\AccP{\tau}{q}} 
    \cdot A(q,T, \ACCMAP[\tau]).
\end{align*}

\subsubsection{ABM experiments for noisy authors}
We next test the robustness of our theoretical results by performing ABM experiments for authors who only receive noisy signals.

\cref{fig:tradeoff_T} (a) shows the QB-tradeoff for time-limited fixed-threshold policies in the continuous model. As expected, by lowering $T$, the conference reduces the maximum conference quality that can be reached by the review policy, but doing so can also reduce the review burden in cases where the desired conference quality is still reachable. 

% \begin{figure}[htb]
%      \centering
%      \begin{subfigure}[b]{0.6\textwidth}
%          \centering
%          \includegraphics[width=\textwidth]{Plots/Finite_round_memoryless_continuous.png}
%          \captionsetup{size=}
%      \end{subfigure}
%      \hfill
%      \caption{The Pareto optimal curves of the QB-tradeoff under $T$-round fixed threshold policies in the \emph{$\boldsymbol{(\sigma=0.5, \mu_{\QualDist} = 0, \sigma_{\QualDist} = 1, \NumReviews = 1, \ConfValue = 3, \TD = .7)}$-Double Gaussian noiseless-author model}. Pareto dominated points are shown as dashed lines, while undominated points are shown as solid lines.  \label{fig:tradeoff_continuous}}
% \end{figure}

\begin{figure}
     \FIGURE
     {\begin{subfigure}[b]{0.47\textwidth}
         \centering
         \includegraphics[width=\textwidth]{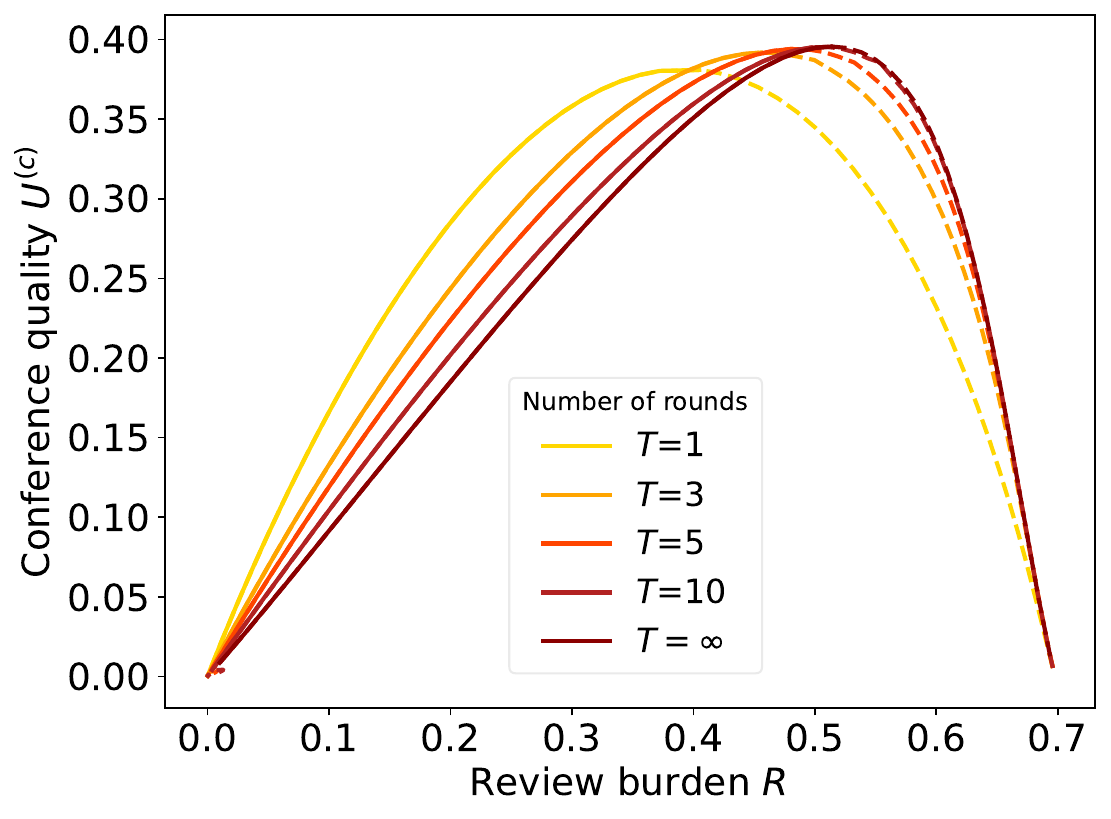}
         \captionsetup{size=}
         \caption{Continuous model with noiseless authors.}
     \end{subfigure}
     \begin{subfigure}[b]{0.46\textwidth}
         \centering
         \includegraphics[width=\textwidth]{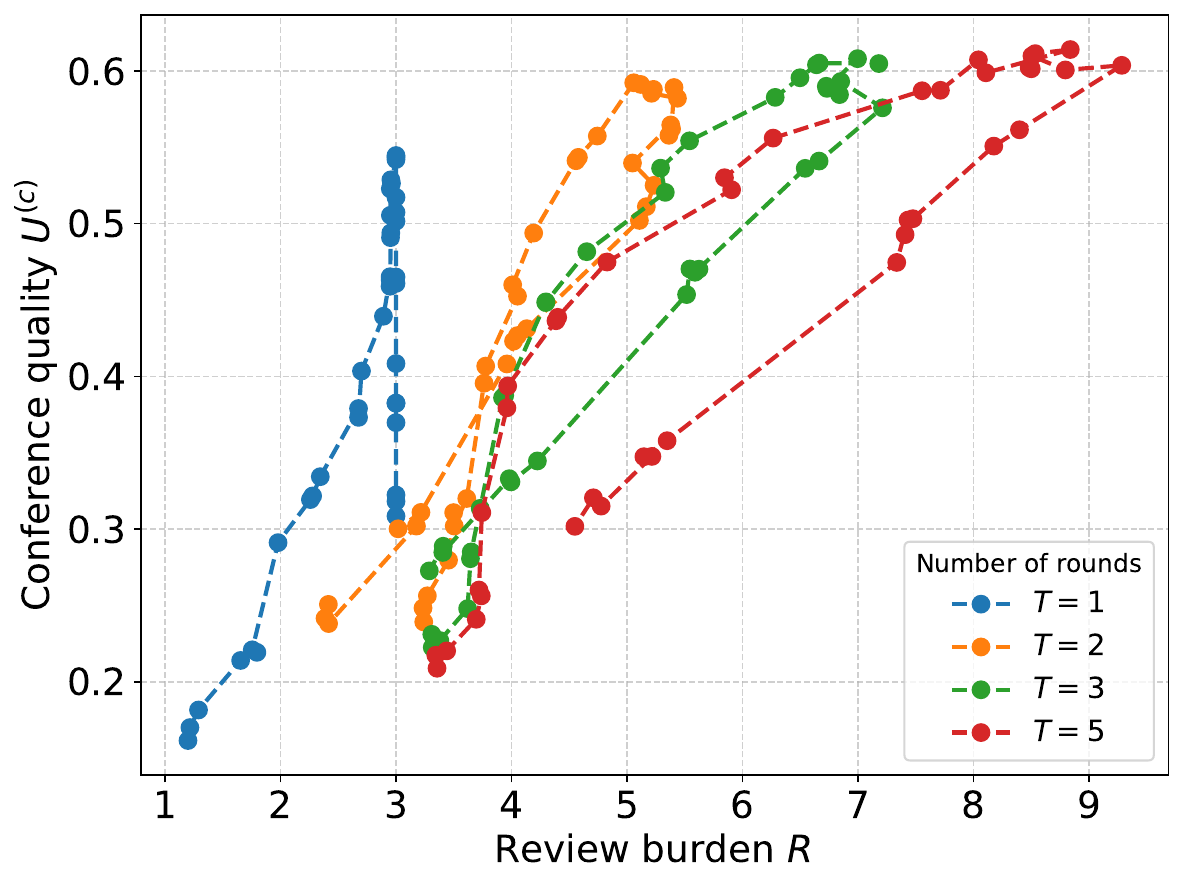}
         \captionsetup{size=}
         \caption{Categorical model with noisy authors.}
     \end{subfigure}}
     {The QB-tradeoff Curves under $\boldsymbol{T}$-round Fixed Threshold Policies in the (a) Continuous Model and (b) Categorical Model. \label{fig:tradeoff_T}}
     {The example is in (a) the \emph{$(\sigma=0.5, \mu_{\QualDist} = -1, \sigma_{\QualDist} = 2, \NumReviews = 1, \TD = 0.7)$-Double Gaussian noiseless-author model} and (b) the \emph{$(\lambda_A=1, \lambda_R=1, \NumReviews = 3,\TD=0.7)$-ICLR$^{2020, 6}$  model. \dkcomment{Left one does not say ``Number of rounds'' before the explanation of colors. Same in the next figure.}}
     % \fangcomment{Can I say if $T\to \infty$, the figure would match Figure EC.1(b) with $m=3$ }\yzcomment{Yes}
     % \dkcomment{Notice that the axis labels are different in both font size and whether the mathematical name of the axis is included (as well as whether we say ``Number of rounds'' in the legend). Make this consistent? Also maybe check other figures to be consistent across all figures in the paper?}
     }
\end{figure}

% \begin{figure}[htb]
%      \FIGURE
%      {\includegraphics[width=0.65\textwidth]{Plots/QB_tradeoffs_continuous_T.pdf}}
%      {The QB-tradeoff Curves under $\boldsymbol{T}$-round Fixed Threshold Policies in the Continuous Model. \label{fig:tradeoff_continuous}}
%      {The example is in the \emph{$(\sigma=0.5, \mu_{\QualDist} = -1, \sigma_{\QualDist} = 2, \NumReviews = 1, \TD = 0.7)$-Double Gaussian noiseless-author model}. Pareto-dominated points are shown as dashed lines, while undominated points are shown as solid lines.}
% \end{figure}

The ABM QB-tradeoff curves for the categorical model and noisy authors are presented in \cref{fig:tradeoff_T}(b). The same pattern can also be observed here: by comparing different colors of dots while fixing a threshold policy (captured by the index of the dot on its corresponding curve), we observe that the conference can reduce the review burden at the expense of quality by lowering $T$.

\begin{figure}
     \FIGURE
     {\begin{subfigure}[b]{0.47\textwidth}
         \centering
         \includegraphics[width=\textwidth]{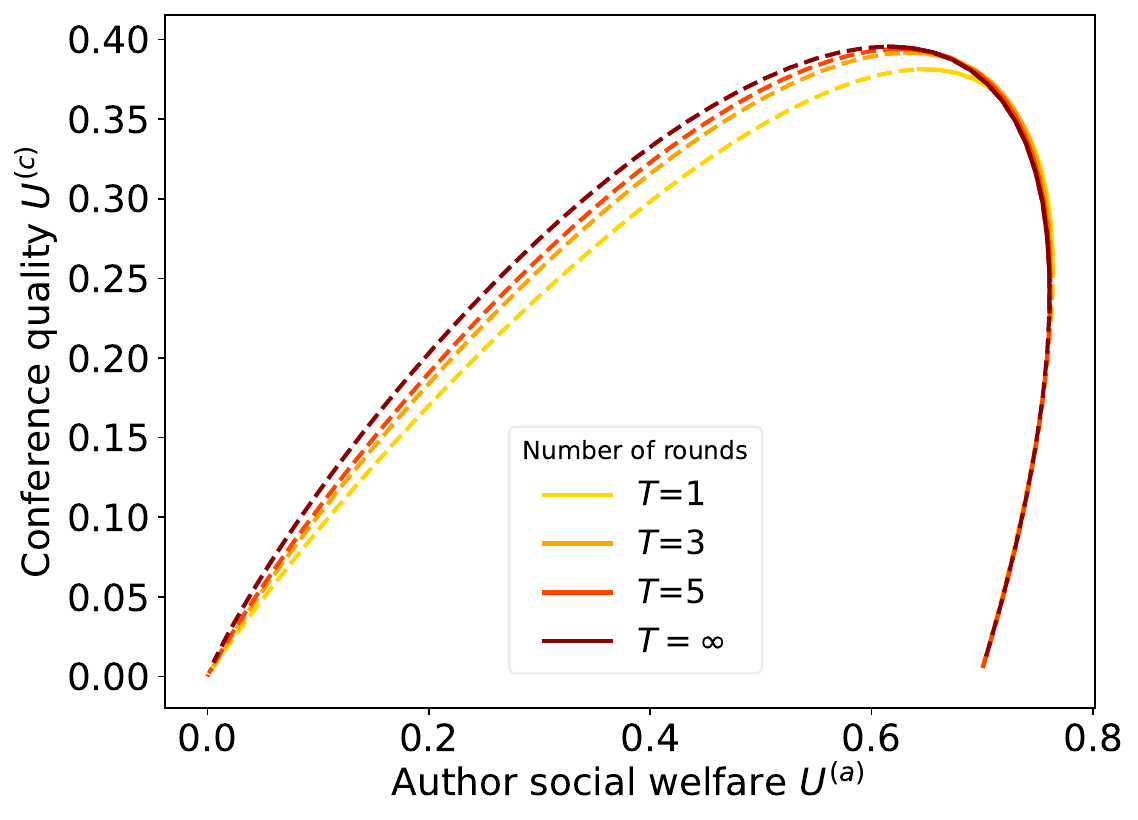}
         \captionsetup{size=}
         \caption{Continuous model with noiseless authors.}
     \end{subfigure}
     \begin{subfigure}[b]{0.46\textwidth}
         \centering
         \includegraphics[width=\textwidth]{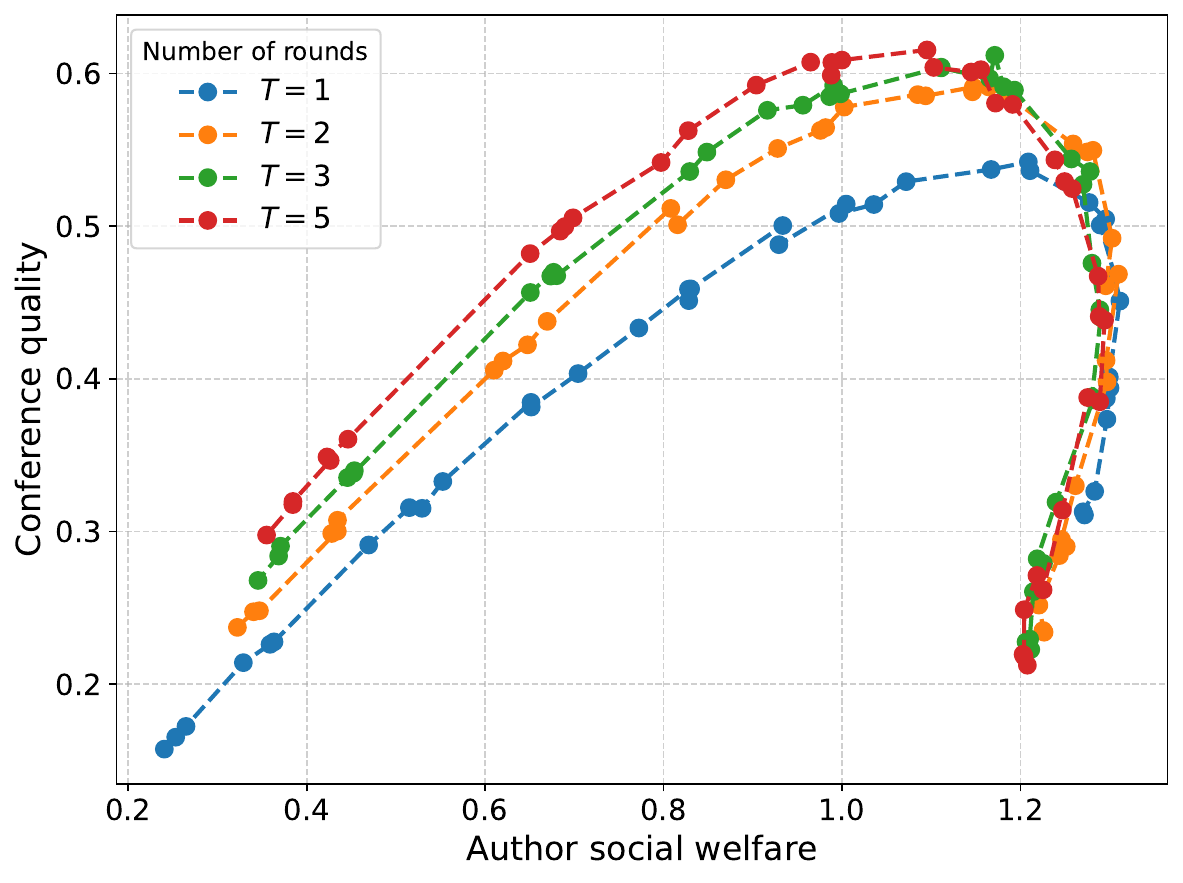}
         \captionsetup{size=}
         \caption{Categorical model with noisy authors.}
     \end{subfigure}}
     {The QA-tradeoff Curves under $\boldsymbol{T}$-round Fixed Threshold Policies in the (a) Continuous Model and (b) Categorical Model. \label{fig:QA_tradeoff_T}}
     {The example is in (a) the \emph{$(\sigma=0.5, \mu_{\QualDist} = -1, \sigma_{\QualDist} = 2, \NumReviews = 1, \TD = 0.7)$-Double Gaussian noiseless-author model} and (b) the \emph{$(\lambda_A=1, \lambda_R=1, \NumReviews = 3,\TD=0.7)$-ICLR$^{2020, 6}$  model}. \dkcomment{(b) does not have $U^{(c)}$ in the $y$-axis label. (a) is missing ``Number of rounds'' in description of colors.}}
\end{figure}

We present the analogous results for the QA-tradeoff curves in \cref{fig:QA_tradeoff_T}. Interestingly, varying $T$ does not reduce author welfare along the Pareto frontiers. While a smaller $T$ limits the number of resubmission opportunities --- potentially lowering the chance of acceptance --- it also enables the conference to adopt a more lenient acceptance threshold to maintain quality. This increases the acceptance probability in each individual round. 
That said, we also observe that smaller values of $T$ result in a lower maximum achievable conference quality.

\subsection{More Fine-Grained Institutional Memory, and Noisy Authors}
\label{sec:memory_noisy}

Next, we investigate the extent to which more fine-grained institutional memory --- different acceptance thresholds in different rounds and reuse of past reviews of a paper --- may further improve the QB-tradeoff for the conference.
We do so in the setting of noisy authors, and therefore --- as before --- use 
ABM experiments. 

\subsubsection{Agent-based Model Setup}

We use a simplified binary model to study review policies with more fine-grained institutional memory: in this model, both the paper quality and the review signal are binary.
This simplification is necessary: when authors know that the acceptance of their paper will depend on all historical reviews, their decisions (even under myopic strategies introduced in \cref{subsec:noisy_setup}), too, will depend on all historical reviews. 
This makes the simulations extremely computationally expensive for a large number of authors, especially for papers that have been rejected many times.
Considering a binary model greatly simplifies the decision-making process, allowing us to compute not only the myopic responses but also the optimal responses from authors.

The \emph{$(\AUTHSIGPROB, \REVSIGPROB, \NumReviews, \TD, T)$-binary model} is a categorical model in which there are two paper qualities $\{-1, +1\}$ and two review signals.  One paper quality is referred to as negative (``bad papers''), and the other as positive (``good papers''). The prior is such that each paper is equally likely to be bad or good. Authors receive the correct signal about their papers with probability $\AUTHSIGPROB$ and otherwise receive the opposite signal. Similarly, each reviewer receives the correct signal about their assigned paper with probability $\REVSIGPROB$ and otherwise receives the opposite signal. The parameters $\NumReviews$, $\TD$, and $T$ are the number of solicited reviews, the discount factor, and the number of times a paper can be submitted, respectively. 

As before, we incorporate the endogenous value of the conference by starting with a suitable $\ConfValue$, and iteratively updating it based on the value of the previous conference.
More specifically, we initialize the conference value $\ConfValue_1 = 2$ in the first round and iteratively update $\ConfValue_t = (1-\lambda_V) \cdot \ConfValue_{t-1} + \lambda_V \cdot \left(1+2 \cdot \ExpectC{Q_i}{\text{paper } i \text{ was accepted in round } t}\right)$, where $\lambda_V = 0.5$. The choice of multiplying by a factor of 2 is so that if all good papers are accepted, the conference value is larger than the initial value $\ConfValue_1$; alternatively, we could define the model as having paper qualities $\QualSet = \SET{-2,2}$ instead of $\SET{-1,1}$.

The \emph{optimal strategy} assumes that the authors are best-responding to the conference's policy. That is, given that the game lasts for $T$ rounds, the author uses backward induction with dynamic programming to optimize the decision. Specifically, she first reasons about the expected utility in round $T$ given all possible histories of reviews; she then similarly updates the reviews-utility mapping in rounds $T-1$, $T-2$, and so on. Eventually, she can infer the optimal action in the current round. 
% We set $T=10$ in our experiments due to running time concerns.

Our approach involves searching over policies with $T$ submissions per paper. Unfortunately, the number of such policies grows exponentially in $T$. We therefore restrict our experiments to $T = 5$, and set $\NumReviews = 3$ for the binary model.
% Computing the optimal strategy for authors in categorical or continuous models requires solving a dynamic program with a much larger state space; the resulting computational requirements prevent us from including such experiments.

We generate candidate policies by searching over different thresholds.
For $T$-round fixed-threshold policies, we select the $40$ candidate thresholds  $\tau \in \{-1, -0.95,\ldots, 0.95\}$; in each run, one of these thresholds is used for all rounds. 
For the other two types of policies, to reduce the number of samples, we only enumerate over the thresholds $\tau^{(t)}$ for the first three rounds while fixing $\tau^{(t)}$ for $t = 4, 5$. 
More precisely, we select $40^3$ candidate threshold vectors $\boldsymbol{\tau}$, as follows:
for each $t \in \{1, 2, 3\}$, we select $40$ thresholds $\tau^{(t)}$ from $\{-1, -0.95,\ldots, 0.95\}$ which gives us $40^3$ vectors of length $3$. 
For each $t=4, 5$, we fix $\tau^{(t)}$ such that the paper is accepted if and only if two of the three newly sampled reviews in round $t$ are positive. 

\subsubsection{Results}\label{sec:memory_noisy_result}

\Cref{fig:memory_policies} shows the Pareto optimal points of the QB-tradeoff for each of the three review policies. 
We summarize the main takeaways as follows:
\begin{itemize}
    \item Compared with time-limited fixed-threshold policies, round-dependent threshold policies (the variable-threshold policy and the review-following policy) tend to have dominating QB-tradeoffs. Our results suggest that having historical reviews follow submissions can help improve the QB-tradeoff. However, all review policies can achieve similar maximum conference quality (the review-following policy has a slightly higher maximum quality, while the improvement is rather marginal). The main advantage of the review-following policy seems to be a reduction in the review burden instead of improving the conference quality.
    \item We further observe that for both variable-threshold policies and review-following policies, with round-dependent thresholds, the Pareto optimal thresholds that lead to large conference qualities tend to have the following pattern: review papers strictly in the first two rounds and more leniently after that.
\end{itemize}

% \begin{figure}[htb]
%      \centering
%      \begin{subfigure}[b]{0.5\textwidth}
%          \centering
%          \includegraphics[width=\textwidth]{Plots/Review_following_m_3.png}
%          \captionsetup{size=}
%      \end{subfigure}
%      \hfill
%      \caption{The Pareto optimal points of the quality-burden tradeoff for three types of policy when authors are noisy in the \emph{$\boldsymbol{(\alpha = 0.75, \beta=0.75, m=3, V = 5, \TD=0.5, T=5)}$-binary model. The Pareto optimal points of the memoryless policy are also shown for comparison; the memoryless policy is approximated by setting a large $\boldsymbol{T=50}$.}\label{fig:memory_policies}}
% \end{figure}

\begin{figure}
     \FIGURE
     {\includegraphics[width=0.65\textwidth]{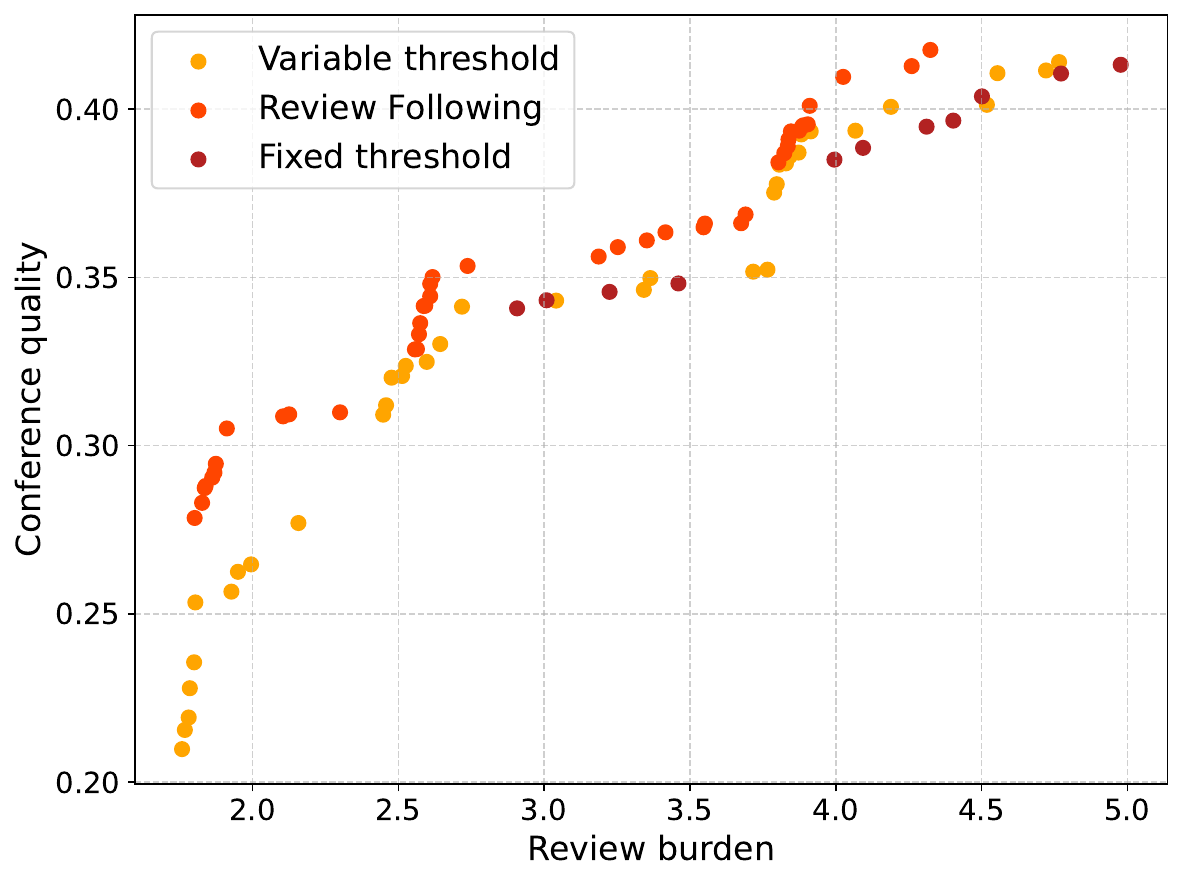}}
     {The Pareto Optimal Points of the QB-tradeoff With Noisy Authors in the Binary Model. \label{fig:memory_policies}}
     {The example is in the \emph{$(\alpha = 0.75, \beta=0.75, m=3, \TD=0.7, T=5)$-binary model}. }
\end{figure}

We provide some intuition for both results. Papers that are still being submitted after multiple rejections tend to be high-quality papers whose authors received strong positive signals. Variable-threshold policies can help reduce the review burden because the conference can induce a desired self-selection (where only the good papers are submitted) while relaxing the acceptance thresholds for papers that are resubmitted for a large number of rounds. This implies that variable-threshold policies can accept good papers within fewer rounds, which reduces the review burden.
Furthermore, in order to induce a desired self-selection, the conference should be strict for the first few rounds; this will cause more authors of low-quality papers to refrain from submitting. This leads to the pattern that the Pareto optimal threshold sequence tends to be strict at the beginning and more lenient later on.

\begin{remark}[Discussions and limitations of the binary model]
    We note that due to computational concerns, our conclusion here is largely based on the ABM experiments in the simplest binary model. Given that there seems to be general resistance to the idea of review-following (perhaps for fear that a bad, but erroneous, review may bring bias to the following reviews and doom a paper's chances for a long time), even if these modest gains generalize, this idea is unlikely to be a game changer.
\end{remark}